\documentclass[a4paper,11pt]{article}
\usepackage{fullpage} % can remove later if want document to be longer

\usepackage{amsmath}
\usepackage{amssymb}
\usepackage{amsthm}
\usepackage{graphicx}
\usepackage{enumerate}
\usepackage{microtype}
\usepackage{mleftright}
\usepackage{mdframed}
\usepackage{authblk}

\graphicspath{ {./figures/} }

\usepackage{algorithm}
\usepackage[noend]{algpseudocode}

\newtheorem{theorem}{Theorem}[section]

\newtheorem{lemma}[theorem]{Lemma}
\newtheorem{corollary}[theorem]{Corollary}

\newtheorem{question}{Open Question}

\theoremstyle{definition}
\newtheorem{definition}[theorem]{Definition}

\DeclareMathOperator{\Div}{Div}
\DeclareMathOperator{\Opt}{Opt}
\DeclareMathOperator{\Spl}{Spl}

\begin{document}

\title{Optimal sets of questions
for Twenty Questions}
\author[1]{Yuval Filmus}
\author[1]{Idan Mehalel}
\affil[1]{The Henry and Marilyn Taub Faculty of Computer Science, Technion, Israel. The research was funded by ISF grant 1337/16.}

\maketitle
\global\long\def\ser#1{\boldsymbol{#1}}%

\begin{abstract}
In the distributional Twenty Questions game, Bob chooses a number $x$ from $1$ to $n$ according to a distribution $\mu$, and Alice (who knows $\mu$) attempts to identify $x$ using Yes/No questions, which Bob answers truthfully. Her goal is to minimize the expected number of questions.

The optimal strategy for the Twenty Questions game corresponds to a Huffman code for $\mu$, yet this strategy could potentially uses  all $2^n$ possible questions. Dagan et al.\ constructed a set of $1.25^{n+o(n)}$ questions which suffice to construct an optimal strategy for \emph{all} $\mu$, and showed that this number is optimal (up to sub-exponential factors) for infinitely many~$n$.

We determine the optimal size of such a set of questions for \emph{all} $n$ (up to sub-exponential factors), answering an open question of Dagan et al. In addition, we generalize the results of Dagan et al.\ to the $d$-ary setting, obtaining similar results with $1.25$ replaced by $1 + (d-1)/d^{d/(d-1)}$.
\end{abstract}

\section{Introduction}

The distributional Twenty Questions game is a cooperative game between two players, Alice and Bob. Bob picks an object in $X_n = \{x_1,\ldots,x_n\}$ according to a distribution $\mu$ known to both players, and Alice determines the object by asking Yes/No questions, to which Bob answers truthfully. Alice's goal is to minimize the expected number of questions she asks.

This game is often related to information theory (see \cite{CT}, for example) as an interpretation of Shannon's entropy~\cite{shannon1948mathematical}. Moreover, it is the prototypical example of a combinatorial search game~\cite{Katona,AhlswedeWegener,ACD}. It is also a model of \emph{combinatorial group testing}~\cite{dorfman1943detection}, and can be interpreted as a learning task in the \emph{interactive learning} model~\cite{cohn1994improving}.

In this game, Alice's strategy corresponds to a prefix code: the code of $x \in X_n$ is the list of Bob's answers to all questions asked by Alice. Alice's optimal strategy therefore corresponds to a \emph{minimum redundancy code} for $\mu$. Huffman~\cite{Huffman} (and, independently, Zimmerman~\cite{Zimmerman}) showed how to construct such a strategy efficiently. However, the strategy produced by Huffman's algorithm could use arbitrary questions. We ask:

\begin{mdframed}
What is the smallest set of questions that allows Alice to construct an optimal strategy for every distribution $\mu$?
\end{mdframed}

We call such a set of questions an \emph{optimal} set of questions, and denote the minimum cardinality of an optimal set of questions for $X_n$ by $q(n)$. We stress that the same set of questions must be used for \emph{all} $\mu$.

Surprisingly, it is possible to improve on the trivial set of all $2^n$ questions \emph{exponentially}: Dagan et al.~\cite{dagan2017twenty,DFGM2019} showed that $q(n) \leq 1.25^{n+o(n)}$, and furthermore, $q(n) \geq 1.25^{n-o(n)}$ for infinitely many $n$ (specifically, $n$ of the form $1.25 \cdot 2^k$). Thus $1.25$ is the smallest constant $C$ such that $q(n) \leq C^{n+o(n)}$ for all $n$.

The fact that the lower bound $q(n) \geq 1.25^{n-o(n)}$ holds only for some $n$ suggests that the upper bound $1.25^{n+o(n)}$ can be improved for other $n$. This is what our first main result shows:

\begin{theorem} \label{thm:main-G}
There exists a function $G\colon [1,2) \to \mathbb{R}$ such that for $\beta \in [1,2)$,
\[
 q(n) = 2^{-G(\beta) n \pm o(n)} \text{ for all $n$ of the form } n = \beta \cdot 2^k.
\]
Furthermore, if $\beta \neq 1.25$ then
\[
 2^{-G(\beta)} < 1.25.
\]
\end{theorem}

This confirms a conjecture of Dagan et al. The exact formula for $G(\beta)$ appears in Theorem~\ref{theo:main}.

\paragraph{Optimal sets of questions and fibers} The proof of Theorem~\ref{thm:main-G} relies on a result of Dagan et al.

\begin{lemma} \label{lem:dyadic-hitter-intro}
A set of questions $\mathcal{Q}$ is optimal if for every dyadic distribution $\mu$ on $X_n$ (that is, a distribution in which the probability of each elements is $2^{-k}$ for some $k \in \mathbb{N}_+$), there is a set $Q \in \mathcal{Q}$ of probability exactly $1/2$.

Equivalently, $\mathcal{Q}$ is optimal if it hits (intersects with) $\Spl(\mu)$ for all dyadic $\mu$, where
\[
 \Spl(\mu) = \{ A : \mu(A) = 1/2 \}.
\]
\end{lemma}

To prove this result, Dagan et al.\ first show that a set of questions is optimal iff there is an optimal strategy for every dyadic distribution. Roughly speaking, given an arbitrary distribution $\nu$, we construct a Huffman code $C$ for $\nu$ and convert it to a distribution $\mu(x_i) = 2^{-|C(x_i)|}$. An optimal strategy for $\mu$ turns out to be an optimal strategy for $\nu$.

Dagan et al.\ then show that an optimal strategy for a dyadic distribution must split the probability evenly at every step. Distributions encountered in this way could have elements whose probability is zero, but by choosing an element of minimal positive probability $2^{-k}$ and ``splitting'' it into elements of probability $2^{-k+1},2^{-k+2},\ldots,2^{-k+t},2^{-k+t}$ (where $t$ is the number of zero probability elements), we can reduce to the case of distributions of full support.

It is easy to see that $\Spl(\mu)$ is an antichain: if $A \subsetneq B$ then $\mu(A) < \mu(B)$. It is less obvious that $\Spl(\mu)$ is a maximal antichain, as observed by Dagan et al. Indeed, given any set $A$, if $\mu(A) > 1/2$ and we arrange the elements of $A$ in nonincreasing order of probability, then some prefix has probability exactly $1/2$, and so some subset of $A$ belongs to $\Spl(\mu)$; and if $\mu(A) < 1/2$, then we can apply the same argument on $\overline{A}$ to find a superset of $A$ in $\Spl(\mu)$.

These observations connect optimal sets of questions with another combinatorial object: \emph{fibers}, defined by Lonc and Rival~\cite{LoncRival}. Given any poset, a \emph{fiber} is a hitting set for the family of all maximal antichains. Any fiber of the lattice $2^{X_n}$ is thus an optimal set of questions.

Duffus, Sands and Winkler~\cite{DuffusSandsWinkler} showed that every fiber of $2^{X_n}$ contains $\Omega(1.25^n)$ elements. To show this, they considered maximal antichains of the following form, for a parameter $a$:
\[
 S(B) = \{ A, \overline{A} : A \subset B, |A| = a \}, \text{ where } |B| = 2a-1.
\]
It is easy to check that these are maximal antichains. There are $\binom{n}{2a-1}$ such maximal antichains, and each set in a fiber can handle at most $\binom{n-a}{a-1}$ of them, giving a lower bound of
\[
 \frac{\binom{n}{2a-1}}{\binom{n-a}{a-1}} %\approx 2^{nh(2a/n) - (n-a)h(a/(n-a)}
 \approx 2^{n(h(2\theta) - (1-\theta) h(\theta/(1-\theta))}, \quad \theta = \frac{a}{n}.
\]
Here $h(p) = p\log_2(1/p) + (1-p) \log_2(1/(1-p))$ is the binary entropy. This expression is maximized at $\theta = 1/5$, giving a lower bound of roughly $1.25^n$.

Dagan et al.\ used the exact same argument to prove their lower bound of $1.25^{n-o(n)}$ for optimal sets of questions. To this end, they needed to realize $S(B)$ as a set of the form $\Spl(\mu)$. The idea is to give all elements of $B$ a probability of $1/2a$, and the remaining elements (the ``tail'') probabilities $1/4a,1/8a,\ldots,1/2^{n-2a+1}a,1/2^{n-2a+1}a$; any set of measure $1/2$ must either contain the tail and $a-1$ elements of $B$, or must consist of $a$ elements of $B$.

For this construction to work, we need $1/2a$ to be a negative power of $2$, that is, we need $a$ to be a power of $2$. Since $a = n/5$, this works as long as $n$ is of the form $1.25 \cdot 2^k$, or at least close to such a number. When $n$ is close to $\beta \cdot 2^k$ for $\beta \in [1,2)$ other than $1.25$, the sets $S(B)$ are not realizable in the form $\Spl(\mu)$.

In order to prove the lower bound part of Theorem~\ref{thm:main-G}, we identify, for each value of $\beta$, a more general collection of hard-to-hit maximal antichains which are realizable as $\Spl(\mu)$ for $n = \beta \cdot 2^k$. Instead of having a single set $B$ of elements of equal probability together with a ``tail'', we allow several such sets $B_1,B_2,B_3,\ldots$, where the probability of elements in $B_t$ is $1/2^ta$. This results in an expression for $G(\beta)$ which describes a game between Asker and Builder, in which Builder picks the proportions of the sets $B_t$, and Asker picks the type of questions which are best suited to handle the sets $S(B_1,B_2,B_3,\ldots)$; we leave the details to Section~\ref{sec:The-exact-formula}.

Dagan et al.\ showed that the bound $1.25^n$ is tight, by constructing an optimal set of questions of this size for every~$n$. The bound is not tight for fibers of $2^{X_n}$: \L{}uczak improved the lower bound $\Omega(1.25^n)$ to $\Omega(2^{n/3}) = \Omega(1.2599^n)$, as described in Duffus and Sands~\cite{DuffusSands}. Lonc and Rival \cite{LoncRival} conjectured that the optimal size of a fiber of $2^{X_n}$ is $\Theta(2^{n/2})$, realized by the collection of all sets comparable with $\{x_1,\ldots,x_{\lfloor n/2 \rfloor}\}$. This is also the best known explicit construction of an optimal set of questions.

Instead of designing an optimal set of questions explicitly, Dagan et al.\ show that if we pick roughly $1.25^n$ random questions of each size, then with high probability we get an optimal set of questions. A similar approach works for proving the upper bound in Theorem~\ref{thm:main-G}, though the calculations are more intricate.

Much of the difficulty in proving Theorem~\ref{thm:main-G} comes from the fact that $G(\beta)$ describes an idealized game between Asker and Builder which only makes sense in the limit $k \to \infty$, which we need to connect with the corresponding game for a fixed value of $k$. This difficulty doesn't come up in Dagan et al.\ since in their case, the optimal construction only has a single set $B_1$.

\paragraph{Computing $G(\beta)$} 
Unfortunately, the formula for $G(\beta)$ involves non-convex optimization over infinitely many variables, and for this reason we are unable to compute $G(\beta)$ beyond the already known value $G(1.25) = -\log_2 1.25$.

Nevertheless, our techniques allow us to improve the bound on $\max_\beta G(\beta)$ given by Dagan et al. Stated in terms of $q(n)$, Dagan et al.\ showed that $q(n) \geq 1.232^{n-o(n)}$ for every $n$, and we improve this to $q(n) \geq 1.236^{n-o(n)}$ for every $n$.

\paragraph{$d$-ary questions} Our second main result concerns \emph{$d$-ary questions}. What happens if Alice asks Bob $d$-ary questions, that is, questions with $d$ possible answers?
The optimal strategy in this case corresponds to a $d$-ary Huffman code, a setting already considered in Huffman's original paper.

We are able to generalize the main result of Dagan et al.\ to this setting:

\begin{theorem} \label{thm:main-d}
Let $q^{(d)}(n)$ be the minimum cardinality of an optimal set of $d$-ary questions.

For all $n$,
\[
 q^{(d)}(n) \leq \left(1 + \frac{d-1}{d^{\frac{d}{d-1}}}\right)^{n+o(n)}.
\]
Furthermore, this inequality is tight (up to sub-exponential factors) for infinitely many values of $n$.
\end{theorem}

This result holds not only for constant $d$, but also uniformly for all $d = o(n/\log^2 n)$. The techniques closely follow the ideas of Dagan et al., as outlined above in the case $d = 2$.

Theorem~\ref{thm:main-d} shows that the magic constant $1.25$ appearing in~\cite{DFGM2019} generalizes to
\[
 1 + \frac{d-1}{d^{\frac{d}{d-1}}} = 2 - \Theta\left(\frac{\log d}{d}\right).
\]
for arbitrary $d$.

We leave a combination of Theorem~\ref{thm:main-G} and Theorem~\ref{thm:main-d} to future work.

\paragraph{Paper organization} After brief preliminaries (Section~\ref{sec:Preliminaries}), we prove the first half of Theorem~\ref{thm:main-G} in Sections~\ref{sec:The-exact-formula}--\ref{sec:Approximating}. We prove the ``furthermore'' part of Theorem~\ref{thm:main-G} and reprove some results of~\cite{DFGM2019} using our framework in Section~\ref{sec:Applications}, in which we also derive the improved lower bound $1.236^{n-o(n)}$ (Theorem~\ref{thm:1.236}). Our results on $d$-ary questions appear in Section~\ref{sec:d-ary}. Section~\ref{sec:Open-questions} closes the paper with some open questions.

\section{Preliminaries \label{sec:Preliminaries}}

% We first present some notations and definitions used or suggested
% in \cite{DFGM2019}. Then, we present a reduction stated and proved
% in \cite{DFGM2019}, which we rely on throughout the document.\par
Given a distribution $\pi$ over $X_n=\mleft\{x_1,\dots,x_n\mright\}$, denote $\pi_{i}:=\pi\mleft(x_i\mright)$. For any set $S\subseteq X_n$
we denote the sum $\sum_{i\in S}\pi_{i}$ with $\pi\mleft(S\mright)$.

\subsection{Decision trees}

We represent a strategy to reveal a secret element $x\in X_n$
as a \emph{decision tree}. A decision tree is a binary tree $T=\mleft(V,\,E\mright)$
such that every internal node $v\in V$ is labeled with a query $Q\subseteq X_n$,
every leaf $l\in V$ is labeled with an object $x_i\in X_n$,
and every edge $e\in E$ is labeled with ``Yes'' or ``No''. Moreover,
if $v$ is an internal node that is labeled with the query $Q$, and
$x$ is the secret element, then $v$ has two outgoing edges: one
is labeled with ``Yes'' (representing the decision ``$x\in Q$'')
and the other with ``No'' (representing the decision ``$x\notin Q$'').\par
Given a set of queries $\mathcal{Q}\subseteq2^{X_n}$ (which
is called the set of \emph{allowed questions}), we say that $T$ is
a \emph{decision tree using $\mathcal{Q}$ }if for any internal node
$v\in V$, the query $Q$ that $v$ is labeled with satisfies $Q\in\mathcal{Q}$. \par
Given a distribution $\pi$ over $X_n$, we say that a
decision tree $T$ is \emph{valid }for $\pi$ if for any object $x\in\operatorname{supp}\mleft(\pi\mright)$
there is a path in $T$ that begins in the root and ends in a leaf
that is labeled with $x$. The decision trees we will consider are
only those in which each object $x\in X_n$ labels at most
one leaf.\par 
If there is a path from the root to $x\in X_n$, we say
that the number of its edges is the depth of $x$, and denote this
number with $T\mleft(x\mright)$. If $T$ is valid for $\pi$, the \emph{cost} of $T$ on $\pi$ is $\sum_{i=1}^{n}\pi_{i}T\mleft(x_{i}\mright)$.

%Denote the expected number of queries used by a decision tree using $\mathcal{Q}=2^{X_n}$ for $\pi$ with $\Opt\mleft(\pi\mright)$.
%A set of allowed questions $\mathcal{Q}$

\subsection{Optimal sets of questions}

For a distribution $\pi$, let $\Opt(\pi)$ be the minimum cost of a decision tree for $\pi$.

A set $\mathcal{Q} \subseteq 2^{X_n}$ of queries is \emph{optimal} if for every distribution $\pi$, there is a decision tree using $\mathcal{Q}$ whose cost is $\Opt(\pi)$. 

We denote the minimum cardinality of an optimal set of queries over $X_n$ by $q(n)$. A major goal of this paper is to estimate $q(n)$ for all values of $n$. We do this using the concept of \emph{maximum relative density}, borrowed from~\cite{DFGM2019}.

\begin{definition}
A distribution $\mu$ is \emph{dyadic} if for all $i$, $\mu_i = 2^{-d}$ for some $d \in \mathbb{N}$ or $\mu_i=0$.

If $\mu$ is a non-constant dyadic distribution, then a set $A \subseteq X_n$ \emph{splits} $\mu$ if $\mu(A) = 1/2$. We denote the collection of all sets splitting $\mu$ by $\Spl(\mu)$, and the collection of all sets of size $i$ splitting $\mu$ by $\Spl(\mu)_i$.

The \emph{$i$'th relative density} of $\Spl(\mu)$ is
\[
 \rho_i(\Spl(\mu)) = \frac{|\Spl(\mu)_i|}{\binom{n}{i}}.
\]
The \emph{maximum relative density} of $\Spl(\mu)$ is
\[
 \rho(\Spl(\mu)) = \max_{i \in \{1,\ldots,n-1\}} \rho_i(\Spl(\mu)).
\]
\end{definition}

The following result reduces the calculation of $q(n)$, up to polynomial factors, to the calculation of the quantity
\[
 \rho_{\min}(n) = \min_{\mu} \rho(\Spl(\mu)),
\]
where the minimum is taken over all non-constant and full-support dyadic distributions.

\begin{theorem}[{\cite[Theorem 3.3.1 and Lemma 3.2.6]{DFGM2019}}]\label{theo:reduc}
It holds that
\[
\frac{1}{\rho_{\min}\mleft(n\mright)}\leq q\mleft(n\mright)\leq n^{2}\log n\frac{1}{\rho_{\min}\mleft(n\mright)}.
\]
\end{theorem}

Hence, from now on, we will consider the problem of finding a formula
for $\rho_{\min}\mleft(n\mright)$ (up to sub-exponential factors) instead of a formula for $q\mleft(n\mright)$. 

\subsection{Tails}

The \emph{tail} of a dyadic distribution $\mu$ over $X_n$
is the largest set $T\subseteq X_n$ which satisfies, for some
$a\in\mathbb{N}$: 
\begin{itemize}
\item The probabilities of the elements in $T$ are $2^{-a-1},2^{-a-2},\dots,2^{-a-\mleft(\mleft|T\mright|-1\mright)},2^{-a-\mleft(\mleft|T\mright|-1\mright)}$.
\item Any element $x_i\in X_n\backslash T$ has probability at least
$2^{-a}$. 
\end{itemize}

If $\mu$ is a dyadic distribution, then \cite[Lemma 3.2.5]{DFGM2019} shows that each set in $\Spl(\mu)$ either contains $T$ or is disjoint from $T$.

\subsection{Entropy}

The \emph{entropy} of a distribution $\pi$ is
\[
H\mleft(\pi\mright):=\sum_{i=1}^n \pi_i \log \frac{1}{\pi_i}.
\]
For $n=2$, define the \emph{binary entropy function}:
\[
h\mleft(\pi_1\mright):= \pi_1 \log \frac{1}{\pi_1}+\mleft(1-\pi_1\mright) \log \frac{1}{1-\pi_1}.
\]

We prove some simple bounds on the binary entropy function, which
will be useful in some of the proofs in this work. If $y-\epsilon\leq x\leq y+\epsilon$
for some $x,y,\epsilon$, denote $x=y\pm\epsilon$. 
\begin{lemma}
\label{lem:entropy_bounds}For any $0\leq x,\epsilon_{1},\epsilon_{2}\leq1$
such that $\epsilon_{2}\leq x\leq1-\epsilon_{1}$ it holds that 
\begin{gather*}
h\mleft(x+\epsilon_{1}\mright)=h\mleft(x\mright)\pm h\mleft(\epsilon_{1}\mright),\\
h\mleft(x-\epsilon_{2}\mright)=h\mleft(x\mright)\pm h\mleft(\epsilon_{2}\mright).
\end{gather*}
\end{lemma}
\begin{proof}
Let $0\leq x,\epsilon_{1}\leq1$ such that $x+\epsilon_{1}\leq1$.
Since $h$ is concave and $h\mleft(0\mright)=0$ it is known that $h$
is sub-additive, that is $h\mleft(x+\epsilon_{1}\mright)\leq h\mleft(x\mright)+h\mleft(\epsilon_{1}\mright)$.
Using that inequality together with the fact that $h$ is symmetric,
we have: 
\[
h\mleft(x\mright)=h\mleft(1-\mleft(x+\epsilon_{1}\mright)+\epsilon_{1}\mright)\leq h\mleft(1-\mleft(x+\epsilon_{1}\mright)\mright)+h\mleft(\epsilon_{1}\mright)=h\mleft(x+\epsilon_{1}\mright)+h\mleft(\epsilon_{1}\mright).
\]
 The second inequality is proved similarly.
\end{proof}

Throughout this paper, we also use the fact that $h\mleft(x\mright)$ is increasing for $x<1/2$.

\section{An exact (and almost exact) formula for $\rho_{\min}\mleft(n\mright)$
\label{sec:The-exact-formula}}

Our goal in this section and the next is to find a formula for $\rho_{\min}\mleft(n\mright)$
up to sub-exponential factors. We use the expression 
\[
\rho_{\min}\mleft(n\mright)=\min_{\mu}\max_{d\in\mleft[n\mright]}\rho_{d}\mleft(\Spl\mleft(\mu\mright)\mright)
\]
as our starting point. %For simplicity and without loss of generality, we consider in this section only non-constant and full support dyadic distributions.
We want to present
$\rho_{\min}\mleft(n\mright)$ in a more ``direct'' or ``numeric''
fashion, rather than through a choice of a non-constant dyadic distribution $\mu$.  Denote $n=\beta\cdot2^{k}$ where
$\beta\in\mleft[1,2\mright)$ and $k\in\mathbb{N}$. From now on and
throughout this paper, when we refer to $\beta$ and $k$ that is
always their meaning unless specified otherwise. Let $\mleft[0,1\mright]^{\mathbb{N}}$ be the
set of all sequences $\mleft\{ c_{i}\mright\} _{i=0}^{\infty}$ where
$0\leq c_{i}\leq1$ for any $i$ and denote a sequence $\mleft\{ c_{i}\mright\} _{i=0}^{\infty}\in\mleft[0,1\mright]^{\mathbb{N}}$
with $\ser c$ (the notation is the elements' letter in bold). In
order to describe a non-constant and full support dyadic distribution $\mu$ in this language, we
can determine the following sufficient and necessary values: 
\begin{itemize}
\item $b\in\mathbb{N}$, where the highest probability in $\mu$ is determined
to be $\mu_{1}=2^{b-k}$. 
\item An ``amount sequence'' $\ser c$ which describes how many elements
$\mu$ will have of each probability. In order to obtain a valid dyadic
distribution the following must hold: 
\begin{gather*}
\sum_{i=0}^{\infty}c_{i}/2^{i}=\frac{1}{\beta\cdot2^{b}},\\
\sum_{i=0}^{\infty}c_{i}\leq1,\\
\forall i\colon c_{i}n\in\mathbb{N}.
\end{gather*}
\end{itemize}
Those values indeed determine $\mu$ uniquely: for any $i$, $c_{i}n$
elements have probability $\mu_{1}/2^{i}$ (assume that $c_{0}>0$).
Actually, in order to describe $\mu$ precisely, we also have to say
exactly \textbf{which elements}\textbf{\emph{ }}are\textbf{ }the $c_{i}n$
elements having probability $\mu_{1}/2^{i}$. However, since we are
only interested in the identity of a distribution $\mu$ which minimizes
$\rho\mleft(\Spl\mleft(\mu\mright)\mright)$, the identity of the $c_{i}n$
elements having probability $\mu_{1}/2^{i}$ does not matter --- what
matters is only their quantity. If $t$ is the highest index such
that $c_{t}>0$, then one element with probability $\mu_{1}/2^{t}$ is ``turned''
into a tail with total probability $\mu_{1}/2^{t}$, such that we get $n$
elements in total. The first constraint assures that the probabilities
in $\mu$ sum up to 1. The second constraint assures that there are
no more than $n$ non-tail elements, that is, exactly $n$ elements
in total. The third constraint assures that there is an integral number
of elements of each type. 

For the proof of our formula for $\rho_{\min}\mleft(n\mright)$ which
we will present soon, we want to distinguish between pairs $\mleft(\ser c,b\mright)\in\mleft[0,1\mright]^{\mathbb{N}}\times\mathbb{N}$
which satisfy all of those three constraints, and those which do not
necessarily satisfy the third ``integrality'' constraint. Hence,
denote the set of all pairs $\mleft(\ser c,b\mright)\in\mleft[0,1\mright]^{\mathbb{N}}\times\mathbb{N}$
satisfying the first two constraints, that is, $\sum_{i=0}^{\infty}2^{b-i}c_{i}\beta=1$
and $\sum_{i=0}^{\infty}c_{i}\leq1$ with $\mathcal{C}=\mathcal{C}\mleft(\beta\mright)$.
If a pair $\mleft(\ser c,b\mright)\in\mathcal{C}$ satisfies the third
constraint as well, we say that $\mleft(\ser c,b\mright)$ (or, simply
$\ser c$, when the identity of $b$ is clear) is \emph{$k$-feasible}.\footnote{Even though this constraint relates to $n$, we choose to relate $k$
instead of $n$ to conform with future notations which relate to $k$
as well. Our discussion will fix $\beta\in\mleft[1,2\mright)$, and
thus $n$ will be determined uniquely by $k$, so this is not a problem.} \par
Now we want to describe the choice of an integer $d$ for the maximization
part in $\rho_{\min}\mleft(n\mright)$. Due to \cite{DFGM2019}, we know
that each splitting set either contains all tail elements, or none
of them. For a sequence $\ser c\in\mleft[0,1\mright]^{\mathbb{N}}$,
denote by $t$ the last index such that $c_{t}>0$. If there is no
such index, $t=\infty$. Fix $\mleft(\ser c,b\mright)\in\mathcal{C}$
which is $k$-feasible, that describes a dyadic distribution $\mu$
(in that case, $t<\infty$). In order to describe the set $\Spl\mleft(\mu\mright)_{d}$
for some $d\in\mleft[n\mright]$ (recall that those are all splitting
sets of $\mu$ of size $d$) we can consider the following sets of
sequences in $\mleft[0,1\mright]^{\mathbb{N}}$: 
\begin{itemize}
\item A set $S_{d}$ of all sequences $\ser{\alpha}\in\mleft[0,1\mright]^{\mathbb{N}}$
describing sets in $\Spl\mleft(\mu\mright)_{d}$ which do not contain
the tail elements. Those sequences satisfy: 
\begin{gather*}
\sum_{i=0}^{t}\alpha_{i}c_{i}/2^{i}=\frac{1}{\beta\cdot2^{b+1}},\\
\sum_{i=0}^{t}\alpha_{i}c_{i}n=d,\\
\forall i\colon \alpha_{i}c_{i}n\in\mathbb{N},\\
\alpha_{t}<1.
\end{gather*}
 
\item A set $T_{d}$ of all sequences $\ser{\alpha}\in\mleft[0,1\mright]^{\mathbb{N}}$
describing sets in $\Spl\mleft(\mu\mright)_{d}$ which contain the tail
elements. Those sequences satisfy: 
\begin{gather*}
\sum_{i=0}^{t}\alpha_{i}c_{i}/2^{i}=\frac{1}{\beta\cdot2^{b+1}},\\
\sum_{i=0}^{t}\alpha_{i}c_{i}n + \mleft(1-\sum_{i=0}^{t}c_{i}\mright)n-1=d,\\
\forall i\colon \alpha_{i}c_{i}n\in\mathbb{N},\\
\alpha_{t}>0.
\end{gather*}
\end{itemize}
The constraints of $S_{d},T_d$ indeed describe splitting sets of
size $d$: The first constraint assures
that the set described by a sequence $\ser{\alpha}$ is a splitting
set. The second constraint assures that its size is $d$. The third
constraint assures that each probability type appears in the set an
integral number of times. The last constraint on $S_d$ or $T_d$ assures that the tail
elements may not be or may be a part of the splitting set, respectively. Note that a sequence
$\ser{\alpha}\in\mleft[0,1\mright]^{\mathbb{N}}$ satisfying those constraints
does not determine \textbf{which elements }are exactly the elements
chosen to the splitting set. We soon handle that, since here, in contrast
to the choice of $\mu$, the identity of the elements chosen having
given probability matters, since any choice of different elements
defines a different splitting set. This discussion implies that we
can write $\rho_{\min}\mleft(n\mright)$ in the following way, which
will be convenient for our purposes: 
\[
\rho_{\min}\mleft(n\mright)  =\min_{\substack{
\mleft(\ser c,b\mright)\in\mathcal{C}\colon\\
\ser c\text{ is \ensuremath{k}-feasible}
}}\max_{\substack{
d\in\mleft[n\mright]\colon\\
S_{d}\cup T_{d}\neq\emptyset
}}\sum_{\ser{\alpha}\in S_{d}} \frac{\binom{c_t n -1}{\alpha_t c_t n}\prod_{i=0}^{t-1}\binom{c_{i}n}{\alpha_{i}c_{i}n}}{\binom{n}{d}} + \sum_{\ser{\alpha}\in T_{d}} \frac{\binom{c_t n -1}{\alpha_t c_t n -1}\prod_{i=0}^{t-1}\binom{c_{i}n}{\alpha_{i}c_{i}n}}{\binom{n}{d}}.
\]
 For $0\leq i< t$, each binomial coefficient $\binom{c_{i}n}{\alpha_{i}c_{i}n}$ is
the number of possibilities to choose $\alpha_{i}c_{i}n$ elements
of probability $2^{-k+b-i}$ to the splitting set. For the index $t$, we use the expressions $\binom{c_t n -1}{\alpha_t c_t n}$ and $\binom{c_t n -1}{\alpha_t c_t n -1}$ because we must use or not use the tail elements, depends on whether $\alpha$ is in $S_d$ or $T_d$. 

Since our goal is to find a formula for $\rho_{\min}\mleft(n\mright)$
up to sub-exponential factors, we can simplify the expression a bit,
and ignore the sequences in $T_{d}$. Define 
\[
\rho_{\min}^{*}\mleft(n\mright)=\min_{\substack{
\mleft(\ser c,b\mright)\in\mathcal{C}:\\
\ser c\text{ is \ensuremath{k}-feasible}
}}\max_{\substack{
d\in\mleft[n\mright]:\\
S_{d}\neq\emptyset
}}\sum_{\ser{\alpha}\in S_{d}}\frac{\prod_{i=0}^{t}{\binom{c_{i}n} {\alpha_{i}c_{i}n}}}{\binom{n}{d}}.
\]
 The idea is that any splitting set $S$ described by a sequence $\alpha\in T_{d}$,
has a matching splitting set $\overline{S}$ (the complement set of
$S$) described by a sequence $\alpha'\in S_{n-d}$ such that 
\[
\frac{\binom{c_t n -1}{\alpha_t c_t n}\prod_{i=0}^{t-1}\binom{c_{i}n}{\alpha_{i}c_{i}n}}{\binom{n}{d}}=\frac{\binom{c_t n -1}{\alpha'_t c_t n-1}\prod_{i=0}^{t-1}\binom{c_{i}n}{\alpha'_{i}c_{i}n}}{\binom{n}{n-d}}.
\]
 Thus by considering only sequences in $S_{d}$, we get an approximation
for $\rho_{\min}\mleft(n\mright)$. Here is a detailed proof for that,
for the interested reader:
\begin{lemma}
\label{lem:MRD_star}It holds that 
\[
\rho_{\min}\mleft(n\mright)/2\leq\rho_{\min}^{*}\mleft(n\mright)\leq n\cdot \rho_{\min}\mleft(n\mright).
\]
\end{lemma}
\begin{proof}
Let $n=\beta\cdot2^{k}$. Fix $\mleft(\ser c,b\mright)\in\mathcal{C}$
which is $k$-feasible and $d\in\mleft[n\mright]$. To handle the case
$S_{d}=\emptyset$ or $T_{d}=\emptyset$, define 
\[
f_{S}\mleft(d\mright)=\begin{cases}
\sum_{\ser{\alpha}\in S_{d}} \frac{\binom{c_t n -1}{\alpha_t c_t n}\prod_{i=0}^{t-1}\binom{c_{i}n}{\alpha_{i}c_{i}n}}{\binom{n}{d}} & S_{d}\neq\emptyset,\\
0 & S_{d}=\emptyset,
\end{cases}
\]
 and 
\[
f_{T}\mleft(d\mright)=\begin{cases}
\sum_{\ser{\alpha}\in T_{d}} \frac{\binom{c_t n -1}{\alpha_t c_t n -1}\prod_{i=0}^{t-1}\binom{c_{i}n}{\alpha_{i}c_{i}n}}{\binom{n}{d}} & T_{d}\neq\emptyset,\\
0 & T_{d}=\emptyset.
\end{cases}
\]
In this language, we can write: 
\[
\rho_{\min}\mleft(n\mright)=\min_{\substack{
\mleft(\ser c,b\mright)\in\mathcal{C}\colon\\
\ser c\text{ is \ensuremath{k}-feasible}
}}\max_{\substack{
d\in\mleft[n\mright]}}f_{S}\mleft(d\mright)+f_{T}\mleft(d\mright).
\]
Define:
\[
\rho_{\min}^{**}\mleft(n\mright)=\min_{\substack{
\mleft(\ser c,b\mright)\in\mathcal{C}\colon\\
\ser c\text{ is \ensuremath{k}-feasible}
}}\max_{\substack{
d\in\mleft[n\mright]}}f_{S}\mleft(d\mright).
\]
 If $f_{S}\mleft(d\mright)\geq f_{T}\mleft(d\mright)$, then 
\[
f_{S}\mleft(d\mright)+f_{T}\mleft(d\mright)\leq2\cdot f_{S}\mleft(d\mright).
\]
Else, assume $f_{S}\mleft(d\mright)<f_{T}\mleft(d\mright)$. Since for
any $\ser{\alpha}\in T_{d}$ we have $\ser{\alpha'}\in S_{n-d}$ such
that 
\[
\frac{\binom{c_t n -1}{\alpha_t c_t n}\prod_{i=0}^{t-1}\binom{c_{i}n}{\alpha_{i}c_{i}n}}{\binom{n}{d}}=\frac{\binom{c_t n -1}{\alpha'_t c_t n-1}\prod_{i=0}^{t-1}\binom{c_{i}n}{\alpha'_{i}c_{i}n}}{\binom{n}{n-d}}
\]
 ($\alpha'_{i}=1-\alpha_{i}$ for any $i$), and the opposite holds
as well in a similar fashion, we have 
\[
f_{S}\mleft(n-d\mright)=f_{T}\mleft(d\mright)>f_{S}\mleft(d\mright)=f_{T}\mleft(n-d\mright)
\]
 and thus
\begin{align*}
f_{S}\mleft(d\mright)+f_{T}\mleft(d\mright)=f_{S}\mleft(n-d\mright)+f_{T}\mleft(n-d\mright) & \leq2\cdot f_{S}\mleft(n-d\mright).
\end{align*}
Hence, we can always choose $d'\in\mleft[n\mright]$ such that 
\[
\rho_{\min}\mleft(n\mright)=\min_{\substack{
\mleft(\ser c,b\mright)\in\mathcal{C}\colon\\
\ser c\text{ is \ensuremath{k}-feasible}}}f_{S}\mleft(d'\mright)+f_{T}\mleft(d'\mright)
\]
 and $f_{S}\mleft(d'\mright)\geq f_{T}\mleft(d'\mright)$. Hence:
\begin{align*}
\rho_{\min}\mleft(n\mright) & \leq2\cdot\min_{\substack{
\mleft(\ser c,b\mright)\in\mathcal{C}\colon\\
\ser c\text{ is \ensuremath{k}-feasible}
}}f_{S}\mleft(d'\mright)\\
 & \leq2\cdot\min_{\substack{
\mleft(\ser c,b\mright)\in\mathcal{C}\colon\\
\ser c\text{ is \ensuremath{k}-feasible}
}}\max_{d\in\mleft[n\mright]}f_{S}\mleft(d\mright)=2\cdot\rho_{\min}^{**}\mleft(n\mright).
\end{align*}
Now, note that:
\[
\binom{x}{y}/x\leq \binom{x-1}{y}\leq \binom{x}{y}
\]
(the left inequality holds as long as $x>y$), and hence $\rho_{\min}\mleft(n\mright)/2\leq \rho_{\min}^{**}\mleft(n\mright)\leq \rho_{\min}^{*}\mleft(n\mright)$ and moreover $\rho_{\min}^{*}\mleft(n\mright)\leq n\cdot \rho_{\min}^{**}\mleft(n\mright)\leq n\cdot \rho_{\min}\mleft(n\mright)$, since $\alpha_t <1$.
Hence the lemma follows.
\end{proof}
Due to that approximation, it is enough to find a formula that estimates
$\rho_{\min}^{*}\mleft(n\mright)$ instead of $\rho_{\min}\mleft(n\mright)$,
up to sub-exponential factors. 

\section{Approximating $\rho_{\min}\mleft(n\mright)$ \label{sec:Approximating}}

In this section we prove our first main result, which is the following theorem:
\begin{theorem} \label{theo:main}
There is a function $G\colon\mleft[1,2\mright)\rightarrow\mathbb{R}$
such that $\rho_{\min}\mleft(n\mright)=2^{G\mleft(\beta\mright)n\pm o\mleft(n\mright)}$,
where $n=\beta\cdot2^{k}$, $k\in\mathbb{N}$ and $\beta\in\mleft[1,2\mright)$.
The function $G$ is given by the following formula: 
\[
G\mleft(\beta\mright)=\inf_{\substack{
b\in\mathbb{N},\ser c\in\mleft[0,1\mright]^{\mathbb{N}}\colon\\
\sum_{i=0}^{\infty}c_{i}/2^{i}=\frac{1}{\beta\cdot2^{b}}\\
\sum_{i=0}^{\infty}c_{i}\leq1
}}\max_{\substack{
\ser{\alpha}\in\mleft[0,1\mright]^{\mathbb{N}}\colon\\
\sum_{i=0}^{\infty}\alpha_{i}c_{i}/2^{i}=\frac{1}{\beta\cdot2^{b+1}}
}}\sum_{i=0}^{\infty}h\mleft(\alpha_{i}\mright)c_{i}-h\mleft(\sum_{i=0}^{\infty}\alpha_{i}c_{i}\mright).
\]
\end{theorem}

\begin{corollary}
Putting together Theorems \ref{theo:reduc} and \ref{theo:main}, it holds that
 $q\mleft(n\mright)=2^{-G\mleft(\beta\mright)n\pm o\mleft(n\mright)}$,
where $n=\beta\cdot2^{k}$, $k\in\mathbb{N}$ and $\beta\in\mleft[1,2\mright)$.
\end{corollary}
An immediate corollary is that the exponent base of $\rho_{\min}\mleft(n\mright)$ (and $q\mleft(n\mright)$) does not depend on $n$, but only on $\beta$. We assume that
$c_{0}>0$, since if $c_{0}=0$ then we can choose another $b$ and
construct an equivalent sequence with $c_{0}>0$. For the rest of
the section, fix $\beta\in\mleft[1,2\mright)$ and denote $P\mleft(\ser c,\ser{\alpha}\mright)=\sum_{i=0}^{\infty}h\mleft(\alpha_{i}\mright)c_{i}-h\mleft(\sum_{i=0}^{\infty}\alpha_{i}c_{i}\mright)$.
For a fixed $\mleft(\ser c,b\mright)\in\mathcal{C}$ (that is, $\mleft(\ser c,b\mright)\in\mleft[0,1\mright]^{\mathbb{N}}\times\mathbb{N}$
which satisfies also $\sum_{i=0}^{\infty}c_{i}/2^{i}=\frac{1}{\beta\cdot2^{b}}$
and $\sum_{i=0}^{\infty}c_{i}\leq1$), we denote by $\mathcal{A}\mleft(\ser c,b\mright)$
(or simply $\mathcal{A}$, from now on, assuming $\mleft(\ser c,b\mright)$
is fixed) the set of all sequences $\ser{\alpha}\in\mleft[0,1\mright]^{\mathbb{N}}$
which satisfy $\sum_{i=0}^{\infty}\alpha_{i}c_{i}/2^{i}=\frac{1}{\beta\cdot2^{b+1}}$
(the maximization constraint in $G\mleft(\beta\mright)$). In this language,
we can write
\[
G\mleft(\beta\mright)=\inf_{\mleft(\ser c,b\mright)\in\mathcal{C}}\max_{\ser{\alpha}\in\mathcal{A}}P\mleft(\ser c,\ser{\alpha}\mright).
\]

\subsection{$G$ is well-defined}

Before we prove our formula, we first show that $G$ is indeed well-defined
and finite: if we change the inner max to sup, then it is clear that
$G\mleft(\beta\mright)$ is well-defined and finite: it always holds
that $\sum_{i=0}^{\infty}c_{i}\leq1$, and thus $-1\leq P\mleft(\ser c,\ser{\alpha}\mright)\leq1$
for any $\mleft(\ser c,b\mright)\in\mathcal{C}$, $\ser{\alpha}\in\mathcal{A}$.
Moreover, for any $\mleft(\ser c,b\mright)\in\mathcal{C}$, $\mathcal{A}$
is non-empty: choosing the sequence $\ser{\alpha}$ to be $\alpha_{i}=1/2$
for any $i$ satisfies $\sum_{i=0}^{\infty}\alpha_{i}c_{i}/2^{i}=\frac{1}{\beta\cdot2^{b+1}}$
for any $\mleft(\ser c,b\mright)\in\mathcal{C}$, thus it always belongs
to $\mathcal{A}$. It is known that supremum/infimum values are defined
and finite for non-empty bounded sets, thus it remains to show that
the inner supremum is attained, and hence can be written as maximum.
Fix $\mleft(\ser c,b\mright)\in\mathcal{C}$. First we show: 
\begin{lemma}
\label{lem:convergent_subsequence}Let $\mleft(\ser{\alpha}^{j}\mright)_{j\in\mathbb{N}}$
be a sequence of sequences $\ser{\alpha}^{j}\in\mathcal{A}$ such
that $\lim_{j\rightarrow\infty}P\mleft(\ser c,\ser{\alpha}^{j}\mright)=\sup_{\ser{\alpha}\in\mathcal{A}}P\mleft(\ser c,\ser{\alpha}\mright)$.
Then there is a sequence $\ser{\alpha}$, and a subsequence of $\mleft(\ser{\alpha}^{j}\mright)_{j\in\mathbb{N}}$
which we denote by $\mleft(\ser{\alpha}'^{j}\mright)_{j\in\mathbb{N}}$,
such that $\ser{\alpha}'^{j}\rightarrow\ser{\alpha}$ pointwise, that
is, for any $i$: $\lim_{j\rightarrow\infty}\alpha_{i}'^{j}=\alpha_{i}$.
\end{lemma}
\begin{proof}
Consider the sequence $\mleft(\alpha_{0}^{j}\mright)_{j\in\mathbb{N}}$.
Since $\alpha_{0}^{j}\in\mleft[0,1\mright]$ for any $j$, $\mleft(\alpha_{0}^{j}\mright)_{j\in\mathbb{N}}$
must have a convergent subsequence due to Bolzano-Weiersrtrass. Denote
that subsequence by $\mleft(\alpha_{0}'^{j}\mright)_{j\in\mathbb{N}}$,
and let $\alpha_{0}=\lim_{j\rightarrow\infty}\alpha_{0}'^{j}$. Deonte
by $\ser{\alpha}^{\mleft(0\mright)}=\mleft(\ser{\alpha}^{\mleft(0\mright),j}\mright)_{j\in\mathbb{N}}$
the subsequence of $\mleft(\ser{\alpha}^{j}\mright)_{j\in\mathbb{N}}$
that is constructed from the same indices as $\mleft(\alpha_{0}'^{j}\mright)_{j\in\mathbb{N}}$.
Now, consider the sequence $\mleft(\alpha_{1}^{\mleft(0\mright),j}\mright)_{j\in\mathbb{N}}$.
This sequence as well has a subsequence which converges, say to $\alpha_{1}$.
Let $\ser{\alpha}^{\mleft(1\mright)}$ be the subsequence of $\mleft(\ser{\alpha}^{j}\mright)_{j\in\mathbb{N}}$
that is constructed from the same indices of the subsequence of $\mleft(\alpha_{1}^{\mleft(0\mright),j}\mright)_{j\in\mathbb{N}}$
which converges to $\alpha_{1}$. Note that in addition to $\alpha_{1}^{\mleft(1\mright),j}\rightarrow\alpha_{1}$,
we also have $\alpha_{0}^{\mleft(1\mright),j}\rightarrow\alpha_{0}$
since the limit of a convergent sequence equals the limit of any of
its subsequences. We can proceed in the same fashion, constructing
a sequence $\ser{\alpha}$, and for any $r$ a subsequence $\ser{\alpha}^{\mleft(r\mright)}$
of $\mleft(\ser{\alpha}^{j}\mright)_{j\in\mathbb{N}}$ such that $\alpha_{s}^{\mleft(r\mright),j}\rightarrow\alpha_{s}$
for any $s\leq r$. We take as $\mleft(\ser{\alpha}'^{j}\mright)_{j\in\mathbb{N}}$
the diagonal sequence $\mleft(\ser{\alpha}^{\mleft(j\mright),j}\mright)_{j\in\mathbb{N}}$
which converges pointwise to $\ser{\alpha}$ . 
\end{proof}
Let $\ser{\alpha}$ be the sequence guaranteed by this lemma. We will
show that the supremum is attained at $\ser{\alpha},$ that is, $P\mleft(\ser c,\ser{\alpha}\mright)=\sup_{\ser{\alpha'}\in\mathcal{A}}P\mleft(\ser c,\ser{\alpha'}\mright)$.
It remains to show:
\begin{lemma}
The sequence $\ser{\alpha}$ found by Lemma \ref{lem:convergent_subsequence}
is in $\mathcal{A}$ (that is, $\ser{\alpha}$ is feasible for $\mleft(\ser c,b\mright)$)
and $\lim_{j\rightarrow\infty}P\mleft(\ser c,\ser{\alpha}^{j}\mright)=P\mleft(\ser c,\ser{\alpha}\mright)$.
\end{lemma}
\begin{proof}
Let us begin by showing $\ser{\alpha}\in\mathcal{A}$: pointwise convergence
of $\mleft(\ser{\alpha}^{j}\mright)_{j\in\mathbb{N}}$ to $\ser{\alpha}$
ensures that $\ser{\alpha}\in\mleft[0,1\mright]^{\mathbb{N}}$, since
$\ser{\alpha}^{j}\in\mleft[0,1\mright]^{\mathbb{N}}$ for any $j$.
It remains to show that $\frac{1}{\beta\cdot2^{b+1}}=\sum_{i=0}^{\infty}\alpha_{i}c_{i}/2^{i}$:
Take an arbitrary $\epsilon>0$ and find $I$ such that $\sum_{i>I}c_{i}<\epsilon/3$.
Then, find $J$ such that $\mleft|\alpha_{i}^{J}-\alpha_{i}\mright|<\epsilon/3$
for all $i\leq I$. Thus: 
\begin{align*}
\frac{1}{\beta\cdot2^{b+1}} & =\sum_{i=0}^{\infty}\alpha_{i}^{J}c_{i}/2^{i}\\
 & =\sum_{i=0}^{I}\alpha_{i}^{J}c_{i}/2^{i}\pm\epsilon/3\\
 & =\sum_{i=0}^{I}\mleft(\alpha_{i}\pm\epsilon/3\mright)c_{i}/2^{i}\pm\epsilon/3\\
 & =\sum_{i=0}^{I}\alpha_{i}c_{i}/2^{i}\pm2\epsilon/3=\sum_{i=0}^{\infty}\alpha_{i}c_{i}/2^{i}\pm\epsilon,
\end{align*}
 and so indeed $\ser{\alpha}\in\mathcal{A}$. Let us show that $\lim_{j\rightarrow\infty}P\mleft(\ser c,\ser{\alpha}^{j}\mright)=P\mleft(\ser c,\ser{\alpha}\mright)$.
For some $I\in\mathbb{N}$, denote $P^{I}\mleft(\ser{c,}\ser{\alpha}\mright)=\sum_{i=0}^{I}h\mleft(\alpha_{i}\mright)c_{i}-h\mleft(\sum_{i=0}^{I}\alpha_{i}c_{i}\mright)$.
Take an arbitrary $\epsilon>0$ and let $I$ such that $\sum_{i>I}c_{i}<\epsilon$. So: 
\begin{gather} \label{eq:sup_1}
P\mleft(\ser c,\ser{\alpha}\mright)=\sum_{i=0}^{I}h\mleft(\alpha_{i}\mright)c_{i}+\sum_{i=I+1}^{\infty}h\mleft(\alpha_{i}\mright)c_{i}-h\mleft(\sum_{i=0}^{I}\alpha_{i}c_{i}+\sum_{i=I+1}^{\infty}\alpha_{i}c_{i}\mright)=P^{I}\mleft(\ser{c,}\ser{\alpha}\mright)\pm\mleft(\epsilon+h\mleft(\epsilon\mright)\mright)
\end{gather}
 due to Lemma \ref{lem:entropy_bounds}. Now we can use an argument
similar to the one used to show feasibility of $\ser{\alpha}$: find
$J$ such that $\mleft|\alpha_{i}^{J}-\alpha_{i}\mright|<\epsilon$
for all $i\leq I$. So:
\begin{align} \label{eq:sup_2}
P^{I}\mleft(\ser{c,}\ser{\alpha}\mright) & =\sum_{i=0}^{I}h\mleft(\alpha_{i}^{J}\pm\epsilon\mright)c_{i}-h\mleft(\sum_{i=0}^{I}\mleft(\alpha_{i}^{J}\pm\epsilon\mright)c_{i}\mright)\nonumber \\
 & =\sum_{i=0}^{I}h\mleft(\alpha_{i}^{J}\mright)c_{i}\pm h\mleft(\epsilon\mright)\sum_{i=0}^{I}c_{i}-h\mleft(\sum_{i=0}^{I}\alpha_{i}^{J}c_{i}\pm\epsilon\sum_{i=0}^{I}c_{i}\mright)=P^{I}\mleft(\ser{c,}\ser{\alpha^{J}}\mright)\pm2h\mleft(\epsilon\mright).
\end{align}
 And so: 
\begin{gather*}
P\mleft(\ser c,\ser{\alpha^{J}}\mright)\underset{\mleft(*\mright)}{=}P^{I}\mleft(\ser c,\ser{\alpha^{J}}\mright)\pm\mleft(\epsilon+h\mleft(\epsilon\mright)\mright)\underset{\eqref{eq:sup_2}}{=}P^{I}\mleft(\ser{c,}\ser{\alpha}\mright)\pm\mleft(\epsilon+3h\mleft(\epsilon\mright)\mright)\underset{\eqref{eq:sup_1}}{=}P\mleft(\ser c,\ser{\alpha}\mright)\pm\mleft(2\epsilon+4h\mleft(\epsilon\mright)\mright),
\end{gather*}
 where $\mleft(*\mright)$ is since $\sum_{i>I}c_{i}<\epsilon$, similarly
to eq. (\ref{eq:sup_1}). So indeed, $\lim_{j\rightarrow\infty}P\mleft(\ser c,\ser{\alpha}^{j}\mright)=P\mleft(\ser c,\ser{\alpha}\mright)$.
\end{proof}
The following desired result is an immediate corollary:
\begin{lemma}
For any $\mleft(\ser c,b\mright)\in\mathcal{C}$ there is $\ser{\alpha}\in\mathcal{A}$
such that $\sup_{\ser{\alpha'}\in\mathcal{A}}P\mleft(\ser c,\ser{\alpha'}\mright)=P\mleft(\ser c,\ser{\alpha}\mright)$.
\end{lemma}

\subsection{Proving our formula for $\rho_{\min}\mleft(n\mright)$}

The following bounds on $\rho_{\min}^{*}\mleft(n\mright)$ immediately
imply Theorem~\ref{theo:main}, due to Lemma \ref{lem:MRD_star}:
\begin{lemma}
\label{lem:low_bound} It holds that $\rho_{\min}^{*}\mleft(n\mright)\geq2^{G\mleft(\beta\mright)n-o\mleft(n\mright)}.$ 
\end{lemma}
\begin{lemma}
\label{lem:upper_bound} It holds that $\rho_{\min}^{*}\mleft(n\mright)\leq2^{G\mleft(\beta\mright)n+o\mleft(n\mright)}$. 
\end{lemma}
In the following subsections we prove those bounds. 

\subsubsection{Lower bounding $\rho_{\min}^{*}\mleft(n\mright)$}

Recall that if a pair $\mleft(\ser c,b\mright)\in\mathcal{C}$ satisfies
$c_{i}n\in\mathbb{N}$ for all $i$, we say that $\ser c$ is\emph{
$k$-}feasible. Similarly, if a sequence $\ser{\alpha}\in\mathcal{A}$
satisfies $\alpha_{i}c_{i}n\in\mathbb{N}$ for all $i$, and $\alpha_{t}<1$ , we say that $\ser{\alpha}$ is $k$-feasible.
Note that for a fixed $\mleft(\ser c,b\mright)\in\mathcal{C}$ which
is $k$-feasible, by our definitions: 
\[
\mleft\{ \ser{\alpha}\in\mathcal{A}\colon\ser{\alpha}\text{ is \ensuremath{k\text{-feasible}}}\mright\} =\bigcup_{d\in\mleft[n\mright]}S_{d}.
\]
 We will use that connection throughout the proof, when linking between
$\rho_{\min}^{*}\mleft(n\mright)$, which uses the set $S_{d}$ for
some optimal $d$, and $G\mleft(\beta\mright)$ which uses the set $\mathcal{A}$.
For a set of sequences $\mathcal{S}\subseteq\mleft[0,1\mright]^{\mathbb{N}}$,
let 
\[
\mathcal{S}^{\leq l}=\mleft\{ \ser s\in\mathcal{S}\colon i>l\implies s_{i}=0\mright\} .
\]
 Lemma \ref{lem:low_bound} can be inferred from the following two
lemmas: 
\begin{lemma}
\label{lem:non-empty} If $\mleft(\ser c,b\mright)\in\mathcal{C}$ and $\ser c$ is $k$-feasible, then there is $\ser{\alpha}\in\mathcal{A}^{\leq k}$
which is $k$-feasible.
\end{lemma}
The purpose of this lemma is to allow us to use the estimate 
\begin{equation}
2^{h\mleft(\lambda\mright)n}/O\mleft(\sqrt{n}\mright)\leq\binom{n}{\lambda n}\leq2^{h\mleft(\lambda\mright)n}\label{eq:binom_bounds}
\end{equation}
 for $0 \leq \lambda \leq 1$ (the lower bound is due to \cite{14476}, for example) while proving
Lemma \ref{lem:low_bound}, in a sufficiently efficient fashion. 
\begin{lemma}
\label{lem:low_bound_1} Fix $\mleft(\ser c,b\mright)\in\mathcal{C}$.
Then: 
\[
\lim_{k\rightarrow\infty}\max_{\substack{
\ser{\alpha}\in\mathcal{A}^{\leq k}\colon\\
\ser{\alpha}\text{ is k-feasible}
}}P\mleft(\ser c,\ser{\alpha}\mright)=\max_{\ser{\alpha}\in\mathcal{A}}P\mleft(\ser c,\ser{\alpha}\mright).
\]
 
\end{lemma}
For large values of $k$, Lemma \ref{lem:low_bound_1} allows us to
remove the $k$-feasibility and $i>k\implies\alpha_{i}=0$ constraints
on $\ser{\alpha}$ without changing much the value of $P\mleft(\ser c,\ser{\alpha}\mright)$.
Having Lemmas \ref{lem:non-empty}--\ref{lem:low_bound_1} in hand
and using the estimate (\ref{eq:binom_bounds}), Lemma \ref{lem:low_bound}
can be proved: 
\begin{proof}
[Proof of Lemma \ref{lem:low_bound}] Let $n=\beta\cdot2^{k}$ and
$\mleft(\ser c,b\mright)\in\mathcal{C}$ which is $k$-feasible. So:
\begin{align*}
\max_{\substack{
d\in\mleft[n\mright]\colon\\
S_{d}\neq\emptyset
}}\sum_{\ser{\alpha}\in S_{d}}\frac{\prod_{i=0}^{\infty}\binom{c_{i}n}{\alpha_{i}c_{i}n}}{\binom{n}{d}} & \underset{\text{Lem \ref{lem:non-empty}}}{\geq}\max_{\substack{
\ser{\alpha}\in\mathcal{A}^{\leq k}\colon\\
\ser{\alpha}\text{ is \ensuremath{k}-feasible}
}}\frac{\prod_{i=0}^{\infty}\binom{c_{i}n}{\alpha_{i}c_{i}n}}{\binom{n}{\sum_{i=0}^{\infty}\alpha_{i}c_{i}n}}\\
 & \underset{\text{\eqref{eq:binom_bounds}}}{\geq}\max_{\substack{
\ser{\alpha}\in\mathcal{A}^{\leq k}\colon\\
\ser{\alpha}\text{ is \ensuremath{k}-feasible}
}}\frac{\exp_{2}\mleft(\sum_{i=0}^{\infty}h\mleft(\alpha_{i}\mright)c_{i}n\mright)/O\mleft(\sqrt{n}^{k}\mright)}{\exp_{2}\mleft(h\mleft(\sum_{i=0}^{\infty}\alpha_{i}c_{i}\mright)n\mright)}\\
 & =\exp_{2}\max_{\substack{
\ser{\alpha}\in\mathcal{A}^{\leq k}\colon\\
\ser{\alpha}\text{ is \ensuremath{k}-feasible}
}}P\mleft(\ser c,\ser{\alpha}\mright)n-o\mleft(n\mright)\\
& \underset{\text{Lem \ensuremath{\ref{lem:low_bound_1}}}}{\geq}\exp_{2}\mleft(\max_{\ser{\alpha}\in\mathcal{A}}P\mleft(\ser c,\ser{\alpha}\mright)n-o\mleft(n\mright)\mright).
\end{align*}
 Let $\mleft(\ser c,b\mright)\in\mathcal{C}$ which is $k$-feasible
and 
\[
\min_{\substack{
\mleft(\ser{c'},b\mright)\in\mathcal{C}\colon\\
\ser{c'}\text{ is \ensuremath{k}-feasible}
}}\max_{\substack{
d\in\mleft[n\mright]\colon\\
S_{d}\neq\emptyset
}}\sum_{\ser{\alpha}\in S_{d}}\frac{\prod_{i=0}^{\infty}\binom{c'_{i}n }{\alpha_{i}c'_{i}n}}{\binom{n}{d}}=\max_{\substack{
d\in\mleft[n\mright]\colon\\
S_{d}\neq\emptyset
}}\sum_{\ser{\alpha}\in S_{d}}\frac{\prod_{i=0}^{\infty}\binom{c_{i}n }{\alpha_{i}c_{i}n}}{\binom{n}{d}},
\]
and deduce: 
\begin{align*}
\rho_{\min}^{*}\mleft(n\mright) & =\min_{\substack{
\mleft(\ser{c'},b\mright)\in\mathcal{C}\colon\\
\ser{c'}\text{ is \ensuremath{k}-feasible}
}}\max_{\substack{
d\in\mleft[n\mright]\colon\\
S_{d}\neq\emptyset
}}\sum_{\ser{\alpha}\in S_{d}}\frac{\prod_{i=0}^{\infty}\binom{c'_{i}n}{\alpha_{i}c'_{i}n}}{\binom{n}{d}}\\
 & =\max_{\substack{
d\in\mleft[n\mright]\colon\\
S_{d}\neq\emptyset
}}\sum_{\ser{\alpha}\in S_{d}}\frac{\prod_{i=0}^{\infty}\binom{c_{i}n}{\alpha_{i}c_{i}n}}{\binom{n}{d}}\\
 & \geq\exp_{2}\mleft(\max_{\ser{\alpha}\in\mathcal{A}}P\mleft(\ser c,\ser{\alpha}\mright)n-o\mleft(n\mright)\mright)\\
 & \geq\inf_{\substack{
\mleft(\ser{c'},b\mright)\in\mathcal{C}}}\exp_{2}\mleft(\max_{\ser{\alpha}\in\mathcal{A}}P\mleft(\ser{c'},\ser{\alpha}\mright)n-o\mleft(n\mright)\mright)=2^{G\mleft(\beta\mright)n-o\mleft(n\mright)}.\qedhere
\end{align*}
\end{proof}
Now we shall prove Lemmas \ref{lem:non-empty}--\ref{lem:low_bound_1}:
\begin{proof}
[Proof of Lemma \ref{lem:non-empty}] Let $\mleft(\ser c,b\mright)\in\mathcal{C}$
and assume $\ser c$ is $k$-feasible, then we have: 
\[
\sum_{i=k+1}^{\infty}c_{i}\cdot2^{b-i}\cdot\beta\underset{\mleft(*\mright)}{<}\sum_{i=k+1}^{\infty}c_{i}\cdot2^{k-1-\mleft(k+1\mright)}\cdot2=2^{-1}\cdot\sum_{i=k+1}^{\infty}c_{i}\leq1/2
\]
 where $\mleft(*\mright)$ is since $b\leq k-1$ (otherwise $\mleft(\ser c,b\mright)$
represents a constant dyadic distribution), $i\geq k+1$, and $\beta<2$.
Since $\sum_{i=0}^{\infty}c_{i}\cdot2^{b-i}\cdot\beta=1$, we deduce
that $\sum_{i=0}^{k}c_{i}\cdot2^{b-i}\cdot\beta>1/2$. Thus, by Lemma
4.1 in \cite{DFGM2019} (called there ``A useful lemma''), we know
that there is a splitting set of the dyadic distribution represented
by $\mleft(\ser c,b\mright)$ containing only elements with probabilities
$\mu_{1},\mu_{1}/2,\dots,\mu_{1}/2^{k}$. The same lemma also implies that $\alpha_t<1$. That is, there is $\ser{\alpha}\in\mathcal{A}^{\leq k}$ which is $k$-feasible.
\end{proof}
The proof of Lemma \ref{lem:low_bound_1} will require the following:
\begin{lemma}
\label{lem:low_bound_2} Fix $\mleft(\ser c,b\mright)\in\mathcal{C}$,
$\epsilon>0$ and $\ser{\alpha}\in\mathcal{A}$. There are $K\in\mathbb{N}$
and $\ser{\tilde{\alpha}}\in\mathcal{A}^{\leq K}$, where $\ser{\tilde{\alpha}}$
is $K$-feasible and satisfies $P\mleft(\ser c,\ser{\alpha}\mright)=P\mleft(\ser c,\ser{\tilde{\alpha}}\mright)\pm\epsilon$. 
\end{lemma}
Having Lemma \ref{lem:low_bound_2}, Lemma \ref{lem:low_bound_1}
is almost immediate: 
\begin{proof}
[Proof of Lemma \ref{lem:low_bound_1}] Fix $\mleft(\ser c,b\mright)\in\mathcal{C}$
and let $\epsilon>0$. Let $\ser{\alpha}\in\mathcal{A}$ such that
$P\mleft(\ser c,\ser{\alpha}\mright)=\max_{\ser{\alpha'}\in\mathcal{A}}P\mleft(\ser c,\ser{\alpha'}\mright)$.
Let $K\in\mathbb{N}$ large enough such that Lemma \ref{lem:low_bound_2}
holds for $\mleft(\ser c,b\mright)$, $\ser{\alpha}$ and $\epsilon$.
Then for any $k\geq K$: 
\begin{gather*}
\max_{\ser{\alpha'}\in\mathcal{A}}P\mleft(\ser c,\ser{\alpha'}\mright)=P\mleft(\ser c,\ser{\alpha}\mright)\underset{\text{Lem }\ref{lem:low_bound_2}}{\leq}P\mleft(\ser c,\ser{\tilde{\alpha}}\mright)+\epsilon\leq\max_{\substack{
\ser{\alpha'}\in\mathcal{A}\\
\ser{\alpha'}\text{ is \ensuremath{k}-feasible}
}}P\mleft(\ser c,\ser{\alpha'}\mright)+\epsilon.
\end{gather*}
 Obviously,
\[
\max_{\substack{
\ser{\alpha'}\in\mathcal{A}^{\leq k}\\
\ser{\alpha'}\text{ is \ensuremath{k}-feasible}
}}P\mleft(\ser c,\ser{\alpha'}\mright)\leq\max_{\ser{\alpha'}\in\mathcal{A}}P\mleft(\ser c,\ser{\alpha'}\mright)
\]
for any large enough $k$, thus the lemma follows. It is necessary
that $k$ is large enough: Note that for the $K$ determined by Lemma
\ref{lem:low_bound_2}, for example, we can be sure that there is
$\ser{\alpha}\in\mathcal{A}^{\leq K}$ which is $K$-feasible, while
for small $k$ values there might not be such $\ser{\alpha}$, and
then the left-hand side is not well defined.
\end{proof}
It remains to prove Lemma \ref{lem:low_bound_2}, which is the main
part of the proof for the lower bound. We first explain the proof
idea, and then give the detailed proof. 
\begin{proof}
[Proof sketch of Lemma \ref{lem:low_bound_2}] Let $\mleft(\ser c,b\mright)\in\mathcal{C}$
and $\ser{\alpha}\in\mathcal{A}$. Since the binary entropy function
is sub-additive and symmetric, it holds that h$\mleft(x\pm\epsilon\mright)=h\mleft(x\mright)\pm h\mleft(\epsilon\mright)$
for $0\leq x\pm\epsilon\leq1$ (as we have shown in Lemma \ref{lem:entropy_bounds}).
Based on that, the proof idea is that if we make very small changes
in $\ser{\alpha}$, to get some other sequence which we denote $\ser{\tilde{\alpha}}$,
we can get $P\mleft(\ser c,\ser{\alpha}\mright)=P\mleft(\ser c,\ser{\tilde{\alpha}}\mright)\pm\epsilon$.
The main difficulty is to make small changes to $\ser{\alpha}$ while
ensuring that for some $K\in\mathbb{N}$, $\ser{\tilde{\alpha}}\in\mathcal{A}^{\leq K}$
and $\ser{\tilde{\alpha}}$ is $K$-feasible. We can solve that difficulty
by defining a sequence $\ser{\epsilon}$ of very small values carefully
selected, and then defining $\ser{\tilde{\alpha}}$ in the following
way: 
\[
\tilde{\alpha}_{i}=\begin{cases}
\alpha_{s}+\epsilon_{s} & i=s,\\
\alpha_{i}-\epsilon_{i} & i\neq s,
\end{cases}
\]
 where $s$ is some index chosen to make sure that $\ser{\tilde{\alpha}}\in\mleft[0,1\mright]^{\mathbb{N}}$.
The value $\tilde{\alpha}_{s}$ adds $\epsilon_{s}$ in order to
``balance'', in a way, the subtraction of $\epsilon_{i}$ in other
indices, such that we have $\sum_{i=0}^{\infty}\alpha_{i}c_{i}/2^{i}=\sum_{i=0}^{\infty}\tilde{\alpha}_{i}c_{i}/2^{i}$
and hence the constraint $\sum_{i=0}^{\infty}\tilde{\alpha}_{i}c_{i}/2^{i}=\frac{1}{\beta\cdot2^{b+1}}$
is satisfied (since the constraint $\sum_{i=0}^{\infty}\alpha_{i}c_{i}/2^{i}=\frac{1}{\beta\cdot2^{b+1}}$
is satisfied). The purpose of the addition or subtraction of the $\ser{\epsilon}$
sequence values is to ``round'' the values of $\ser{\alpha}$ to
get $\ser{\tilde{\alpha}}$ which is $K$-feasible and belongs to
$\mathcal{A}^{\leq K}$. The exact choice of the sequence $\ser{\epsilon}$
that guarantees that is described in the detailed proof. 
\end{proof}
Before we give a detailed proof of Lemma \ref{lem:low_bound_2}, we
prove a lemma which will be useful in the detailed proof, and also
later when upper bounding $\rho_{\min}^{*}\mleft(n\mright)$. Its purpose
is to ensure that when we change a sequence $\ser{\alpha}\in\mathcal{A}$
to a sequence $\ser{\tilde{\alpha}}$, we use $\ser{\tilde{\alpha}}\in\mleft[0,1\mright]^{\mathbb{N}}$.
\begin{lemma}
\label{lem:s_index}Fix $b\in\mathbb{N}$. For any $\ser c\in\mleft[0,1\mright]^{\mathbb{N}}$
such that $\sum_{i=0}^{\infty}c_{i}/2^{i}=\frac{1}{\beta\cdot2^{b}}$
for some $1\leq\beta<2$, and for any $\ser{\alpha}\in\mleft[0,1\mright]^{\mathbb{N}}$
which satisfies $\sum_{i=0}^{\infty}\alpha_{i}c_{i}/2^{i}=\frac{1}{2^{b+1}\beta}$,
there are indices $s_{1},s_{2}\leq2^{b+4}$ (possibly $s_{1}=s_{2}$)
such that $c_{s_{1}},c_{s_{2}}>1/2^{2\mleft(b+5\mright)}$ and $\alpha_{s_{1}}>1/2^{2\mleft(b+5\mright)}$,
$\alpha_{s_{2}}<3/4$. 
\end{lemma}
\begin{proof}
Fix $b\in\mathbb{N}$ and take arbitrary $\ser c,\beta$ such that
$\sum_{i}c_{i}/2^{i}=\frac{1}{\beta\cdot2^{b}}$. Let us find $s_{1}$:
take an arbitrary sequence $\ser{\alpha}$ such that $\sum_{i=0}^{\infty}\alpha_{i}c_{i}/2^{i}=\frac{1}{2^{b+1}\beta}$.
Denote $I=\mleft\{ i:i\leq2^{b+4}\mright\} $ and let $S\subseteq I$
be the set of indices in $I$ which satisfy $c_{i}\geq1/2^{2\mleft(b+5\mright)}$.
Assume towards contradiction that for any $i\in S,$ $\alpha_{i}\leq\frac{1}{2^{2\mleft(b+5\mright)}}$.
Note that since $\sum_{i=0}^{\infty}\alpha_{i}c_{i}/2^{i}=\frac{1}{2^{b+1}\beta}>1/2^{b+2}$
and $\sum_{i>2^{b+4}}c_{i}/2^{i}\leq\sum_{i>2^{b+4}}1/2^{i}=\frac{1}{2^{b+4}}$,
it must hold that $\sum_{i=0}^{2^{b+4}}\alpha_{i}c_{i}/2^{i}\geq\frac{3}{2^{b+4}}$.
However, we have: 
\begin{align*}
\sum_{i=0}^{2^{b+4}}\alpha_{i}c_{i}/2^{i} & =\sum_{i\in S}\alpha_{i}c_{i}/2^{i}+\sum_{i\in I\backslash S}\alpha_{i}c_{i}/2^{i}\\
&  \leq\frac{1}{2^{2\mleft(b+5\mright)}}\cdot\sum_{i\in S}c_{i}+\sum_{i\in I\backslash S}c_{i}\\
& <\frac{2^{b+5}}{2^{2\mleft(b+5\mright)}}+\frac{2^{b+5}}{2^{2\mleft(b+5\mright)}}\\
& =2\cdot\frac{1}{2^{b+5}}=\frac{1}{2^{b+4}},
\end{align*}
 and that is a contradiction, thus there is $i\in S$ with $\alpha_{i}\geq1/2^{2\mleft(b+5\mright)}$.
That is, there is an index $i\leq2^{b+4}$ with $c_{i},\alpha_{i}\geq1/2^{2\mleft(b+5\mright)}$,
and this is $s_{1}$. Let us now find $s_{2}$: assume towards contradiction
that for any $i\in S$, $\alpha_{i}\geq3/4$. We show that if that
assumption is true, then $\text{\ensuremath{\sum_{i\in S}\alpha_{i}c_{i}}/\ensuremath{2^{i}} }$
is too large. First, we have: 
\[
\frac{1}{2^{b}\beta}=\sum_{i=0}^{\infty}c_{i}/2^{i}=\sum_{i=0}^{2^{b+4}}c_{i}/2^{i}+\sum_{i>2^{b+4}}c_{i}/2^{i}\leq\sum_{i=0}^{2^{b+4}}c_{i}/2^{i}+1/2^{b+4},
\]
 that is: 
\[
\sum_{i=0}^{2^{b+4}}c_{i}/2^{i}\geq\frac{1}{2^{b}\beta}-1/2^{b+4}.
\]
 Moreover, it holds that: 
\[
\sum_{i\in I\backslash S}c_{i}/2^{i}\leq\frac{2^{b+5}}{2^{2\mleft(b+5\mright)}}=1/2^{b+5}.
\]
Hence, we get that: 
\[
\sum_{i\in S}c_{i}/2^{i}\geq\frac{1}{2^{b}\beta}-1/2^{b+4}-1/2^{b+5}=\frac{1}{2^{b}\beta}-3/2^{b+5}.
\]
And thus, assuming $\alpha_{i}\geq3/4$ for any $i\in S$: 
\[
\sum_{i\in S}\alpha_{i}c_{i}/2^{i}\geq\frac{3}{4}\sum_{i\in S}c_{i}/2^{i}\geq\frac{3}{4}\cdot\mleft(\frac{1}{2^{b}\beta}-3/2^{b+5}\mright).
\]
 But then we get: 
\begin{align*}
\frac{3}{4}\mleft(\frac{1}{2^{b}\beta}-3/2^{b+5}\mright) &= \frac{3}{4}\mleft(\frac{2^{5}-3\beta}{2^{b+5}\beta}\mright)\\
&> \frac{3}{4}\cdot\frac{2^{5}-2^{3}}{2^{b+5}\beta}\\
&= \frac{3}{4}\cdot\frac{2^{3}\mleft(2^{2}-1\mright)}{2^{b+5}\beta}\\
&= \frac{3\cdot3}{2^{2}\cdot2^{b+2}\beta}\\
&> \frac{2^{3}}{2^{b+4}\beta}=\frac{1}{2^{b+1}\beta},
\end{align*}
and this contradicts the fact $\sum_{i=0}^{\infty}\alpha_{i}c_{i}/2^{i}=\frac{1}{2^{b+1}\beta}$.
Thus there is $i\leq2^{b+4}$ with $c_{i}\geq1/2^{2\mleft(b+5\mright)}$
and $\alpha_{i}<3/4$, and this is $s_{2}$. 
\end{proof}
Now we can go on with the detailed proof of Lemma \ref{lem:low_bound_2}. 
\begin{proof}
[Proof of Lemma \ref{lem:low_bound_2}] We divide the proof into
three parts: First we define the sequence $\ser{\tilde{\alpha}}$
and show it exists. Then we show $\ser{\tilde{\alpha}}\in\mathcal{A}^{\leq K'}$
and $\ser{\tilde{\alpha}}$ is $K'$-feasible for some $K'\in\mathbb{N}$.
Finally, we show $P\mleft(\ser c,\ser{\alpha}\mright)=P\mleft(\ser c,\ser{\tilde{\alpha}}\mright)\pm\epsilon$. 

\paragraph{Defining $\protect\ser{\tilde{\alpha}}$.}

Fix $\mleft(\ser c,b\mright)\in\mleft[0,1\mright]^{\mathbb{N}}$ satisfying
$\sum_{i=0}^{\infty}c_{i}/2^{i}=\frac{1}{\beta\cdot2^{b}}$ and $\sum_{i=0}^{\infty}c_{i}\leq1$
(that is, $\mleft(\ser c,b\mright)\in\mathcal{C}$). Fix $\ser{\alpha}\in\mleft[0,1\mright]^{\mathbb{N}}$
satisfying $\sum_{i=0}^{\infty}\alpha_{i}c_{i}/2^{i}=\frac{1}{\beta\cdot2^{b+1}}$
(that is, $\ser{\alpha}\in\mathcal{A}$). Let $\epsilon>0$ small
enough (the proof holds for any $\epsilon>0$ smaller than some constant).
Let $s$ be the lowest index satisfying $c_{s}>1/2^{2\mleft(b+5\mright)}$
and $\alpha_{s}<3/4$ ($s$ exists due to Lemma \ref{lem:s_index}).
Since $\sum_{i=0}^{\infty}c_{i}\leq1$, there is $K\in\mathbb{N}$
such that $\sum_{i>K}c_{i}<\epsilon$ and $K\geq s$. Define the sequence
$\ser{\tilde{\alpha}}$ as follows: 
\[
\tilde{\alpha}_{i}=\begin{cases}
\alpha_{s}+\epsilon_{s} & i=s,\\
\alpha_{i}-\epsilon_{i} & i\neq s,
\end{cases}
\]
 where $\ser{\epsilon}$ is an arbitrary sequence of small values
satisfying the following constraints:
\begin{enumerate}
\item If $i>K$, or $c_{i}=0$, or $\alpha_{i}=0$: $\epsilon_{i}=\alpha_{i}$.
\item Otherwise, if $0\leq i\leq K$ and $i\neq s$: $\epsilon_{i}=\alpha_{i}-\frac{l_{i}}{c_{i}\cdot\beta\cdot2^{t_{i}}}$,
where $l_{i},t_{i}\in\mathbb{N}$, $\epsilon_{i}>0$ and $h\mleft(\epsilon_{i}\mright)\leq\epsilon/K$.
\item $\epsilon_{s}=\sum_{i\neq s}\epsilon_{i}c_{i}\cdot2^{s-i}/c_{s}=\frac{l_{s}}{c_{s}\cdot\beta\cdot2^{t_{s}}}-\alpha_{s}$,
where $l_{s},t_{s}\in\mathbb{N}$. 
\end{enumerate}
In order to continue with the proof, we first have to show that such
a sequence $\ser{\epsilon}$ exists. Denote by $I$ the set of all
indices ``that matter'' in $\ser{\tilde{\alpha}}$, that is, $I=\mleft\{ i\leq K\colon c_{i},\alpha_{i}>0\mright\} $.
It is not hard to construct a sequence $\ser{\epsilon}$ that satisfies
constraints $\mleft(1\mright),\mleft(2\mright)$. We should satisfy constraint
$\mleft(3\mright)$ as well, that is 
\begin{align*}
\frac{l_{s}}{c_{s}\cdot\beta\cdot2^{t_{s}}}-\alpha_{s} & =\sum_{i\neq s}\epsilon_{i}c_{i}\cdot2^{s-i}/c_{s}\\
 & =\frac{2^{s}}{c_{s}}\mleft[\sum_{i\in I\backslash\mleft\{ s\mright\} }\mleft(\alpha_{i}-\text{\ensuremath{\frac{l_{i}}{c_{i}\cdot\beta\cdot2^{t_{i}}}}}\mright)\cdot c_{i}/2^{i}+\sum_{i>K}\alpha_{i}c_{i}/2^{i}\mright]\\
 & =\frac{2^{s}}{c_{s}}\mleft[\sum_{i\neq s}\alpha_{i}c_{i}/2^{i}-\sum_{i\in I\backslash\mleft\{ s\mright\} }\text{\ensuremath{\frac{l_{i}}{\beta\cdot2^{t_{i}+i}}}}\mright].
\end{align*}
 This can be written as: 
\begin{gather*}
2^{s}\mleft[\sum_{i\neq s}\alpha_{i}c_{i}/2^{i}-\sum_{i\in I\backslash\mleft\{ s\mright\} }\text{\ensuremath{\frac{l_{i}}{\beta\cdot2^{t_{i}+i}}}}\mright]=\frac{l_{s}}{\beta\cdot2^{t_{s}}}-\alpha_{s}c_{s}\\
\iff\\
\sum_{i\neq s}\alpha_{i}c_{i}/2^{i}-\sum_{i\in I\backslash\mleft\{ s\mright\} }\text{\ensuremath{\frac{l_{i}}{\beta\cdot2^{t_{i}+i}}}}=\frac{l_{s}}{\beta\cdot2^{t_{s}+s}}-\alpha_{s}c_{s}/2^{s}\\
\iff\\
\sum_{i=0}^{\infty}\alpha_{i}c_{i}/2^{i}-\sum_{i\in I\backslash\mleft\{ s\mright\} }\text{\ensuremath{\frac{l_{i}}{\beta\cdot2^{t_{i}+i}}}}=\frac{l_{s}}{\beta\cdot2^{t_{s}+s}}.
\end{gather*}
 Recall that $\sum_{i=0}^{\infty}\alpha_{i}c_{i}/2^{i}=\frac{1}{\beta\cdot2^{b+1}}$,
thus we have: 
\begin{gather*}
\frac{1}{\beta\cdot2^{b+1}}-\sum_{i\in I\backslash\mleft\{ s\mright\} }\text{\ensuremath{\frac{l_{i}}{\beta\cdot2^{t_{i}+i}}}}=\frac{l_{s}}{\beta\cdot2^{t_{s}+s}}\\
\iff\\
1/2^{b+1}-\sum_{i\in I\backslash\mleft\{ s\mright\} }\text{\ensuremath{l_{i}/2^{t_{i}+i}}}=l_{s}/2^{t_{s}+s},
\end{gather*}
 and there are $l_{s},t_{s}\in\mathbb{N}$ satisfying this equation:
Let $t_{s}=b+1+\sum_{i\in I\backslash\mleft\{ s\mright\} }t_{i}+i$,
then $l_{s}$ is determined accordingly such that the equation holds.
Clearly, $t_{s}\in\mathbb{N}$. As for $l_{s}$, note that by the
constraints:
\[
0\leq\mleft(\alpha_{s} + \sum_{i\neq s}\epsilon_{i}c_{i}\cdot2^{s-i}/c_{s}\mright)\cdot c_{s}\cdot\beta\cdot2^{t_{s}}=l_{s}.
\]
 Since clearly $l_{s}\in\mathbb{Z}$ as a sum of numbers in $\mathbb{Z}$,
we get $l_{s}\in\mathbb{N}$, and thus there is such a sequence $\ser{\epsilon}$.
We now show a few bounds on values involving the sequence $\ser{\epsilon}$,
which will help us during the rest of the proof. By the definition
of the sequence $\ser{\epsilon}$ and since for $x\leq1/2$ we have
$h\mleft(x\mright)\geq x\log\frac{1}{x}\geq x$, it holds that:
\begin{equation}
\sum_{i\leq K,i\neq s}\epsilon_{i}\leq\sum_{i\leq K,i\neq s}h\mleft(\epsilon_{i}\mright)\leq\sum_{i\leq K,i\neq s}\epsilon/K=\epsilon.\label{eq:small_epsilons_1_lower}
\end{equation}
 Moreover:
\begin{gather*}
\sum_{i\neq s}\epsilon_{i}c_{i}/2^{i}=\sum_{i\leq K,i\neq s}\epsilon_{i}c_{i}/2^{i}+\sum_{i>K}\epsilon_{i}c_{i}/2^{i}\leq\sum_{i\leq K,i\neq s}\epsilon_{i}+\sum_{i>K}c_{i}=2\epsilon,
\end{gather*}
 and thus: 
\begin{equation}
\epsilon_{s}=\sum_{i\neq s}\epsilon_{i}c_{i}\cdot2^{s-i}/c_{s}\leq2^{s}2\epsilon/c_{s}.\label{eq:small_epsilons_2_lower}
\end{equation}

\paragraph{$\protect\ser{\tilde{\alpha}}$ is feasible.}

Now we show that $\ser{\tilde{\alpha}}\in\mathcal{A}^{\leq K}$ and
$K'$-feasible for some $K'\in\mathbb{N}$. Based on the fact $\ser{\alpha}\in\mathcal{A}$,
we first show $\ser{\tilde{\alpha}}\in\mathcal{A}^{\leq K}$, that
is:
\begin{enumerate}
\item $\ser{\tilde{\alpha}}\in\mleft[0,1\mright]^{\mathbb{N}}$.
\item $\sum_{i=0}^{\infty}\tilde{\alpha}_{i}c_{i}/2^{i}=\frac{1}{\beta\cdot2^{b+1}}$.
\item $i>K\implies\tilde{\alpha}_{i}=0$ (this is obvious by the definition
of $\ser{\tilde{\alpha}}$).
\end{enumerate}
We show $\mleft(1\mright)$: For $i\neq s$, it is not hard to check
that $0\leq\tilde{\alpha}_{i}\leq1$ by the definition of the sequence
$\ser{\epsilon}$. For $i=s$, recall that $0\leq\alpha_{s}<3/4$
and thus $0\leq\tilde{\alpha}_{s}=\alpha_{s}+\epsilon_{s}\leq1$ for
small enough $\epsilon$ (due to (\ref{eq:small_epsilons_2_lower})).
Let us show $\mleft(2\mright)$, depending on the fact $\sum_{i=0}^{\infty}\alpha_{i}c_{i}/2^{i}=\frac{1}{\beta\cdot2^{b+1}}$:
 
\begin{align*}
\sum_{i=0}^{\infty}\tilde{\alpha}_{i}c_{i}/2^{i} & =\mleft(\alpha_{s}+\epsilon_{s}\mright)c_{s}/2^{s}+\sum_{i\neq s}\mleft(\alpha_{i}-\epsilon_{i}\mright)c_{i}/2^{i}\\
 & =\sum_{i=0}^{\infty}\alpha_{i}c_{i}/2^{i}+\epsilon_{s}c_{s}/2^{s}-\sum_{i\neq s}\epsilon_{i}c_{i}/2^{i}\\
 & =\frac{1}{\beta\cdot2^{b+1}}+\frac{\sum_{i\neq s}\epsilon_{i}c_{i}\cdot2^{s-i}}{c_{s}}c_{s}/2^{s}-\sum_{i\neq s}\epsilon_{i}c_{i}/2^{i}=\frac{1}{\beta\cdot2^{b+1}}.
\end{align*}
 Thus, indeed $\ser{\tilde{\alpha}}\in A^{\leq K}$. Now we show that
$\ser{\tilde{\alpha}}$ is $K'$-feasible where 
\[
K'=\max\mleft(\mleft\{ K\mright\} \cup\mleft\{ t_{i}:i\in I\mright\} \mright).
\]
 That is:
\begin{enumerate}
\item For any $i\in\mathbb{N}$: $\tilde{\alpha}_{i}c_{i}n\in\mathbb{N}$
where $n=\beta\cdot2^{K'}$. 
\item If there is $t$ such that $c_{t}>0$ and $c_{i}=0$ for any $i>t$,
then $\alpha_{t}<1$.
\end{enumerate}
We show $\mleft(1\mright)$: If $i>K$, or $c_{i}=0$, or $\alpha_{i}=0$
then by definition of $\ser{\tilde{\alpha}}$, $\tilde{\alpha}_{i}c_{i}n=0\in\mathbb{N}$.
Otherwise: 
\[
\tilde{\alpha}_{i}c_{i}n=\frac{l_{i}}{c_{i}\cdot\beta\cdot2^{t_{i}}}c_{i}\cdot\beta\cdot2^{K'}=l_{i}\cdot2^{K'-t_{i}}\in\mathbb{N}
\]
 since $K'\geq t_{i}$. Let us show $\mleft(2\mright)$: If $t>K$ or
$\alpha_{t}=0$ then $\tilde{\alpha}_{t}=0<1$. Otherwise, if $t\leq K$
and $t\neq s$ then $\epsilon_{t}>0$ and thus 
\[
\tilde{\alpha}_{t}=\alpha_{t}-\epsilon_{t}\leq1-\epsilon_{t}<1.
\]
 If $t=s$, then since $\alpha_{s}<3/4$, we have $\tilde{\alpha}_{t}=\alpha_{s}+\epsilon_{s}<1$
for small enough $\epsilon$. So $\ser{\tilde{\alpha}}\in\mathcal{A}^{\leq K}\subseteq\mathcal{A}^{\leq K'}$
and is $K'$-feasible, as required.

\paragraph{$P\mleft(\protect\ser c,\protect\ser{\tilde{\alpha}}\mright)$ approximates
$P\mleft(\protect\ser c,\protect\ser{\alpha}\mright)$.}

It only remains to show $P\mleft(\ser c,\ser{\alpha}\mright)=P\mleft(\ser c,\ser{\tilde{\alpha}}\mright)\pm\epsilon$.
Due to Lemma \ref{lem:entropy_bounds}, We have: 
\begin{align*}
P\mleft(\ser c,\ser{\alpha}\mright) & =\sum_{i=0}^{\infty}h\mleft(\alpha_{i}\mright)c_{i}-h\mleft(\sum_{i=0}^{\infty}\alpha_{i}c_{i}\mright)\\
 & =h\mleft(\tilde{\alpha}_{s}-\epsilon_{s}\mright)c_{s}+\sum_{i\neq s}h\mleft(\tilde{\alpha}_{i}+\epsilon_{i}\mright)c_{i}-h\mleft(\mleft(\tilde{\alpha}_{s}-\epsilon_{s}\mright)c_{s}+\sum_{i\neq s}\mleft(\tilde{\alpha}_{i}+\epsilon_{i}\mright)c_{i}\mright)\\
 & =\sum_{i=0}^{\infty}h\mleft(\tilde{\alpha}_{i}\mright)c_{i}-h\mleft(\sum_{i=0}^{\infty}\tilde{\alpha}_{i}c_{i}\mright)\pm h\mleft(\sum_{i=0}^{\infty}\epsilon_{i}c_{i}\mright)\pm\sum_{i=0}^{\infty}h\mleft(\epsilon_{i}\mright)c_{i}\\
 & =P\mleft(\ser c,\ser{\tilde{\alpha}}\mright)\pm h\mleft(\sum_{i=0}^{\infty}\epsilon_{i}c_{i}\mright)\pm\sum_{i=0}^{\infty}h\mleft(\epsilon_{i}\mright)c_{i}.
\end{align*}
 We show that the expressions $h\mleft(\sum_{i=0}^{\infty}\epsilon_{i}c_{i}\mright)$
and $\sum_{i=0}^{\infty}h\mleft(\epsilon_{i}\mright)c_{i}$ are small.
It holds that:
\begin{align*}
h\mleft(\sum_{i=0}^{\infty}\epsilon_{i}c_{i}\mright) & \underset{\text{Lem \ref{lem:entropy_bounds}}}{\leq}h\mleft(\sum_{i=0}^{K}\epsilon_{i}c_{i}\mright)+h\mleft(\sum_{i=K+1}^{\infty}\epsilon_{i}c_{i}\mright)\\
& \leq h\mleft(\sum_{i=0}^{K}\epsilon_{i}\mright)+h\mleft(\sum_{i=K+1}^{\infty}c_{i}\mright)  \underset{\eqref{eq:small_epsilons_1_lower},\eqref{eq:small_epsilons_2_lower}}{\leq}h\mleft(\epsilon+2^{s}2\epsilon/c_{s}\mright)+h\mleft(\epsilon\mright),
\end{align*}
 and: 
\begin{align*}
\sum_{i=0}^{\infty}h\mleft(\epsilon_{i}\mright)c_{i}= & \sum_{i=0}^{K}h\mleft(\epsilon_{i}\mright)c_{i}+\sum_{i=K+1}^{\infty}h\mleft(\epsilon_{i}\mright)c_{i}\\
& \leq h\mleft(\epsilon_{s}\mright)+\sum_{i\leq K,i\neq s}h\mleft(\epsilon_{i}\mright)+\sum_{i=K+1}^{\infty}c_{i}\underset{\eqref{eq:small_epsilons_1_lower},\eqref{eq:small_epsilons_2_lower}}{\leq}h\mleft(2^{s}2\epsilon/c_{s}\mright)+2\epsilon.
\end{align*}
 Thus, we can choose $\epsilon'>0$ small enough and apply the proof
for $\epsilon'$, such that $P\mleft(\ser c,\ser{\alpha}\mright)=P\mleft(\ser c,\ser{\tilde{\alpha}}\mright)\pm\epsilon$. 
\end{proof}

\subsubsection{Upper bounding $\rho_{\min}^{*}\mleft(n\mright)$}

The idea here is similar to the idea of the lower bound proof. Here
$\mathcal{C}^{\leq l}$ is the set of all pairs $\mleft(\ser c,b\mright)\in\mathcal{C}$
such that if $i>l$ then $c_{i}=0$. In order to prove Lemma \ref{lem:upper_bound},
we will prove two claims. The first allows us to remove or add constraints
on the choice of a pair $\mleft(\ser c,b\mright)\in\mathcal{C}$ without
changing much the value of $P\mleft(\ser c,\ser{\alpha}\mright)$: 
\begin{lemma}
\label{lem:upper_bound_1} It holds that: 
\[
\lim_{k\rightarrow\infty}\min_{\substack{
\mleft(\ser c,b\mright)\in\mathcal{C}^{\leq k}\colon\\
\ser c\text{ is \ensuremath{k\text{-feasible}}}
}}\max_{\ser{\alpha}\in\mathcal{A}}P\mleft(\ser c,\ser{\alpha}\mright)=G\mleft(\beta\mright).
\]
 
\end{lemma}
The second claim shows, essentially, that the summation appearing
in $\rho_{\min}^{*}\mleft(n\mright)$ is redundant for approximation
up to sub-exponential factors, if the pair $\mleft(\ser c,b\mright)$
chosen by the minimization belongs to $\mathcal{C}^{\leq k}$: 
\begin{lemma}
\label{lem:redundant_sum} Let $n=\beta\cdot2^{k}$ and $\mleft(\ser c,b\mright)\in\mathcal{C}^{\leq k}$
such that $\ser c$ is $k$-feasible. Then: 
\[
\max_{d\in\mleft[n\mright]}\sum_{\ser{\alpha}\in S_{d}^{\leq k}}\frac{\prod_{i=0}^{\infty}\binom{c_{i}n}{\alpha_{i}c_{i}n}}{\binom{n}{d}}\leq\exp_{2}\mleft(\max_{\alpha\in\mathcal{A}}P\mleft(\ser c,\ser{\alpha}\mright)n+o\mleft(n\mright)\mright).
\]
\end{lemma}
Having Lemmas \ref{lem:upper_bound_1} and \ref{lem:redundant_sum},
we can prove Lemma \ref{lem:upper_bound}: 
\begin{proof}
[Proof of Lemma \ref{lem:upper_bound}] We have: 
\begin{align*}
\rho_{\min}^{*}\mleft(n\mright) & =\min_{\substack{
\mleft(\ser c,b\mright)\in\mathcal{C}\colon\\
\ser c\text{ is \ensuremath{k\text{-feasible}}}
}}\max_{\substack{
d\in\mleft[n\mright]\colon\\
S_{d}\neq\emptyset
}}\sum_{\ser{\alpha}\in S_{d}}\frac{\prod_{i=0}^{\infty}\binom{c_{i}n}{\alpha_{i}c_{i}n}}{\binom{n}{d}}\\
 & \leq\min_{\substack{
\mleft(\ser c,b\mright)\in\mathcal{C}^{\leq k}\colon\\
\ser c\text{ is \ensuremath{k\text{-feasible}}}
}}\max_{\substack{
d\in\mleft[n\mright]\colon\\
S_{d}\neq\emptyset
}}\sum_{\ser{\alpha}\in S_{d}^{\leq k}}\frac{\prod_{i=0}^{\infty}\binom{c_{i}n}{\alpha_{i}c_{i}n}}{\binom{n}{d}}\\
 & \underset{\text{Lem \ref{lem:redundant_sum}}}{\leq}\exp_{2}\mleft(\min_{\substack{
\mleft(\ser c,b\mright)\in\mathcal{C}^{\leq k}\colon\\
\ser c\text{ is \ensuremath{k\text{-feasible}}}
}}\max_{\alpha\in\mathcal{A}}P\mleft(\ser c,\ser{\alpha}\mright)n+o\mleft(n\mright)\mright)\\
 & \underset{\text{Lem \ref{lem:upper_bound_1}}}{\leq}\exp_{2}\mleft(\inf_{\mleft(\ser c,b\mright)\in\mathcal{C}}\max_{\alpha\in\mathcal{A}}P\mleft(\ser c,\ser{\alpha}\mright)n+o\mleft(n\mright)\mright)=2^{G\mleft(\beta\mright)n+o\mleft(n\mright)}.\qedhere
\end{align*}
\end{proof}
Now we shall prove Lemmas \ref{lem:upper_bound_1}--\ref{lem:redundant_sum}.
We prove Lemma \ref{lem:redundant_sum} first since it is simpler. 
\begin{proof}
[Proof of Lemma \ref{lem:redundant_sum}] Let $n=\beta\cdot2^{k}$
and $\mleft(\ser c,b\mright)\in\mathcal{C}^{\leq k}$ such that $\ser c$
is $k$-feasible. Let $\ser{\alpha}\in\bigcup_{d\in\mleft[n\mright]}S_{d}^{\leq k}$
such that 
\begin{equation}
\max_{\substack{
\ser{\alpha'}\in\bigcup_{d\in\mleft[n\mright]}S_{d}^{\leq k}}}\frac{\prod_{i=0}^{\infty}\binom{c_{i}n}{\alpha'_{i}c_{i}n}}{\binom{n}{\sum_{i=0}^{\infty}\alpha'_{i}c_{i}n}}=\frac{\prod_{i=0}^{\infty}\binom{c_{i}n}{\alpha_{i}c_{i}n}}{\binom{n}{ \sum_{i=0}^{\infty}\alpha_{i}c_{i}n}}.\label{eq:max_alpha}
\end{equation}
Combinatorial considerations imply that 
\begin{equation}
\mleft|\bigcup_{d\in\mleft[n\mright]}S_{d}^{\leq k}\mright|\leq\mleft(n+1\mright)^{k+1}=O\mleft(n^{\log n+1}\mright)\label{eq:redundant_sum}
\end{equation}
 since for a sequence $\ser{\alpha'}\in\bigcup_{d\in\mleft[n\mright]}S_{d}^{\leq k}$,
if $0\leq i\leq k$ then $\alpha'_{i}n$ can potentially be any number
between $0$ and $n$, and else $\alpha'_{i}n=0$. Hence: 
\begin{align*}
\max_{d\in\mleft[n\mright]}\sum_{\ser{\alpha'}\in S_{d}^{\leq k}}\frac{\prod_{i=0}^{\infty}\binom{c_{i}n}{\alpha'_{i}c_{i}n}}{\binom{n}{d}} & \leq\sum_{\ser{\alpha'}\in\bigcup_{d\in\mleft[n\mright]}S_{d}^{\leq k}}\frac{\prod_{i=0}^{\infty}\binom{c_{i}n}{\alpha'_{i}c_{i}n}}{\binom{n}{\sum_{i=0}^{\infty}\alpha'_{i}c_{i}n}}\\
 & \underset{\eqref{eq:max_alpha},\eqref{eq:redundant_sum}}{\leq}O\mleft(n^{\log n+1}\mright)\cdot\frac{\prod_{i=0}^{\infty}\binom{c_{i}n}{\alpha_{i}c_{i}n}}{\binom{n}{\sum_{i=0}^{\infty}\alpha_{i}c_{i}n}}\\
 & \underset{\eqref{eq:binom_bounds}}{\leq}O\mleft(n^{\log n+1}\mright)\cdot\frac{\exp_{2}\mleft(  \sum_{i=0}^{\infty} \mleft(h\mleft(\alpha_{i}\mright)c_{i}n\mright) \mright)}{\exp_{2}\mleft(h\mleft(\sum_{i=0}^{\infty}\alpha_{i}c_{i}\mright)n\mright)/O\mleft(\sqrt{n}\mright)}\\
 & =\exp_{2}\mleft(P\mleft(\ser c,\ser{\alpha}\mright)n+o\mleft(n\mright)\mright)\underset{\mleft(*\mright)}{\leq}\exp_{2}\mleft(\max_{\ser{\alpha'}\in\mathcal{A}}P\mleft(\ser c,\ser{\alpha'}\mright)n+o\mleft(n\mright)\mright)
\end{align*}
where $\mleft(*\mright)$ is since $\bigcup_{d\in\mleft[n\mright]}S_{d}^{\leq k}\subseteq\mathcal{A}$
by definition, and hence $\ser{\alpha}\in\mathcal{A}$. 
\end{proof}
The proof of Lemma \ref{lem:upper_bound_1} is implied by the following:
\begin{lemma}
\label{lem:upper_bound_2}Fix $\mleft(\ser c,b\mright)\in\mathcal{C}$
and let $\epsilon>0$. There is $K\in\mathbb{N}$ and $\mleft(\ser{\tilde{c}},b\mright)\in\mathcal{C}^{\leq K}$,
where $\ser{\tilde{c}}$ is $K$-feasible, such that for any $\ser{\tilde{\alpha}}\in\mleft[0,1\mright]^{\mathbb{N}}$
satisfying $\sum_{i=0}^{\infty}\tilde{\alpha}{}_{i}\tilde{c}_{i}/2^{i}=\frac{1}{\beta\cdot2^{b+1}}$
there is $\ser{\alpha}\in\mleft[0,1\mright]^{\mathbb{N}}$ satisfying
$\sum_{i=0}^{\infty}\alpha_{i}c_{i}/2^{i}=\frac{1}{\beta\cdot2^{b+1}}$,
and $P\mleft(\ser{\tilde{c}},\ser{\tilde{\alpha}}\mright)=P\mleft(\ser c,\ser{\alpha}\mright)\pm\epsilon$. 
\end{lemma}
Having Lemma \ref{lem:upper_bound_2} in hand, we can prove Lemma
\ref{lem:upper_bound_1}:
\begin{proof}
[Proof of Lemma \ref{lem:upper_bound_1}] Let $\epsilon>0$ and
$\epsilon'=\epsilon/3$. Let $\mleft(\ser c,b\mright)\in\mathcal{C}$
which satisfies: 
\begin{equation}
\max_{\ser{\alpha}\in\mathcal{A}}P\mleft(\ser c,\ser{\alpha}\mright)\leq G\mleft(\beta\mright)+\epsilon'.\label{eq:upper_bound_1}
\end{equation}
 Let $K\in\mathbb{N}$ large enough such that Lemma \ref{lem:upper_bound_2}
holds for $\mleft(\ser c,b\mright)$ and $\epsilon'$. Let $\ser{\tilde{\alpha}}\in\mleft[0,1\mright]^{\mathbb{N}}$
which satisfies $\sum_{i=0}^{\infty}\tilde{\alpha}_{i}\tilde{c}/2^{i}=\frac{1}{\beta\cdot2^{b+1}}$
and: 
\begin{equation}
\max_{\ser{\alpha}\in\mathcal{A}}P\mleft(\ser{\tilde{c}},\ser{\alpha}\mright)\leq P\mleft(\ser{\tilde{c}},\ser{\tilde{\alpha}}\mright)+\epsilon'.\label{eq:upper_bound_2}
\end{equation}
 Hence for any $k\geq K$: 
\begin{align*}
\min_{\substack{
\mleft(\ser c,b\mright)\in\mathcal{C}^{\leq k}\\
\ser c\text{ is \ensuremath{k\text{-feasible}}}
}}\max_{\ser{\alpha}\in\mathcal{A}}P\mleft(\ser c,\ser{\alpha}\mright) & \underset{\mleft(*\mright)}{\leq}\max_{\ser{\alpha}\in\mathcal{A}}P\mleft(\ser{\tilde{c}},\ser{\alpha}\mright)\\
 & \underset{\eqref{eq:upper_bound_2}}{\leq}P\mleft(\ser{\tilde{c}},\ser{\tilde{\alpha}}\mright)+\epsilon'\\
 & \underset{\text{Lem \ref{lem:upper_bound_2}}}{\leq}P\mleft(\ser c,\ser{\alpha}\mright)+2\epsilon'\\
 & \underset{\mleft(**\mright)}{\leq}\max_{\ser{\alpha}\in\mathcal{A}}P\mleft(\ser c,\ser{\alpha}\mright)+2\epsilon'\underset{\eqref{eq:upper_bound_1}}{\leq}G\mleft(\beta\mright)+3\epsilon'
\end{align*}
 where $\mleft(*\mright)$ is since $\mleft(\ser{\tilde{c}},b\mright)\in\mathcal{C}^{\leq k}$
and $\ser{\tilde{c}}$ is $k$-feasible, and $\mleft(**\mright)$ is
since $\ser{\alpha}\in\mleft[0,1\mright]^{\mathbb{N}}$ and satisfies
$\sum_{i=0}^{\infty}\alpha_{i}c_{i}/2^{i}=\frac{1}{\beta\cdot2^{b+1}}$.
Since $3\epsilon'=\epsilon$, and since obviously 
\[
\inf_{\mleft(\ser c,b\mright)\in\mathcal{C}}\max_{\ser{\alpha}\in\mathcal{A}}P\mleft(\ser c,\ser{\alpha}\mright)\leq\min_{\substack{
\mleft(\ser c,b\mright)\in\mathcal{C}^{\leq k}\\
\ser c\text{ is \ensuremath{k\text{-feasible}}}
}}\max_{\ser{\alpha}\in\mathcal{A}}P\mleft(\ser c,\ser{\alpha}\mright),
\]
the lemma follows. 
\end{proof}
It remains to prove Lemma \ref{lem:upper_bound_2}. The proof idea
is similar to the idea appearing in the proof of Lemma \ref{lem:low_bound_2}
presented in the previous subsection. Hence, we only present a detailed
proof for this Lemma (without a proof sketch): 
\begin{proof}
[Proof of Lemma \ref{lem:upper_bound_2}] Fix $\mleft(\ser c,b\mright)\in\mleft[0,1\mright]^{\mathbb{N}}$
satisfying $\sum_{i=0}^{\infty}c_{i}/2^{i}=\frac{1}{\beta\cdot2^{b}}$
and $\sum_{i=0}^{\infty}c_{i}\leq1$ (that is, $\mleft(\ser c,b\mright)\in\mathcal{C}$).
Consider two different cases. First, assume that $\ser c$ is the
following sequence: 
\[
c_{i}=\begin{cases}
1 & i=0,\\
0 & i\neq0.
\end{cases}
\]
 It is possible since the equation $\sum_{i=0}^{\infty}c_{i}/2^{i}=1=\frac{1}{\beta\cdot2^{b}}$
holds whenever $\beta=1,b=0$. In that case, $\ser c\in\mathcal{C}^{\leq0}$
and $\ser c$ is $0$-feasible, thus the lemma follows with $K=0$,
$\ser{\tilde{c}}=\ser c$. So, assume now that $c_{0}<1$. We divide
the proof under that assumption into four parts: First we define $\ser{\tilde{c}}$.
Then we show that $\ser{\tilde{c}}\in\mathcal{C}^{\leq K'}$ and $\ser{\tilde{c}}$
is $K'$-feasible for some $K'\in\mathbb{N}$. After that, given $\ser{\tilde{\alpha}}$
we define $\ser{\alpha}$ and show $\sum_{i=0}^{\infty}\alpha_{i}c_{i}/2^{i}=\frac{1}{\beta\cdot2^{b+1}}$.
Finally, we show $P\mleft(\ser{\tilde{c}},\ser{\tilde{\alpha}}\mright)=P\mleft(\ser c,\ser{\alpha}\mright)\pm\epsilon$. 

\paragraph{Defining $\protect\ser{\tilde{c}}$.}

Recall that $\sum_{i=0}^{\infty}c_{i}\leq1$ and thus there is $K\in\mathbb{N}$
such that $\sum_{i>K}c_{i}<\epsilon$. Let $\ser{\tilde{c}}$ be the
following sequence: 
\[
\tilde{c}_{i}=\begin{cases}
c_{0}+\epsilon_{0} & i=0,\\
c_{i}-\epsilon_{i} & i\neq0,
\end{cases}
\]
 where $\ser{\epsilon}$ is an arbitrary sequence satisfying the following
constraints:
\begin{enumerate}
\item If $i>K$: $\epsilon_{i}=c_{i}$.
\item Otherwise, if $1\leq i\leq K$: $\epsilon_{i}=c_{i}-\frac{l_{i}}{\beta\cdot2^{t_{i}}}$,
where $l_{i},t_{i}\in\mathbb{N}$ and $0\leq\epsilon_{i}\leq\epsilon/K$.
\item $\epsilon_{0}=\sum_{i=1}^{\infty}\epsilon_{i}/2^{i}=\frac{l_{0}}{\beta\cdot2^{t_{0}}}-c_{0}$,
where $l_{0},t_{0}\in\mathbb{N}$. 
\end{enumerate}
In order to continue with the proof, we first have to show that such
a sequence $\ser{\epsilon}$ exists. It is not hard to construct a
sequence that satisfies constraints $\mleft(1\mright),\mleft(2\mright)$.
We should satisfy constraint $\mleft(3\mright)$ as well, that is: 
\begin{align*}
\frac{l_{0}}{\beta\cdot2^{t_{0}}}-c_{0} & =\sum_{i=1}^{\infty}\epsilon_{i}/2^{i}\\
 & =\sum_{i=1}^{K}\mleft(c_{i}-\frac{l_{i}}{\beta\cdot2^{t_{i}}}\mright)/2^{i}+\sum_{i>K}c_{i}/2^{i}=\sum_{i=1}^{\infty}c_{i}/2^{i}-\sum_{i=1}^{K}\frac{l_{i}}{\beta\cdot2^{t_{i}+i}}.
\end{align*}
That can be written as: 
\[
\sum_{i=0}^{\infty}c_{i}/2^{i}-\sum_{i=1}^{K}\frac{l_{i}}{\beta\cdot2^{t_{i}+i}}=\frac{l_{0}}{\beta\cdot2^{t_{0}}}.
\]
Recall that since $\mleft(\ser c,b\mright)\in\mathcal{C}$, we have
$\sum_{i=0}^{\infty}c_{i}/2^{i}=\frac{1}{\beta\cdot2^{b}}$, and thus
we get: 
\begin{gather*}
\frac{1}{\beta\cdot2^{b}}-\sum_{i=1}^{K}\frac{l_{i}}{\beta\cdot2^{t_{i}+i}}=\frac{l_{0}}{\beta\cdot2^{t_{0}}}\\
\iff\\
\frac{1}{2^{b}}-\sum_{i=1}^{K}\frac{l_{i}}{2^{t_{i}+i}}=\frac{l_{0}}{2^{t_{0}}}.
\end{gather*}
 We can find $l_{0},t_{0}\in\mathbb{N}$ satisfying this equation:
$t_{0}=b+\sum_{i=1}^{K}t_{i}+i$, and $l_{0}$ is determined accordingly.
Obviously, $t_{0}\in\mathbb{N}$. As for $l_{0}$, it is obviously
in $\mathbb{Z}$ as a sum of numbers in $\mathbb{Z}$. Moreover, by
the constraints: 
\[
l_{0}=\mleft(c_{0}+\sum_{i=1}^{\infty}\epsilon_{i}/2^{i}\mright)\beta\cdot2^{t_{0}}\geq0,
\]
 and thus $l_{0}\in\mathbb{N}$. We have satisfied all constraints,
thus we can find such a sequence $\ser{\epsilon}$. Let us now show
that the values of $\ser{\epsilon}$ are small, even if we sum all
of them together. That fact will help us show that the change of $\ser c$
to $\ser{\tilde{c}}$ has only little affect. 
\begin{gather}
\sum_{i=0}^{\infty}\epsilon_{i}=\sum_{i=1}^{\infty}\epsilon_{i}/2^{i}+\sum_{i=1}^{K}\epsilon_{i}+\sum_{i=K+1}^{\infty}\epsilon_{i}\leq2\sum_{i=1}^{K}\epsilon_{i}+2\sum_{i=K+1}^{\infty}\epsilon_{i}\leq2\sum_{i=1}^{K}\epsilon/K+2\epsilon=4\epsilon.\label{eq:small_epsilons_upper}
\end{gather}

\paragraph{$\protect\ser{\tilde{c}}$ is feasible.}

Now we will show that $\mleft(\ser{\tilde{c}},b\mright)\in\mathcal{C}^{\leq K}$
and $\ser{\tilde{c}}$ is $K'$-feasible for some $K'\geq K$. Based
on the fact that $\mleft(\ser c,b\mright)\in\mathcal{C}$, we first
show $\mleft(\ser{\tilde{c}},b\mright)\in\mathcal{C}^{\leq K}$, that
is: 
\begin{enumerate}
\item $\ser{\tilde{c}}\in\mleft[0,1\mright]^{\mathbb{N}}.$
\item $\sum_{i=0}^{\infty}\tilde{c}_{i}\leq1.$
\item $\sum_{i=0}^{\infty}\tilde{c}_{i}/2^{i}=\frac{1}{\beta\cdot2^{b}}.$
\item $i>K\implies\tilde{c_{i}}=0$ (this is obvious by the definition of
$\ser{\tilde{c}}$).
\end{enumerate}
We show $\mleft(1\mright)$: For any $i\neq0$ it is obvious that $0\leq\tilde{c_{i}}\leq1$
from the definition of $\ser{\tilde{c}}$. For $i=0$, $0\leq c_{0}<1$
and hence $0\leq\tilde{c_{0}}=c_{0}+\epsilon_{0}\leq1$ for small
enough $\epsilon$. Thus $\ser{\tilde{c}}\in\mleft[0,1\mright]^{\mathbb{N}}$.
Now we show $\mleft(2\mright)$: 
\begin{align*}
\sum_{i=0}^{\infty}\tilde{c}_{i} &= c_{0}+\sum_{i=1}^{\infty}\epsilon_{i}/2^{i}+\sum_{i=1}^{K}\mleft(c_{i}-\epsilon_{i}\mright)\\
&= \sum_{i=1}^{\infty}\epsilon_{i}/2^{i}+\sum_{i=0}^{K}c_{i}-\sum_{i=1}^{K}\epsilon_{i}\\
& \leq\sum_{i=0}^{K}c_{i}+\sum_{i=K+1}^{\infty}\epsilon_{i}\\
&= \sum_{i=0}^{\infty}c_{i}\leq1.
\end{align*}
 And finally $\mleft(3\mright)$: 

\begin{gather*}
\sum_{i=0}^{\infty}\tilde{c}_{i}/2^{i}=c_{0}+\epsilon_{0}+\sum_{i=1}^{\text{\ensuremath{\infty}}}\mleft(c_{i}-\epsilon_{i}\mright)/2^{i}=\sum_{i=0}^{\text{\ensuremath{\infty}}}c_{i}/2^{i}+\epsilon_{0}-\sum_{i=1}^{\text{\ensuremath{\infty}}}\epsilon_{i}/2^{i}=\sum_{i=0}^{\text{\ensuremath{\infty}}}c_{i}/2^{i}=\frac{1}{\beta\cdot2^{b}}.
\end{gather*}
 So indeed $\mleft(\ser{\tilde{c}},b\mright)\in\mathcal{C}^{\leq K}$.
We show that $\mleft(\ser{\tilde{c}},b\mright)$ is $K'$-feasible where
\[
K'=\max\mleft\{ K,t_{0},\dots,t_{K}\mright\} .
\]
 Let $i\in\mathbb{N}$. If $i>K$ then $\tilde{c}_{i}n=0\in\mathbb{N}$.
Else: 
\[
\tilde{c}_{i}n=\frac{l_{i}}{\beta\cdot2^{t_{i}}}\beta\cdot2^{K'}=l_{i}\cdot2^{K'-t_{i}}\in\mathbb{N},
\]
 since $K'\geq t_{i}$.

\paragraph{Defining $\protect\ser{\alpha}$.}

Given $\ser{\tilde{\alpha}}\in\mleft[0,1\mright]^{\mathbb{N}}$ such
that $\sum_{i=0}^{\infty}\tilde{\alpha}_{i}\tilde{c}_{i}/2^{i}=\frac{1}{\beta\cdot2^{b+1}}$,
we construct $\ser{\alpha}\in\mleft[0,1\mright]^{\mathbb{N}}$ that
``imitates'' $\ser{\tilde{\alpha}}$ and satisfies $\sum_{i=0}^{\infty}\alpha_{i}c_{i}/2^{i}=\frac{1}{\beta\cdot2^{b+1}}$.
For such a sequence $\ser{\tilde{\alpha}}$, consider the following
expression: 
\[
r\mleft(\ser{\tilde{\alpha}}\mright)=\sum_{i=1}^{\infty}\tilde{\alpha}_{i}\epsilon_{i}/2^{i}-\tilde{\alpha}_{0}\epsilon_{0}.
\]
 If $r\mleft(\ser{\tilde{\alpha}}\mright)\geq0$, let $s$ be the first
index such that $\tilde{c}_{s},\tilde{\alpha}_{s}>1/2^{2\mleft(b+5\mright)}$.
Else, let $s$ be the first index such that $\tilde{c}_{s}>1/2^{2\mleft(b+5\mright)}$
and $\tilde{\alpha}_{s}<3/4$. In any case, $s$ exists and is bounded
by $2^{b+4}$, by Lemma \ref{lem:s_index}. Define the sequence $\ser{\alpha}$
as follows: 
\[
\alpha_{i}=\begin{cases}
\tilde{\alpha}_{s}-\delta_{s} & i=s,\\
\tilde{\alpha}_{i} & i\neq s,
\end{cases}
\]
 where $\delta_{s}=\frac{2^{s}}{c_{s}}r\mleft(\ser{\tilde{\alpha}}\mright)$.
Note that $\mleft|\delta_{s}\mright|$ is small since $\mleft|r\mleft(\ser{\tilde{\alpha}}\mright)\mright|$
is small: 
\[
\mleft|r\mleft(\ser{\tilde{\alpha}}\mright)\mright|\leq\sum_{i=1}^{\infty}\tilde{\alpha}_{i}\epsilon_{i}/2^{i}+\tilde{\alpha}_{0}\epsilon_{0}\leq\sum_{i=0}^{\infty}\epsilon_{i}\underset{\eqref{eq:small_epsilons_upper}}{\leq}4\epsilon,
\]
and $s$ is bounded by a constant. We show that $\ser{\alpha}\in\mleft[0,1\mright]^{\mathbb{N}}$:
If $i\neq s$ then $0\leq\tilde{\alpha}_{i}=\alpha_{i}\leq1$. For
$i=s$, if $r\mleft(\ser{\tilde{\alpha}}\mright)\geq0$ then: 
\[
\alpha_{s}=\tilde{\alpha}_{s}-\delta_{s}>1/2^{2\mleft(b+5\mright)}-\delta_{s}>0
\]
 for small enough $\epsilon$ and obviously $\alpha_{s}=\tilde{\alpha}_{s}-\delta_{s}\leq\tilde{\alpha}_{s}\leq1$.
Otherwise, assume $r\mleft(\ser{\tilde{\alpha}}\mright)<0$, then: 
\[
\alpha_{s}=\tilde{\alpha}_{s}-\delta_{s}<3/4-\delta_{s}\leq1
\]
 for small enough $\epsilon$ and obviously $\alpha_{s}=\tilde{\alpha}_{s}-\delta_{s}\geq\tilde{\alpha}_{s}\geq0$.
Thus $\ser{\alpha}\in\mleft[0,1\mright]^{\mathbb{N}}$. We show that
$\sum_{i=1}^{\infty}\alpha_{i}c_{i}/2^{i}=\frac{1}{\beta\cdot2^{b+1}}$,
that is, $\sum_{i=1}^{\infty}\alpha_{i}c_{i}/2^{i}=\sum_{i=1}^{\infty}\tilde{\alpha}_{i}\tilde{c}_{i}/2^{i}$:
\begin{align*}
\sum_{i=1}^{\infty}\tilde{\alpha}_{i}\tilde{c}_{i}/2^{i} & =\tilde{\alpha}_{0}\mleft(c_{0}+\epsilon_{0}\mright)+\sum_{i=1}^{\infty}\tilde{\alpha}_{i}\mleft(c_{i}-\epsilon_{i}\mright)/2^{i}\\
 & =\sum_{i=0}^{\infty}\tilde{\alpha}_{i}c_{i}/2^{i}+\tilde{\alpha}_{0}\epsilon_{0}-\sum_{i=1}^{\infty}\tilde{\alpha}_{i}\epsilon_{i}/2^{i}\\
 & =\sum_{i=0}^{\infty}\alpha_{i}c_{i}/2^{i}+\delta_{s}c_{s}/2^{s}-r\mleft(\ser{\tilde{\alpha}}\mright)\\
 & =\sum_{i=0}^{\infty}\alpha_{i}c_{i}/2^{i}+\frac{2^{s}}{c_{s}}r\mleft(\ser{\tilde{\alpha}}\mright)\cdot\frac{c_{s}}{2^{s}}-r\mleft(\ser{\tilde{\alpha}}\mright)=\sum_{i=0}^{\infty}\alpha_{i}c_{i}/2^{i}.
\end{align*}

\paragraph{$P\mleft(\protect\ser c,\protect\ser{\alpha}\mright)$ approximates
$P\mleft(\protect\ser{\tilde{c}},\protect\ser{\tilde{\alpha}}\mright)$.}

It remains to show $P\mleft(\ser{\tilde{c}},\ser{\tilde{\alpha}}\mright)=P\mleft(\ser c,\ser{\alpha}\mright)\pm\epsilon$.
Due to Lemma \ref{lem:entropy_bounds}, indeed: 
\begin{align*}
P\mleft(\ser{\tilde{c}},\ser{\tilde{\alpha}}\mright) & =\sum_{i=0}^{\infty}h\mleft(\tilde{\alpha}_{i}\mright)\tilde{c}_{i}-h\mleft(\sum_{i=0}^{\infty}\tilde{\alpha}_{i}\tilde{c}_{i}\mright)\\
 & =h\mleft(\alpha_{s}+\delta_{s}\mright)\tilde{c}_{s}+\sum_{i\neq s}h\mleft(\alpha_{i}\mright)\tilde{c}_{i}-h\mleft(\sum_{i=0}^{\infty}\alpha_{i}\tilde{c}_{i}+\delta_{s}\tilde{c}_{s}\mright)\\
 & =\sum_{i=0}^{\infty}h\mleft(\alpha_{i}\mright)\tilde{c}_{i}-h\mleft(\sum_{i=0}^{\infty}\alpha_{i}\tilde{c}_{i}\mright)\pm h\mleft(\mleft|\delta_{s}\mright|\mright)\tilde{c}_{s}\pm h\mleft(\mleft|\delta_{s}\mright|\tilde{c}_{s}\mright)\\
 & =P\mleft(\ser{\tilde{c}},\ser{\alpha}\mright)\pm h\mleft(\mleft|\delta_{s}\mright|\mright)\tilde{c}_{s}\pm h\mleft(\mleft|\delta_{s}\mright|\tilde{c}_{s}\mright).
\end{align*}
 Recall that $\mleft|\delta_{s}\mright|$ is small. Now: 
\begin{align*}
P\mleft(\ser{\tilde{c}},\ser{\alpha}\mright) & =\sum_{i=0}^{\infty}h\mleft(\alpha_{i}\mright)\tilde{c}_{i}-h\mleft(\sum_{i=0}^{\infty}\alpha_{i}\tilde{c}_{i}\mright)\\
 & =\sum_{i=0}^{\infty}h\mleft(\alpha_{i}\mright)c_{i}+h\mleft(\alpha_{0}\mright)\epsilon_{0}-\sum_{i=1}^{\infty}h\mleft(\alpha_{i}\mright)\epsilon_{i}-h\mleft(\sum_{i=0}^{\infty}\alpha_{i}c_{i}+\alpha_{0}\epsilon_{0}-\sum_{i=1}^{\infty}\alpha_{i}\epsilon_{i}\mright)\\
 & =\sum_{i=0}^{\infty}h\mleft(\alpha_{i}\mright)c_{i}-h\mleft(\sum_{i=0}^{\infty}\alpha_{i}c_{i}\mright)+h\mleft(\alpha_{0}\mright)\epsilon_{0}-\sum_{i=1}^{\infty}h\mleft(\alpha_{i}\mright)\epsilon_{i}\pm h\mleft(\alpha_{0}\epsilon_{0}\mright)\pm h\mleft(\sum_{i=1}^{\infty}\alpha_{i}\epsilon_{i}\mright)\\
 & =P\mleft(\ser c,\ser{\alpha}\mright)+h\mleft(\alpha_{0}\mright)\epsilon_{0}-\sum_{i=1}^{\infty}h\mleft(\alpha_{i}\mright)\epsilon_{i}\pm h\mleft(\alpha_{0}\epsilon_{0}\mright)\pm h\mleft(\sum_{i=1}^{\infty}\alpha_{i}\epsilon_{i}\mright).
\end{align*}
Recall that $\sum_{i=0}^{\infty}\epsilon_{i}\leq4\epsilon$ due to
(\ref{eq:small_epsilons_upper}). Thus, we can choose $\epsilon'$
small enough and apply the proof for $\epsilon'$, such that 
\[
P\mleft(\ser{\tilde{c}},\ser{\tilde{\alpha}}\mright)=P\mleft(\ser{\tilde{c}},\ser{\alpha}\mright)\pm\epsilon'=P\mleft(\ser c,\ser{\alpha}\mright)\pm2\epsilon'=P\mleft(\ser c,\ser{\alpha}\mright)\pm\epsilon.\qedhere
\]
\end{proof}

\section{Applications of Theorem~\ref{theo:main} \label{sec:Applications}}

\subsection{Alternative proofs for known bounds on $q\mleft(n\mright)$}

In the previous sections we have shown the estimate $q\mleft(n\mright)=2^{-G\mleft(\beta\mright)n\pm o\mleft(n\mright)}$.
Unfortunately, we do not know how to calculate $G\mleft(\beta\mright)$
in general. However, we can use this estimate to give alternative
proofs for known bounds on $q\mleft(n\mright)$, and
in the next subsection, also to give a better lower bound. The following
theorems are stated and proved in \cite{DFGM2019}: 
\begin{theorem}[\cite{DFGM2019}]
\label{theo:upper_bound}For any $n$, it holds that $q\mleft(n\mright)\leq1.25^{n+o\mleft(n\mright)}$.
\end{theorem}
\begin{theorem}[\cite{DFGM2019}]
\label{theo:lower_bound} For $n=\beta\cdot2^{k}$ it holds that $q\mleft(n\mright)\geq2^{\mleft(h\mleft(\frac{1}{2^{b+1}\beta}\mright)-\frac{1}{2^{b}\beta}\mright)n-o\mleft(n\mright)}$
for any $b\in\mathbb{N}$. 
\end{theorem}
For any $\beta$, we can find a ``good" lower bound on $q\mleft(n\mright)$
by choosing $b=0$ or $b=1$ and applying $\mbox{Theorem \ref{theo:lower_bound}}$.
Specifically, when $\beta=1.25$ we get $q\mleft(n\mright)=1.25^{n\pm o\mleft(n\mright)}$
by choosing $b=1$, and for other values of $\beta$ we can always
ensure that $q\mleft(n\mright)\geq1.232^{n-o\mleft(n\mright)}$
by choosing $b=0$ or $b=1$, depending on $\beta$. We give alternative,
simple proofs for these bounds: 
\begin{proof}
[Proof of Theorem \ref{theo:upper_bound}] Fix $\beta\in\mleft[1,2\mright)$.
Let $\ser{\alpha}$ such that: $\alpha_{i}=1/2$ for any $i$. Note
that $\ser{\alpha}\in\mathcal{A}$ for any fixed $\mleft(\ser c,b\mright)\in\mathcal{C}$.
Thus: 
\begin{align*}
G\mleft(\beta\mright) & =\inf_{\mleft(\ser c,b\mright)\in\mathcal{C}}\max_{\ser{\alpha}\in\mathcal{A}}\sum_{i=0}^{\infty}h\mleft(\alpha_{i}\mright)c_{i}-h\mleft(\sum_{i=0}^{\infty}\alpha_{i}c_{i}\mright)\\
 & \geq\inf_{\mleft(\ser c,b\mright)\in\mathcal{C}}\sum_{i=0}^{\infty}h\mleft(1/2\mright)c_{i}-h\mleft(\sum_{i=0}^{\infty}\frac{1}{2}c_{i}\mright)\\
 & =\inf_{\mleft(\ser c,b\mright)\in\mathcal{C}}\sum_{i=0}^{\infty}c_{i}-h\mleft(\frac{1}{2}\sum_{i=0}^{\infty}c_{i}\mright).
\end{align*}
 Denote $x=\sum_{i=0}^{\infty}c_{i}$, so $0\leq x\leq1$, since $\mleft(\ser c,b\mright)\in\mathcal{C}$.
Hence : 
\[
\inf_{\mleft(\ser c,b\mright)\in\mathcal{C}}\sum_{i=0}^{\infty}c_{i}-h\mleft(\frac{1}{2}\sum_{i=0}^{\infty}c_{i}\mright)\geq\inf_{0\leq x\leq1}x-h\mleft(x/2\mright).
\]
Calculation shows that: 
\[
\inf_{0\leq x\leq1}x-h\mleft(x/2\mright)=0.4-h\mleft(0.4/2\mright)=-\log1.25.
\]
That is, $G\mleft(\beta\mright)\geq-\log1.25$, and hence indeed: 
\[
q\mleft(n\mright)=2^{-G\mleft(\beta\mright)n\pm o\mleft(n\mright)}\leq2^{\log1.25n+o\mleft(n\mright)}=1.25^{n+o\mleft(n\mright)}.\qedhere
\]
\end{proof}

\begin{proof}
[Proof of Theorem \ref{theo:lower_bound}] Fix $\beta\in\mleft[1,2\mright)$.
Let $\mleft(\ser{c\mleft(b\mright)},b\mright)\in\mathcal{C}$ such that
$b\in\mathbb{N}$ and $\ser{c\mleft(b\mright)}$ is the following sequence:
\[
c\mleft(b\mright)_{i}=\begin{cases}
\frac{1}{2^{b}\beta} & i=0,\\
0 & i\neq0.
\end{cases}
\]
 Indeed $\mleft(\ser{c\mleft(b\mright)},b\mright)\in\mathcal{C}$ for
any $b$, since all constraints are satisfied:
\begin{itemize}
\item $b\in\mathbb{N}.$
\item We have $0\leq\frac{1}{2^{b}\cdot2}\leq\frac{1}{2^{b}\beta}\leq\frac{1}{2^{0}\cdot1}=1$
and hence $0\leq c\mleft(b\mright)_{i}\leq1$ for any $i$.
\item $\sum_{i=0}^{\infty}c\mleft(b\mright)_{i}=\frac{1}{2^{b}\beta}\leq1$.
\item $\sum_{i=0}^{\infty}c\mleft(b\mright)_{i}2^{b-i}\cdot\beta=\frac{1}{2^{b}\beta}\cdot2^{b}\beta=1.$ 
\end{itemize}
Moreover, a sequence $\ser{\alpha}\in\mathcal{A}$ must satisfy $\alpha_{0}=1/2$,
and for any other $i$ the value of $\alpha_{i}$ does not have any
effect. Thus: 
\begin{align*}
G\mleft(\beta\mright) & =\inf_{\mleft(\ser c,b\mright)\in\mathcal{C}}\max_{\ser{\alpha}\in\mathcal{A}}\sum_{i=0}^{\infty}h\mleft(\alpha_{i}\mright)c_{i}-h\mleft(\sum_{i=0}^{\infty}\alpha_{i}c_{i}\mright)\\
 & \leq\max_{\ser{\alpha}\in\mathcal{A}}\sum_{i=0}^{\infty}h\mleft(\alpha_{i}\mright)c\mleft(b\mright)_{i}-h\mleft(\sum_{i=0}^{\infty}\alpha_{i}c\mleft(b\mright)_{i}\mright)\\
 & =h\mleft(1/2\mright)\cdot\frac{1}{2^{b}\beta}-h\mleft(\frac{1}{2}\cdot\frac{1}{2^{b}\beta}\mright)=\frac{1}{2^{b}\beta}-h\mleft(\frac{1}{2^{b+1}\beta}\mright).
\end{align*}
 Hence: 
\[
q\mleft(n\mright)=2^{-G\mleft(\beta\mright)n\pm o\mleft(n\mright)}\geq2^{\mleft(h\mleft(\frac{1}{2^{b+1}\beta}\mright)-\frac{1}{2^{b}\beta}\mright)n-o\mleft(n\mright)}.\qedhere
\]
\end{proof}

\subsection{\label{subsec:A-new-lower}A new lower bound on $q\mleft(n\mright)$}

Using our estimate $q\mleft(n\mright)=2^{-G\mleft(\beta\mright)n\pm o\mleft(n\mright)}$
we can find a tighter lower bound on $q\mleft(n\mright)$
than the one appearing in \cite{DFGM2019}, that is, $1.232^{n-o\mleft(n\mright)}$.
We do that by finding a matching upper bound on $G\mleft(\beta\mright)$.
We already know that $G\mleft(\beta\mright)\leq\frac{1}{2^{b}\beta}-h\mleft(\frac{1}{2^{b+1}\beta}\mright)$
for any $b\in\mathbb{N}$, as described in our alternative proof for
the known lower bound on $q\mleft(n\mright)$. For $\beta\leq1.7$
and $b=1$ we have $G\mleft(\beta\mright)\leq\frac{1}{2\cdot1.7}-h\mleft(\frac{1}{4\cdot1.7}\mright)\approx-0.3083$
and for $\beta\geq1.95$ we have $G\mleft(\beta\mright)\leq\frac{1}{1.95}-h\mleft(\frac{1}{2\cdot1.95}\mright)\approx-0.30846$.
So, if we find $M>-0.3083$ such that for $1.7<\beta<1.95$: $G\mleft(\beta\mright)\leq M$,
then $M$ is an upper bound for $G\mleft(\beta\mright)$. As we will
now show, it is possible for $M\approx-0.305758$. The idea is to
fix $b=1$ and consider sequences $\ser c$ in which $s=c_{0}+c_{1}$
is fixed and $c_{i}=0$ for all $i\geq2$. Then, due to the constraint
$\sum_{i=0}^{\infty}c_{i}/2^{i}=1/2\beta$ we can express $c_{0}$
and $c_{1}$ in terms of $\beta$. Finally, we use Lagrange multipliers
to find the maximizing $\ser{\alpha}$ for $\ser c$. 
\begin{theorem} \label{thm:1.236}
For any $n\in\mathbb{N}$, it holds that $q\mleft(n\mright)\geq1.236^{n-o\mleft(n\mright)}$. 
\end{theorem}
\begin{proof}
Consider the sequence $\ser c$ defined by $c_{0}=1/\beta-s$, $c_{1}=2s-1/\beta$
and $c_{i}=0$ for all $i\geq2$, for some fixed $s$. It is feasible:
\[
\sum_{i=0}^{\infty}c_{i}/2^{i}=1/\beta-s+\mleft(2s-1/\beta\mright)/2=1/2\beta,
\]
 as required. We calculate $\max_{\ser{\alpha}\in\mathcal{A}}P\mleft(\ser c,\ser{\alpha}\mright)$
using Lagrange multipliers. Our only constraint is $\sum_{i=0}^{\infty}\alpha_{i}c_{i}/2^{i}-1/4\beta=0$,
so we get the Lagrangian function 
\[
\mathcal{L}\mleft(\alpha_{0},\alpha_{1},\lambda\mright)=h\mleft(\alpha_{0}\mright)c_{0}+h\mleft(\alpha_{1}\mright)c_{1}-h\mleft(\alpha_{0}c_{0}+\alpha_{1}c_{1}\mright)+\lambda\mleft(\alpha_{0}c_{0}+\alpha_{1}c_{1}/2-1/4\beta\mright).
\]
Recall that $h'\mleft(x\mright)=\log\frac{1-x}{x}$ and compute the
derivatives: 
\begin{align} \label{eq:d1}
\frac{d\mathcal{L}\mleft(\alpha_{0},\alpha_{1},\lambda\mright)}{d\alpha_{0}} &= c_{0}\log\frac{1-\alpha_{0}}{\alpha_{0}}-c_{0}\log\frac{1-\alpha_{0}c_{0}-\alpha_{1}c_{1}}{\alpha_{0}c_{0}+\alpha_{1}c_{1}}+\lambda c_{0}\nonumber\\
&= c_{0}\log\frac{\mleft(1-\alpha_{0}\mright)\mleft(\alpha_{0}c_{0}+\alpha_{1}c_{1}\mright)}{\alpha_{0}\mleft(1-\alpha_{0}c_{0}-\alpha_{1}c_{1}\mright)}+\lambda c_{0}=0,
\end{align}

\begin{align} \label{eq:d2}
\frac{d\mathcal{L}\mleft(\alpha_{0},\alpha_{1},\lambda\mright)}{d\alpha_{1}} &= c_{1}\log\frac{1-\alpha_{1}}{\alpha_{1}}-c_{1}\log\frac{1-\alpha_{0}c_{0}-\alpha_{1}c_{1}}{\alpha_{0}c_{0}+\alpha_{1}c_{1}}+\lambda c_{1}/2\nonumber\\
&= c_{1}\log\frac{\mleft(1-\alpha_{1}\mright)\mleft(\alpha_{0}c_{0}+\alpha_{1}c_{1}\mright)}{\alpha_{1}\mleft(1-\alpha_{0}c_{0}-\alpha_{1}c_{1}\mright)}+\lambda c_{1}/2=0,
\end{align}

\begin{equation}
\frac{d\mathcal{L}\mleft(\alpha_{0},\alpha_{1},\lambda\mright)}{d\lambda}=\alpha_{0}c_{0}+\alpha_{1}c_{1}/2-1/4\beta=0.\label{eq:d3}
\end{equation}
We assume $c_{0},c_{1}>0$, so equations \eqref{eq:d1},\eqref{eq:d2}
can be written as 
\begin{gather*}
\log\frac{\mleft(1-\alpha_{0}\mright)\mleft(\alpha_{0}c_{0}+\alpha_{1}c_{1}\mright)}{\alpha_{0}\mleft(1-\alpha_{0}c_{0}-\alpha_{1}c_{1}\mright)}=-\lambda,\\
\log\mleft(\frac{\mleft(1-\alpha_{1}\mright)\mleft(\alpha_{0}c_{0}+\alpha_{1}c_{1}\mright)}{\alpha_{1}\mleft(1-\alpha_{0}c_{0}-\alpha_{1}c_{1}\mright)}\mright)^{2}=-\lambda,
\end{gather*}
 so we get: 
\begin{gather*}
\frac{1-\alpha_{0}}{\alpha_{0}}\cdot\frac{\alpha_{0}c_{0}+\alpha_{1}c_{1}}{1-\alpha_{0}c_{0}-\alpha_{1}c_{1}}=\frac{\mleft(1-\alpha_{1}\mright)^{2}}{\alpha_{1}^{2}}\cdot\mleft(\frac{\alpha_{0}c_{0}+\alpha_{1}c_{1}}{1-\alpha_{0}c_{0}-\alpha_{1}c_{1}}\mright)^{2}\\
\iff\\
\frac{1-\alpha_{0}}{\alpha_{0}}=\frac{\mleft(1-\alpha_{1}\mright)^{2}}{\alpha_{1}^{2}}\cdot\frac{\alpha_{0}c_{0}+\alpha_{1}c_{1}}{1-\alpha_{0}c_{0}-\alpha_{1}c_{1}}
\end{gather*}
 since $\alpha_{0}c_{0}+\alpha_{1}c_{1}>0$. Together with equation
(\ref{eq:d3}), we can solve two equations with two unknowns and find
$\alpha_{0},\alpha_{1}$. Calculation shows that there is only one
real solution, which is a critical point, and we will deduce that
it is also a maximum point. First, we have to fix $s$. We used a
software program to try and find a good value for $s$. It must hold
that $c_{1}=2s-1/\beta>0$, that is, $s>\frac{1}{2\beta}$, implying
$s>5/17$ (since $\beta>1.7$). Iteratively checking all feasible
values of $s$ with gaps of $1/1000$, it seems that $s=829/2000$
is a good choice. For that value of $s$, the maximal value of $P\mleft(\ser c,\ser{\alpha}\mright)$
(where $\ser{\alpha}$ is defined by the critical point solution found
for $\alpha_{0},\alpha_{1}$) is $\sim-0.305758$, attained at $\beta\approx1.80941$.
It remains to check that the critical point solution found for $\alpha_{0},\alpha_{1}$
is indeed a maximum point. Since in our critical point $\alpha_{0},\alpha_{1}\notin\mleft\{ 0,1\mright\} $
for any $\beta$, we can do that by checking the values attained when
$\alpha_{0}\in\mleft\{ 0,1\mright\} $ or $\alpha_{1}\in\mleft\{ 0,1\mright\} $
and verify they are lower than the value attained at our critical
point. Calculation shows that when choosing such points we get a lower
value for $P\mleft(\ser c,\ser{\alpha}\mright)$ for any $\beta$, compared
to the value attained at the critical point. Thus the critical point
we have found is indeed a maximum point and $G\mleft(\beta\mright)\leq-0.305758$
for any $1\leq\beta<2$. That is, $q\mleft(n\mright)\geq2^{0.305758n-o\mleft(n\mright)}>1.236^{n-o\mleft(n\mright)}$. 
\end{proof}

Figure~\ref{fig:newbound_vs_old} demonstrates our new bound for different values of $\beta$.

\begin{figure}[!ht]
    \centering
    \textbf{Upper bounds on $G\mleft(\beta\mright)$ for $\beta \geq 1.5$}\par\medskip
    \includegraphics[scale=0.8]{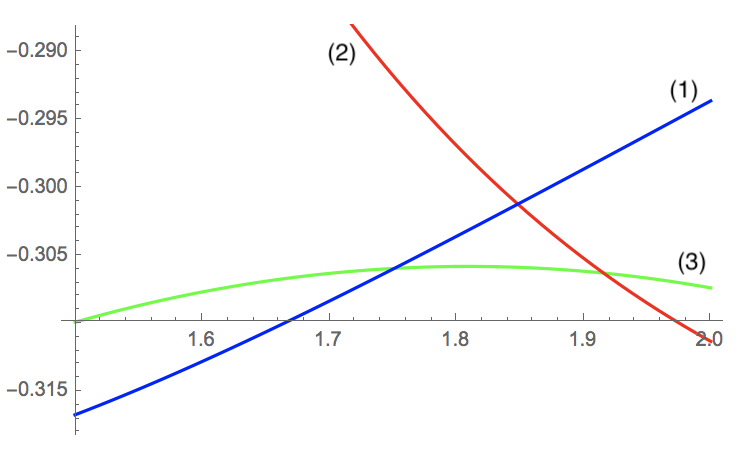}
    \caption{The blue (1) and red (2) curves are the known upper bounds $\frac{1}{2\beta}-h\mleft(\frac{1}{4\beta}\mright)$ and $\frac{1}{\beta}-h\mleft(\frac{1}{\beta}\mright)$, respectively. Our new upper bound is the green (3) curve, which is better in a range of $\beta$ values.}
    \label{fig:newbound_vs_old}
\end{figure}

\subsection{Improving the bound $q\mleft(n\mright)\leq 1.25^{n+o\mleft(n\mright)}$ for $n$ far from $1.25\cdot 2^k$}

Using our formula for $G\mleft(\beta\mright)$, it can be shown that when $\beta\neq 1.25$, the upper bound of $q\mleft(n\mright)\leq 1.25^{n+o\mleft(n\mright)}$ can be exponentially improved. This result is formally stated as follows:

\begin{theorem} \label{theo:better_upper}
If $\beta = 5/4+\delta_{\beta}$ such that $\mleft|\delta_{\beta}\mright|>0$, then there exists $\gamma>1$ such that $q\mleft(n\mright)\leq \mleft(1.25/\gamma\mright)^{n+o\mleft(n\mright)}$. Furthermore, $\gamma=2^{\Omega\mleft(\mleft|\delta_{\beta}\mright|^{2+\epsilon}\mright)}$, where $\epsilon>0$ is any fixed constant of our choice.
\end{theorem}

We prove the theorem using a sequence of steps. First, we show that sequences $\ser{c}\in \mathcal{C}$ with $\sum_{i=0}^{\infty} c_i$ which is far from $2/5$ can be ruled out as witnesses for $q\mleft(n\mright)= 1.25^{n \pm o\mleft(n\mright)}$:

\begin{lemma} \label{lem:no_triv_sum}
Let $\ser{c} \in \mathcal{C}$ and $x=\sum_{i=0}^\infty c_i$. If $x=2/5 \pm \delta$, then $\max_{\ser{\alpha}\in \mathcal{A}}P\mleft(\ser{c},\ser{\alpha}\mright)\geq -\log \mleft(5/4\mright)+\Omega\mleft(\delta^2\mright)$.
\end{lemma}

\begin{proof}
Define $\ser{\alpha}\in \mathcal{A}$ that assigns $\alpha_i = 1/2$ for all $i \in \mathbb{N}$, which is always a feasible choice. Then:
\[
 P(\ser{c},\ser{\alpha}) = x - h(x/2).
\]
The function $x - h(x/2)$ attains its unique minimum at $x = 2/5$, at which point its value is $-\log(5/4)$. By a linear approximation around $2/5$,
\[
x-h\mleft(x/2\mright)=-\log 1.25+\Omega\mleft(\delta^2\mright). \qedhere
\]
\end{proof}

The next technical lemma is required for the argument used in the proof of the theorem.

\begin{lemma} \label{lem:better_upper_tech}
For every $\delta_\beta \neq 0$ the following holds for $\delta_0=\mleft|\delta_{\beta}\mright|/100$.

Let $\beta=5/4+\delta_{\beta}$ and let $\ser{c}\in \mathcal{C}$ with $x:=\sum_{i=0}^{\infty} c_i=2/5 \pm \delta \geq \frac{1}{\beta 2^b}$, where $0 \leq \delta \leq \delta_0$. Then there exists $I\in \mathbb{N}$ such that:
\begin{gather*}
x - \frac{2^{I-b}}{\beta}=\Omega\mleft(\mleft|\delta_{\beta}\mright|\mright), \\
x - \frac{2^{I+1-b}}{\beta}=-\Omega\mleft(\mleft|\delta_{\beta}\mright|\mright).
\end{gather*}
\end{lemma}

\begin{proof}
Let $I\in \mathbb{N}$ be the largest value for which $x - \frac{2^{I-b}}{\beta} \geq 0$. There must be such $I$: clearly, there exists $I'$ for which $x - \frac{2^{i-b}}{\beta}<0$ for any $i\geq I'$. Moreover, by assumption, when choosing $I=0$ we have $x - \frac{2^{I-b}}{\beta}=x-\frac{1}{\beta 2^b} \geq 0$.\par
Having that, the correctness of the lemma boils down to whether the expression $\mleft|x - \frac{2^{l}}{\beta}\mright|$ might be very small when $x$ is close enough to $2/5$ and $\beta$ is far enough from $1.25$. The answer is negative: if $l\leq-2$ then $2^l/\beta\leq \frac{1}{4\beta}\leq 1/4$ and hence $\mleft|x - \frac{2^{l}}{\beta}\mright|=\Omega\mleft(1\mright)$ even for a constant $\delta$. On the other hand, if $l\geq 0$ then $2^l/\beta \geq 1/\beta \geq 1/2$, so in this case too $\mleft|x - \frac{2^{l}}{\beta}\mright|=\Omega\mleft(1\mright)$ even for a constant $\delta$. So, the only difficult value of $l$ is $l=-1$. In this case:
\[
x - \frac{2^{l}}{\beta}
=
2/5 \pm \delta -  \frac{1}{2\mleft(1.25 + \delta_{\beta}\mright)}
=
\frac{\delta_{\beta}\mleft(8 \pm 20\delta\mright) \pm 25\delta}{25 + 25\delta_{\beta}}
=
\pm \Omega\mleft(\mleft|\delta_{\beta}\mright|\mright). \qedhere
\] 

\end{proof}

Now we are ready to prove the theorem.
\begin{proof}
[Proof of Theorem~\ref{theo:better_upper}]
Let $\epsilon>0$ be a fixed constant. We will show that $G(\beta) \geq -\log(5/4) + \Omega\mleft(\mleft|\delta_{\beta}\mright|^{2+\epsilon}\mright)$, thus implying the result stated in the theorem. We show that for all $b \in \mathbb{N}$ and $\ser{c} \in [0,1]^{\mathbb{N}}$ satisfying $\sum_{i=0}^\infty c_i/2^i = \frac{1}{\beta \cdot 2^b}$ and $\sum_{i=0}^\infty c_i \leq 1$, we can find $\ser{\alpha} \in [0,1]^{\mathbb{N}}$ such that $\sum_{i=0}^\infty \alpha_i c_i/2^i = \frac{1}{\beta \cdot 2^{b+1}}$ and $P(\ser{c},\ser{\alpha}) \geq -\log(5/4) + \Omega\mleft(\mleft|\delta_{\beta}\mright|^{2+\epsilon}\mright)$.\par

Let $\delta_0=\mleft|\delta_{\beta}\mright|/100$ and denote $x=\sum_{i=0}^{\infty} c_i=2/5\pm \delta$. By Lemma~\ref{lem:no_triv_sum}, if $\mleft|x-2/5\mright|\geq\delta_0$, then the theorem follows even with $\epsilon=0$. Hence we can assume, from now on, that $|x - 2/5| < \delta_0$.\par

Let $S \subseteq \mathbb{N}$, let $T = \mathbb{N} \setminus S$, and let $\eta_S,\eta_T \in [-1,1]$ be two parameters small in magnitude. Define
\begin{align*}
    p_S &= \sum_{i \in S} c_i/2^i, &
    q_S &= \sum_{i \in S} c_i, \\
    p_T &= \sum_{i \in T} c_i/2^i, &
    q_T &= \sum_{i \in T} c_i.
\end{align*}
Consider the assignment
\[
 \alpha_i = \begin{cases}
 \frac{1}{2} + \eta_S & \text{if } i \in S, \\
 \frac{1}{2} + \eta_T & \text{if } i \in T.
 \end{cases}
\]
Since $\sum_{i=0}^\infty \frac{1}{2} \cdot c_i/2^i = \frac{1}{\beta \cdot 2^{b+1}}$, this assignment is feasible if
\[
 \eta_S p_S + \eta_T p_T = 0.
\]
We assume henceforth that $\eta = \max(\mleft|\eta_S\mright|,\mleft|\eta_T\mright|) \leq \eta_0$, where $\eta_0 = O\mleft(\mleft|\delta_{\beta}\mright|^{1+\epsilon}\mright)$. By construction, we have $h(\alpha_i) = 1 - O(\eta^2)$ using a linear approximation. In contrast,
\[
 \sum_{i=0}^\infty \alpha_i c_i = \frac{x}{2} + \eta_S q_S + \eta_T q_T.
\]
Since $|x-2/5| < \delta_0$, this shows that
\[
 h\left(\sum_{i=0}^\infty \alpha_i c_i\right) =
 h(x/2) + (\eta_S q_S + \eta_T q_T) h'(x/2) \pm O(\eta^2)
\]
using a linear approximation. Overall, this shows that
\[
 P(\ser{c},\ser{\alpha}) = x - h(x/2) + (\eta_S q_S + \eta_T q_T) h'(x/2) \pm O\mleft(\mleft|\delta_{\beta}\mright|^{2+2\epsilon}\mright).
\]
Moreover, we have $h'(x/2) = \Omega(1)$ since $x < 2/5 + \delta < 1/2$. Since $x - h(x/2) \geq -\log(5/4)$, it suffices to show that there exist $\eta_S,\eta_T$ which are bounded in magnitude by $\eta_0$ such that $|\eta_S q_S + \eta_T q_T| = \Omega(\mleft|\delta_{\beta}\mright|^{2+\epsilon})$ (if $\eta_S q_S + \eta_T q_T<0$, we simply negate $\eta_S,\eta_T$).

Suppose $p_S \geq p_T$, or equivalently $p_S \geq (p_S + p_T)/2 = \frac{1}{\beta \cdot 2^{b+1}}$. Then the condition $\eta_S p_S + \eta_T p_T = 0$ shows that $\eta_S = -\frac{p_T}{p_S} \eta_T$, and so
\[
 \eta_S q_S + \eta_T q_T =
 \left(q_T - \frac{p_T}{p_S} q_S \right) \eta_T.
\]
Since $|\eta_S| \leq |\eta_T|$, if $|q_T - \frac{p_T}{p_S} q_S| = \Omega(\mleft|\delta_{\beta}\mright|)$ then by choosing $\eta_T = \Omega\mleft(\mleft|\delta_{\beta}\mright|^{1+\epsilon}\mright)$ we would be done.

Recall that $p_S + p_T = \frac{1}{\beta \cdot 2^b}$ and $q_S + q_T = x$. Therefore
\[
 q_T - \frac{p_T}{p_S} q_S =
 x - q_S - \frac{p_T}{p_S} q_S =
 x - \frac{1}{\beta \cdot 2^b} \frac{q_S}{p_S}.
\]
Our goal therefore is to find a set $S$ such that the following two conditions hold:
\begin{gather*}
    p_S \geq \frac{1}{\beta \cdot 2^{b+1}}, \\
    \left| x - \frac{1}{\beta \cdot 2^b} \frac{q_S}{p_S} \right| = \Omega(\mleft|\delta_{\beta}\mright|).
\end{gather*}

Let $S_{\leq I}=\mleft\{0,\dots,I\mright\}$ and $S_{>I}=\mleft\{I+1,\dots\mright\}$. Let us first assume that the choice of $S$ to be $S_{\leq 0}$ yields 
\[
x - \frac{1}{\beta \cdot 2^b} \frac{q_{S_{\leq 0}}}{p_{S_{\leq 0}}}
=
x - \frac{1}{\beta \cdot 2^b}
<
0.  
\]
So, it must hold that $b=1$: notice that 
\[
\frac{1}{2^{b+1}}
\leq
\frac{1}{\beta 2^b}
=
\sum_{i=0}^{\infty} c_i/2^i
\leq
x
=
2/5 \pm \delta,
\]
so if $b=0$ it is a contradiction. On the other hand, if $b\geq2$ then
\[
0 
>
x - \frac{1}{\beta \cdot 2^b}
\geq
2/5 \pm \delta-\frac{1}{4\beta}
\geq
2/5 \pm \delta-\frac{1}{4}
\]
is a contradiction. So, we get that
\[
0
>
x - \frac{1}{\beta \cdot 2^b} \frac{q_{S_{\leq 0}}}{p_{S_{\leq 0}}}
=
x-\frac{1}{2\beta}
=
2/5 \pm \delta -\frac{2}{5 + 4\delta_{\beta}}, 
\]
implying that $\delta_{\beta}<0$, and hence indeed $2/5+\delta -\frac{2}{5+4\delta_{\beta}}=-\Omega(\mleft|\delta_{\beta}\mright|)$. Notice that
\[
\frac{q_{S_{\leq i}}}{p_{S_{\leq i}}}
\geq
\frac{q_{S_{\leq 0}}}{p_{S_{\leq 0}}}
=
1
\]for any $i$, so we choose $S$ to be $S_{\leq I}$ such that $p_{S_{\leq I}} \geq \frac{1}{\beta \cdot 2^{b+1}}$ and then both conditions hold.\par
Suppose now that the choice of $S$ to be $S_{\leq 0}$ yields
\[
x - \frac{1}{\beta \cdot 2^b} \frac{q_{S_0}}{p_{S_0}}
=
x - \frac{1}{\beta \cdot 2^b}
\geq
0  
\]
and let $I\in \mathbb{N}$ be the one picked by Lemma~\ref{lem:better_upper_tech}. By construction, one of $S_{\leq I},S_{>I}$ satisfies the first condition. As for the second condition,
\begin{gather*}
x - \frac{1}{\beta \cdot 2^b} \frac{q_{S_{\leq I}}}{p_{S_{\leq I}}} \geq
x - \frac{1}{\beta \cdot 2^b} 2^{I}
=
x - \frac{2^{I-b}}{\beta}
=\Omega(\mleft|\delta_{\beta}\mright|),\\
x - \frac{1}{\beta \cdot 2^b} \frac{q_{S_{>I}}}{p_{S_{>I}}} \leq
x - \frac{1}{\beta \cdot 2^b} 2^{I+1}
=
x - \frac{2^{I+1-b}}{\beta}
=
-\Omega(\mleft|\delta_{\beta}\mright|),
\end{gather*}
by Lemma~\ref{lem:better_upper_tech}, and so it is satisfied for both $S_{\leq I},S_{>I}$.
\end{proof}

\section{$d$-ary questions} \label{sec:d-ary}
In this section we generalize some results on $q\mleft(n\mright)$ appearing in \cite{DFGM2019} to the $d$-ary setting, in which each question has $d$ possible answers (instead of only ``Yes" or ``No"). In this setting, a set of allowed questions $\mathcal{Q}$ contains a collection of partitions of $X_n$ to $d$ distinguished subsets $\mleft(S_i\mright)_{i\in \mleft[d\mright]}$. We denote the natural generalization of $q\mleft(n\mright)$ to the $d$-ary setting with $q^{\mleft(d\mright)}\mleft(n\mright)$. That is, $q^{\mleft(d\mright)}\mleft(n\mright)$ is the minimal size of a set of allowed questions $\mathcal{Q}$ that allows Alice to construct an optimal strategy for any distribution on $X_n$ picked by Bob.\par
We present two results in this section. The first states that for any $d=o\mleft(n/\log^{2}n\mright)$, it holds that $q^{\mleft(d\mright)}\mleft(n\mright)<2^{n+o\mleft(n\mright)}$; this improves exponentially on the trivial upper
bound $q^{\mleft(d\mright)}\mleft(n\mright)\leq d^{n}/d!$.
The second result is that for any fixed $d$, the upper bound we have just mentioned is tight up
to sub-exponential factors for infinitely many $n$ values.\par
In the binary setting, our results on $q\mleft(n\mright)$ rely on the reduction of~\cite{DFGM2019} from calculating $q\mleft(n\mright)$ to calculating $\rho_{\min}\mleft(n\mright)$, that is, on the fact that $q\mleft(n\mright)\approx 1/\rho_{\min}\mleft(n\mright)$ (Theorem~\ref{theo:reduc}). We take here the same approach: define $\rho^{\mleft(d\mright)}_{\min}\mleft(n\mright)$ to be the natural generalization of $\rho_{\min}\mleft(n\mright)$ to the $d$-ary setting (a formal definition appears later). We will show that $q^{\mleft(d\mright)}\mleft(n\mright)\approx 1/\rho^{\mleft(d\mright)}_{\min}\mleft(n\mright)$, and then we find bounds on $\rho^{\mleft(d\mright)}_{\min}\mleft(n\mright)$ to derive bounds on $q^{\mleft(d\mright)}\mleft(n\mright)$. Most of the lemmas in this section are simple generalizations of those appearing in \cite{DFGM2019}.\par
Let us first generalize some basic notions from the standard binary setting. All logarithms have base $d$, unless written otherwise.
\begin{definition}
A distribution $\mu$ is $d$-adic if every element with non-zero probability
in $\mu$ has probability $d^{-\ell}$ for some positive integer $\ell$. 
\end{definition}

%\paragraph{Valid $n$ values.}
%\textcolor{red}{Notice that defining a $d$-adic distribution over $X_n$ requires that $n = 1 \mod (d-1)$. Throughout this section, we assume that this is the case. Nevertheless, the bounds we prove  holds also when this is not the case: $q(n)$ is obviously increasing. Therefore, for every $n' \geq n$, an upper bound on $q(n')$ is also an upper bound on $q(n)$. For every $n$, we can choose $n'$ so that $n \leq n' \leq n + d - 1$. This will hold since we assume that $d = o(n)$, and do not analyze sub-exponential factors.}

\paragraph{$d$-ary search trees.}\label{sec:dary}
In the $d$-ary setting, similarly to the standard binary setting, a strategy to reveal the secret element $x$ is represented by a search tree. The difference is that in the $d$-ary setting, we use $d$-ary search trees (instead of binary search trees, namely, decision trees): each internal node, representing
a question, has $d$ outgoing edges, representing the possible answers.\par
However, if $n=\mleft|X_n\mright|$ is not equivalent
to $1$ modulo $d-1$, then such a tree cannot be constructed.
%\textcolor{red}{As explained above, our bounds allow assuming that $n = 1 \mod (d-1)$.}
So, if that is the case, we add a minimal set $X_l$ of zero probability
elements, such that $n+l$ is equivalent to $1$ modulo $d-1$. A
$d$-ary search tree can now be successfully constructed for $X_n\cup X_l$.
For our convenience, we still consider $X_n$ to be the set of elements (rather than $X_n\cup X_l$): note that $l<d$, and thus if we assume that $d$ is an
asymptotically small enough function of $n$, this has no effect on
the results in this section (in particular, we do not care about
sub-exponential factors in our estimates). Indeed, we have to limit
the discussion only for $d=o\mleft(n/\mleft(\log n \log \log n\mright)\mright)$, for other reasons,
even when $n$ is equivalent to $1$ modulo $d-1$, so this issue
has no meaning in our work.
\par
Definitions and notations from Section~\ref{sec:Preliminaries} regarding decision trees naturally generalize to $d$-ary search trees.

\paragraph{$d$-ary Huffman algorithm.}
Similarly to the binary case, if $\mu$ is a distribution over $X_n$, then the $d$-ary version of Huffman's algorithm finds a $d$-adic distribution $\tau$ that defines a search tree $T$ with $T\mleft(x_i\mright)=\log\frac{1}{\tau_{i}}$ for any non-zero element, such that the cost of $T$ on $\mu$, which is
\[
T\mleft(\mu\mright)=\sum_{i=1}^{n}\mu_{i}\log\frac{1}{\tau_{i}}=\sum_{i=1}^{n}\mu_{i}\log\frac{1}{\mu_{i}}+\sum_{i=1}^{n}\mu_{i}\log\frac{\mu_{i}}{\tau_{i}}=H\mleft(\mu\mright)+D\mleft(\mu\|\tau\mright)
\]
(where $D\mleft(\mu\|\tau\mright)=\sum_{i=1}^n\mu_i \log\mleft(\mu_i/\tau_i\mright)$ is the \emph{Kullback--Leibler divergence}), is optimal. This implies the inequality $T\mleft(\mu\mright)\geq H\mleft(\mu\mright)$ due to non-negativity of $D\mleft(\mu\|\tau\mright)$. It holds as equality when $\mu$ is $d$-adic.

\paragraph{The chain rule of conditional entropy.}
Let $S=\mleft(S_{j}\mright)_{j\in\mleft[d\mright]}$ be a partition of
$X_n$ into $d$ sets, and let $\mu$ be a distribution
over $X_n$. Let $M$ be a random variable drawn from $\mu$,
and let $P$ be a random variable indicating the set in $S$ that
$M$ belongs to. The probability distribution of $P$ is the distribution
$\pi$, defined by $\pi_{j}=\sum_{i\in D_{j}}\mu_{i}$ for any $j\in\mleft[d\mright]$.
The chain rule of conditional entropy states that: 
\[
H\mleft(M\mright)=H\mleft(P\mright)+H\mleft(M|P\mright),
\]
 where 
\[
H\mleft(M|P\mright)=\sum_{p}\Pr\mleft[P=p\mright]\cdot H\mleft[M|P=p\mright].
\]
 We will use it in the following equivalent form: 
\[
H\mleft(\mu\mright)=H\mleft(\pi\mright)+\sum_{j=1}^{d}\pi_{j}H\mleft(\mu|_{S_{j}}\mright).
\]

\paragraph{The multinomial coefficient.}
Let $n\in \mathbb{N}$ and $k_1,\dots,k_d\in \mathbb N$ such that $\sum_{i=1}^d k_d=n$. Let $\pi$ be the induced distribution defined by $\pi_i=k_i/n$ for any $1\leq i\leq d$. We will use the following known bounds on the multinomial coefficient (see \cite{csiszar2004information}, for example):
\begin{equation}\label{eq:multi_bounds}
\frac{1}{O\mleft(n^d\mright)} 2^{n H\mleft(\pi\mright)} \leq \binom{n}{k_1,k_2,\dots,k_d}\leq 2^{n H\mleft(\pi\mright)}.
\end{equation}

In the following subsections we show the reduction $q^{\mleft(d\mright)}\mleft(n\mright)\approx 1/\rho^{\mleft(d\mright)}_{\min}\mleft(n\mright)$, then we upper and lower bound $\rho^{\mleft(d\mright)}_{\min}\mleft(n\mright)$, and finally prove the two main results of this section.

\subsection{Reduction to $d$-adic hitters}

First we state the following reduction.
\begin{lemma}
\label{lem:red_dadic_hitters}A set $\mathcal{Q}$ of questions is
optimal if and only if $c\mleft(\mathcal{Q},\mu\mright)=\Opt\mleft(\mu\mright)$
for all $d$-adic distributions $\mu$. 
\end{lemma}

\begin{proof}
Assume that $\mathcal{Q}$ is optimal for all $d$-adic distributions
and let $\pi$ be some arbitrary distribution. Let $\mu$ be a $d$-adic
distribution such that: 
\[
\Opt\mleft(\pi\mright)=\sum_{i=1}^{n}\pi_{i}\log\frac{1}{\mu_{i}}.
\]
 Let $T$ be an optimal decision tree for $\mu$ using $\mathcal{Q}$
only, and let $\tau$ be the corresponding $d$-adic distribution,
that is $\tau_{i}=d^{-T\mleft(x_{i}\mright)}$. Since $\tau$ minimizes
$H\mleft(\mu\mright)+D\mleft(\mu\|\tau\mright)$, $\tau=\mu$ must hold.
Hence: 
\[
T\mleft(\pi\mright)=\sum_{i=1}^{n}\pi_{i}\log\frac{1}{\tau_{i}}=\sum_{i=1}^{n}\pi_{i}\log\frac{1}{\mu_{i}}=\Opt\mleft(\pi\mright).\qedhere
\]
\end{proof}
Now we define the notion of \emph{$d$-adic hitters}.
\begin{definition}
If $\mu$ is a non-constant $d$-adic distribution, we say that a partition $\mleft(S_{i}\mright)_{i\in\mleft[d\mright]}$
of $X_n$ \emph{divides} $\mu$ if $\mu\mleft(S_{i}\mright)=1/d$
for any $i\in\mleft[d\mright]$. The collection of all such partitions
of $X_n$ is denoted $\Div\mleft(\mu\mright)$. A set $\Div\mleft(\mu\mright)$,
for some distribution $\mu$, is called a \emph{$d$-adic set}. A set of questions
$\mathcal{Q}$ is called a \emph{$d$-adic hitter} if it intersects
$\Div\mleft(\mu\mright)$ for all non-constant $d$-adic distributions
$\mu$. 
\end{definition}
Let us generalize the ``useful lemma'' appearing in \cite{DFGM2019}
for our usage:
\begin{lemma}
\label{lem:useful}Let $d\in\mathbb{N}$ and let $p_{1}\geq\dots\geq p_{n}$
be a non-increasing list of numbers of the form $p_{i}=d^{-a_{i}}$,
where $a_{i}\in\mathbb{N}$, and let $a\in\mathbb{N}$ be such that
$a\leq a_{1}$. If $\sum_{i=1}^{n}p_{i}\geq d^{-a}$ then for some
$m$ we have $\sum_{i=1}^{m}p_{i}=d^{-a}$. If furthermore $\sum_{i=1}^{n}p_{i}=l\cdot d^{-a}$
for some $l\in\mathbb{N}$ then we can divide $\mleft[n\mright]$ to
$l$ intervals $\mleft(I_{j}\mright)_{j\in\mleft[l\mright]}$ such that
for any interval $I_{j}\subset\mleft[n\mright]$ we have $\sum_{i\in I_{j}}p_{i}=d^{-a}$.
\end{lemma}
\begin{proof}
Let $m$ be the maximal index such that $\sum_{i=1}^{m}p_{i}\leq d^{-a}$.
If $m=n$ then we are done, so suppose that $m<n$. Let $S=\sum_{i=1}^{m}p_{i}$.
We would like to show that $S=d^{-a}$. The condition $p_{1}\geq\dots\geq p_{n}$
implies that $a_{m+1}\geq\dots\geq a_{1}$, and so $k=d^{a_{m+1}}S=\sum_{i=1}^{m}d^{a_{m+1}-a_{i}}$
is an integer. By assumption, $k\leq d^{a_{m+1}-a}$, whereas $k+1=d^{a_{m+1}}\sum_{i=1}^{m+1}d^{-a_{i}}>d^{a_{m+1}-a}$.
Since $d^{a_{m+1}-a}\in\mathbb{N}$ (since $a_{m+1}\geq a_{1}\geq a$),
we conclude that $k=d^{a_{m+1}-a}$, and so $S=d^{-a}$. 

To prove the furthermore part, notice that by repeated applications
of the first part of the lemma we can partition $\mleft[n\mright]$
into intervals with probabilities $d^{-a}$. 
\end{proof}

Among else, this lemma shows (by choosing $a=1$) that $\Div\mleft(\mu\mright)$
is non-empty for any non-constant $d$-adic $\mu$. 

\begin{lemma} \label{lem:dadic_hit_reduc}
A set $\mathcal{Q}$ of partitions of $X_n$ to $d$ subsets
is an optimal set of questions if and only if it is a $d$-adic hitter
in $X_n$. 
\end{lemma}
\begin{proof}
Let $\mathcal{Q}$ be a $d$-adic hitter in $X_n$, and
let $\mu$ be a $d$-adic distribution. We show by induction on the
support size $m\leq n$ that $c\mleft(\mathcal{Q},\mu\mright)=H\mleft(\mu\mright)$.
Recall that $\Opt\mleft(\mu\mright)=H\mleft(\mu\mright)$, and thus due
to Lemma \ref{lem:red_dadic_hitters} optimality of $\mathcal{Q}$
will follow. The base case $m=1$ is trivial. So, suppose that $m>1$
and hence $\mu$ is non-constant, and therefore $\mathcal{Q}$ contains
a partition $D=\mleft(D_{i}\mright)_{i\in\mleft[d\mright]}\in \Div\mleft(\mu\mright)$.
Since $D\in \Div\mleft(\mu\mright)$, it holds that $\mu|_{D_{i}}$ is
$d$-adic for all $i\in\mleft[d\mright]$. The induction hypothesis
implies $c\mleft(\mathcal{Q},\mu|_{D_{i}}\mright)=H\mleft(\mu|_{D_{i}}\mright)$
for all $i\in\mleft[d\mright]$. Having that, let us calculate $H\mleft(\mu\mright)$.
Let $\pi$ be the distribution defined by $\pi_{i}=\mu\mleft(D_{i}\mright)=1/d$
for any $i\in\mleft[d\mright]$, so due to the chain rule of conditional
entropy and the induction hypothesis:
\begin{align*}
H\mleft(\mu\mright) & =H\mleft(\pi\mright)+\sum_{i=1}^{d}\pi_{i}H\mleft(\mu|_{D_{i}}\mright)=1+\sum_{i=1}^{d}\frac{1}{d}c\mleft(\mathcal{Q},\mu|_{D_{i}}\mright).
\end{align*}
 Now, consider the cost of a decision tree $T$ asking $D$, and then uses
the implied algorithms for $\mu|_{D_{1}},\dots,\mu|_{D_{d}}$, depending
on the answer for $D$: 
\[
T\mleft(\mu\mright)=1+\sum_{i=1}^{d}\mu\mleft(D_{i}\mright)\cdot c\mleft(\mathcal{Q},\mu|_{D_{i}}\mright)=1+\sum_{i=1}^{d}\frac{1}{d}c\mleft(\mathcal{Q},\mu|_{D_{i}}\mright)=H\mleft(\mu\mright),
\]
 and so $c\mleft(\mathcal{Q},\mu\mright)\leq H\mleft(\mu\mright)$, thus
$\mathcal{Q}$ is optimal. 

Conversely, suppose that $\mathcal{Q}$ is not a $d$-adic hitter,
so let $\mu$ be a non-constant $d$-adic distribution such that $\Div\mleft(\mu\mright)$
is disjoint from $\mathcal{Q}$. Consider an arbitrary decision tree $T$
for $\mu$ using $\mathcal{Q},$ and let $P=\mleft(P_{i}\mright)_{i\in\mleft[d\mright]}$
be its first question. Let also $\pi$ be the distribution defined
by $\pi_{i}=\mu\mleft(P_{i}\mright)$ for any $i\in\mleft[d\mright]$.
Then
\[
T\mleft(\mu\mright)\geq 1+\sum_{i=1}^{d}\pi_{i}\cdot c\mleft(\mathcal{Q},\mu|_{P_{i}}\mright)>H\mleft(\pi\mright)+\sum_{i=1}^{d}\pi_{i}\cdot H\mleft(\mu|_{P_{i}}\mright)=H\mleft(\mu\mright),
\]
 since there is $i$ such that $\pi_{i}\neq1/d$, otherwise it contradicts
$\mathcal{Q}$ and $\Div\mleft(\mu\mright)$ being disjoint, thus
$H\mleft(\pi\mright)<1$, and moreover $c\mleft(\mathcal{Q},\mu|_{P_{i}}\mright)\geq H\mleft(\mu|_{P_{i}}\mright)$. So the cost of any such arbitrary tree is
more than $H\mleft(\mu\mright)$, thus $\mathcal{Q}$ is not optimal. 
\end{proof}

\subsection{Reduction to maximum relative density}

Let us generalize the concept of maximum relative density defined in Section~\ref{sec:Preliminaries}.
\begin{definition}
Let $\mathcal{D}$ be a collection of partitions $D=\mleft(D_{i}\mright)_{i\in\mleft[d\mright]}$
of $X_n$. Let $K$ be the set of all vectors $\overline{k}=\mleft(k_{1},k_{2},\dots,k_{d}\mright)\in\mleft\{ 0,\dots,n\mright\} ^{d}$
such that $\sum_{i=1}^{d}k_{i}=n$. For $\overline{k}\in K$, denote
by $\mathcal{D}_{\overline{k}}\subset\mathcal{D}$ the restriction
of $\mathcal{D}$ to partitions with $\mleft|D_{i}\mright|=k_{i}$ for
all $i\in\mleft[d\mright]$. We say that each such vector $\overline{k}\in K$
is a \emph{type} of partitions, as this usage is similar to the concept
of types in the theory of types. Define $\overline{k}$\emph{'s relative
density} of $\mathcal{D}$, denoted $\rho_{\overline{k}}\mleft(\mathcal{D}\mright)$,
as 
\[
\rho_{\overline{k}}\mleft(\mathcal{D}\mright):=\frac{\mleft|\mathcal{D}_{\overline{k}}\mright|}{\binom{n}{k_{1},k_{2},\dots,k_{d}}}.
\]
 We define the \emph{maximum relative density} of $\mathcal{D}$,
denoted $\rho\mleft(\mathcal{D}\mright)$, as 
\[
\rho\mleft(\mathcal{D}\mright):=\max_{\overline{k}\in K}\rho_{\overline{k}}\mleft(\mathcal{D}\mright).
\]
 
\end{definition}
Define $\rho^{\mleft(d\mright)}_{\min}\mleft(n\mright)$ to be the minimal $\rho\mleft(\mathcal{D}\mright)$
over all $d$-adic sets $\mathcal{D}$. We will show that calculating
$q^{\mleft(d\mright)}\mleft(n\mright)$ up to sub-exponential factors can be reduced to calculating $\rho^{\mleft(d\mright)}_{\min}\mleft(n\mright)$.
First, we prove an argument used in the reduction: 
\begin{lemma}
\label{lem:dadic_size}There are at most $n^{n}$ non-constant $d$-adic distributions over $X_n$.
\end{lemma}
\begin{proof}
Let $\mu$ be a non-constant $d$-adic distribution over $X_n$. We
assume that the minimal non-zero probability in $\mu$ is $d^{-l}$ and show
that $n>l$ by induction on $l$. This argument implies that for a
fixed $n$, the possible probabilities are only $0,d^{-1},d^{-2},\dots.,d^{-\mleft(n-1\mright)}$
and hence there are at most $n^{n}$ ways to construct a $d$-adic
distribution on $X_n$. For the base case $l=0$ it holds that $n>0$.
For the induction step, assume that the claim holds for $l-1$. Let
us first show that the number of elements with probability $d^{-l}$ is a multiple of $d$. Denote the set of these elements with $L$. Since the minimal non-zero probability in $X_n\backslash L$ is at least $d^{-l+1}$, the total weight of the elements in $X_n\backslash L$
can be written as $x\cdot d^{-l+1}$ where $x\in\mathbb{N}$, because
each element with probability $d^{-l+y}$ for some $y\geq 1$ simply contributes
$d^{y-1}$ to $x$, and $d^{y-1}$ is an integer. So, the following
must hold: 
\begin{gather*}
1=\sum_{i=1}^{n}\mu_{i}=\sum_{\mu_{i}:x_i\in L}\mu_{i}+\sum_{\mu_{i}:x_i\in X_n\backslash L}\mu_{i}=\mleft|L\mright|\cdot d^{-l}+x\cdot d^{-l+1}\\
\iff\\
\mleft|L\mright|\cdot d^{-l}=1-x\cdot d^{-l+1}\\
\iff\\
\mleft|L\mright|=d^{l}-x\cdot d=d\mleft(d^{l-1}-x\mright).
\end{gather*}
That is, $\mleft|L\mright|$ is a multiple of $d$, since $\mleft(d^{l-1}-x\mright)$
is an integer. Following that, we define a new distribution $\mu'$ on $X_{n'}$ by merging the elements in $L$ into $\mleft(d^{l-1}-x\mright)$ elements with probability $d^{-l+1}$. Now the minimal non-zero probability in $\mu'$ is $d^{-l+1}$ and since we have merged at least $d>1$ elements, it holds that $n'\leq n-1$. So, by the induction hypothesis we have $n-1\geq n'>l-1$, that is, $n>l$. 
\end{proof}
Now we can prove the reduction.
\begin{theorem}
\label{theo:reduction}Fix $n\in\mathbb{N}.$ Then:
\[
1/\rho^{\mleft(d\mright)}_{\min}\mleft(n\mright)\leq q^{\mleft(d\mright)}\mleft(n\mright)\leq2n^{2d}\ln n/\rho^{\mleft(d\mright)}_{\min}\mleft(n\mright).
\]
\end{theorem}
\begin{proof}
Recall that $q^{\mleft(d\mright)}\mleft(n\mright)$ is actually the size of a minimal
$d$-adic hitter for $X_n$, due to Lemma~\ref{lem:dadic_hit_reduc}. Hence we bound the size of
such a set, instead of $q^{\mleft(d\mright)}\mleft(n\mright)$ directly. Fix a
$d$-adic set $\mathcal{D}$ over $X_n$ with $\rho\mleft(\mathcal{D}\mright)=\rho^{\mleft(d\mright)}_{\min}\mleft(n\mright)$.
Fix $\overline{k}\in K$ and consider an arbitrary partition $S=\mleft(S_{i}\mright)_{i\in\mleft[d\mright]}$
of $X_n$ with $\mleft|S_{i}\mright|=k_{i}$ for any $i\in\mleft[d\mright]$.
Let $\sigma$ be a uniformly random permutation on $X_n$,
then: 
\[
\rho_{\overline{k}}\mleft(\mathcal{D}\mright)=\Pr\mleft[S\in\sigma\mleft(\mathcal{D}\mright)\mright].
\]
 Having that and the definition of $\rho^{\mleft(d\mright)}_{\min}\mleft(n\mright)$, it follows
that for any partition $S$ on $X_n$: 
\[
\Pr\mleft[S\in\sigma\mleft(\mathcal{D}\mright)\mright]\leq\rho^{\mleft(d\mright)}_{\min}\mleft(n\mright).
\]
 Let $\mathcal{Q}$ be a collection of partitions of $X_n$
with $\mleft|\mathcal{Q}\mright|<1/\rho^{\mleft(d\mright)}_{\min}\mleft(n\mright)$. Then
by the union bound: 
\[
\Pr\mleft[\mathcal{Q}\cap\sigma\mleft(\mathcal{D}\mright)\neq\emptyset\mright]\leq\sum_{Q\in\mathcal{Q}} \Pr\mleft[Q\in\sigma\mleft(\mathcal{D}\mright)\mright]<\frac{1}{\rho^{\mleft(d\mright)}_{\min}\mleft(n\mright)}\cdot\rho^{\mleft(d\mright)}_{\min}\mleft(n\mright)=1.
\]
 Thus, there is a permutation $\sigma$ such that $\mathcal{Q}\cap\sigma\mleft(\mathcal{D}\mright)=\emptyset$.
Since $\sigma\mleft(\mathcal{D}\mright)$ is a $d$-adic set, it follows that $\mathcal{Q}$ is not a $d$-adic hitter. So, indeed any $d$-adic
hitter must contain at least $1/\rho^{\mleft(d\mright)}_{\min}\mleft(n\mright)$ questions. 

Now we shall upper bound $q^{\mleft(d\mright)}\mleft(n\mright)$. Construct a set
$\mathcal{Q}$ of questions containing, for any $\overline{k}\in K,$
$\frac{1}{\rho^{\mleft(d\mright)}_{\min}\mleft(n\mright)}2n\ln n$ uniformly chosen
partitions $\mleft(S_{i}\mright)_{i\in\mleft[d\mright]}$ of $X_n$
with $\mleft|S_{i}\mright|=k_{i}$ for any $i\in\mleft[d\mright]$. Note
that $\mleft|K\mright|\leq\mleft(n+1\mright)^{d}$ and thus $\mleft|\mathcal{Q}\mright|\leq\frac{1}{\rho^{\mleft(d\mright)}_{\min}\mleft(n\mright)}2n^{2d}\ln n$.
We will show that with positive probability, $\mathcal{Q}$ is a $d$-adic
hitter. Fix an arbitrary $d$-adic set $\mathcal{D}$. Let $\overline{k}\in K$
such that $\rho_{\overline{k}}\mleft(\mathcal{D}\mright)=\rho\mleft(\mathcal{D}\mright).$
The probability that a random partition $\mleft(S_{i}\mright)_{i\in\mleft[d\mright]}$
of $X_n$ with $\mleft|S_{i}\mright|=k_{i}$ for all $i\in\mleft[d\mright]$
does not belong to $\mathcal{D}$ is at most 
\[
1-\rho_{\overline{k}}\mleft(\mathcal{D}\mright)=1-\rho\mleft(\mathcal{D}\mright)\leq1-\rho^{\mleft(d\mright)}_{\min}\mleft(n\mright)
\]
 (since $\rho\mleft(\mathcal{D}\mright)\geq\rho^{\mleft(d\mright)}_{\min}\mleft(n\mright)$).
Therefore the probability that $\mathcal{Q}$ is disjoint from $\mathcal{D}$
is at most 
\[
\mleft(1-\rho^{\mleft(d\mright)}_{\min}\mleft(n\mright)\mright)^{\frac{1}{\rho^{\mleft(d\mright)}_{\min}\mleft(n\mright)}2n\ln n}\leq e^{-\rho^{\mleft(d\mright)}_{\min}\mleft(n\mright)\frac{1}{\rho^{\mleft(d\mright)}_{\min}\mleft(n\mright)}2n\ln n}=n^{-2n}.
\]
 By Lemma \ref{lem:dadic_size}, there are fewer than $n^{2n}$ $d$-adic
distributions over $X_n$. Having that, a union bound shows
that the probability that a $d$-adic set $\mathcal{D}$ (corresponding
to some $d$-adic distribution $\mu$) which is
disjoint from $\mathcal{Q}$ exists is less than $1$. That is, the
probability that $\mathcal{Q}$ is a $d$-adic hitter is positive. 
\end{proof}
Due to this theorem, if $d=o\mleft(n/\mleft(\log n \log \log n\mright)\mright)$, we have: 
\[
q^{\mleft(d\mright)}\mleft(n\mright)=2^{\pm o\mleft(n\mright)}\cdot\frac{1}{\rho^{\mleft(d\mright)}_{\min}\mleft(n\mright)}.
\]
 Hence, from now on we discuss $\rho^{\mleft(d\mright)}_{\min}\mleft(n\mright)$ instead
of $q^{\mleft(d\mright)}\mleft(n\mright)$, and restrict the discussion
to $d=o\mleft(n/\mleft(\log n \log \log n\mright)\mright)$. 

Before we discuss some bounds on $\rho^{\mleft(d\mright)}_{\min}\mleft(n\mright)$, let us define
the \emph{generalized tail} of a $d$-adic distribution:
\begin{definition}
Let $\mu$ be a $d$-adic distribution over $X_n$. The
\emph{generalized tail} of $\mu$ is the largest set $T\subset X_n$
such that for some $a\geq1$:
\begin{enumerate}
\item $\mu\mleft(T\mright)=d^{-a}$.
\item $T$ does not contain zero-probability elements.
\item All elements in $X_n\backslash T$ have probability at
least $d^{-a}$ or zero. 
\end{enumerate}
If there are a few sets satisfying those requirements, the generalized
tail is one of them, arbitrarily. 
\end{definition}
\begin{lemma}
Suppose that $\mu$ is a non-constant $d$-adic distribution. Let
$D=\mleft(D_{j}\mright)_{j\in\mleft[d\mright]}\in \Div\mleft(\mu\mright)$ be a partition of $X_n$. For all $j\in\mleft[d\mright]$, $D_{j}$
either contains $T$ or disjoint from $T$.
\end{lemma}
\begin{proof}
Let $j\in\mleft[d\mright]$.
If $D_{j}$ is disjoint from $\mu$ then we are done. So, Assume that
$D_{j}\cap T\neq\emptyset$. Since all non-zero elements in $X_n\backslash T$
have probability at least $d^{-a}$, we can denote $\mu\mleft(D_{j}\cap\mleft(X_n\backslash T\mright)\mright)=s\cdot d^{-a}$
where $s\in\mathbb{N}$. Recall that $\mu\mleft(D_{j}\mright)=1/d$,
so if we denote $\mu\mleft(D_{j}\cap T\mright)=y$ we can write: 
\[
1/d=s\cdot d^{-a}+y.
\]
 Now, note that $s\leq d^{a-1}-1$: recall that $D_{j}\cap T\neq\emptyset$
and thus $y>0$. Assume towards contradiction that $s>d^{a-1}-1$,
that is, $s\geq d^{a-1}$. Then: 
\[
1/d=s/d^{a}+y\geq d^{a-1}/d^{a}+y=1/d+y>1/d
\]
which is a contradiction. Having that, we lower bound $y$:
\[
y=1/d-s/d^{a}\geq1/d-\frac{d^{a-1}-1}{d^{a}}=1/d-1/d+1/d^{a}=\mu\mleft(T\mright),
\]
 and hence $\mu\mleft(D_{j}\cap T\mright)=\mu\mleft(T\mright)$, that
is, $D_{j}\cap T=T$, and so $D_{j}$ contains $T$ completely. 
\end{proof}
In the following sections we prove upper and lower bounds on $\rho^{\mleft(d\mright)}_{\min}\mleft(n\mright)$.
The following function $f_{d}:\mleft(0,1\mright)\rightarrow\mathbb{R}$,
defined for any $d\in\mathbb{N}$, will appear in both of our bounds:
\[
f_{d}\mleft(\beta\mright)=d\cdot\beta\cdot\log_{2}d-\mleft(\mleft(1-\mleft(d-1\mright)\beta\mright)\log_{2}\frac{1}{1-\mleft(d-1\mright)\beta}+\mleft(d-1\mright)\beta\log_{2}\frac{1}{\beta}\mright).
\]

\subsection{Upper bounding $\rho^{\mleft(d\mright)}_{\min}\mleft(n\mright)$}

The following lemma implies different upper bounds on $\rho^{\mleft(d\mright)}_{\min}\mleft(n\mright)$
for different sequences of $n$ values. 
\begin{lemma}
\label{lem:dadic_upper_bound}Fix $d\in\mathbb{N}$ and $\frac{1}{d^2+1}\leq\beta\leq1/d$.
For any $n$ of the form $n=\mleft\lfloor \frac{d^{a}}{d\cdot\beta}\mright\rfloor $,
where $a\in\mathbb{N}$, there exists a $d$-adic distribution $\mu$
over $X_n$ which satisfies 
\[
\rho\mleft(\Div\mleft(\mu\mright)\mright)\leq2^{f_{d}\mleft(\beta\mright)n+o\mleft(n\mright)}.
\]
 
\end{lemma}
\begin{proof}
We first assume that $\mleft\lfloor \frac{d^{a}}{d\cdot\beta}\mright\rfloor =\frac{d^{a}}{d\cdot\beta}$, and that $\frac{d^{a}}{d\cdot\beta} \equiv 1 \pmod {d-1}$.
Let $n=\frac{d^{a}}{d\cdot\beta}$ where $a\in\mathbb{N}$. Note that
$\beta n=d^{a-1}$, and construct the following $d$-adic distribution
$\mu$ on $X_n$:
\begin{enumerate}
\item For $i\in\mleft[d\cdot\beta n-1\mright]$: $\mu_{i}=d^{-a}=\frac{1}{d\cdot d^{a-1}}=\frac{1}{d\cdot\beta n}$.
\item All other $\mleft(1-d\beta\mright)n+1$ elements are a (generalized)
tail of probability $d^{-a}$ (this is possible since $n \equiv 1 \pmod {d-1}$).
\end{enumerate}
As we have shown, the generalized tail elements must be chosen to the same set in a partition,
in order to get a partition which divides $\mu$, thus we can think
of them as a single element when constructing a partition in $\Div\mleft(\mu\mright)$,
such that we have $d\cdot\beta n$ elements in total, with equal probabilities.
Thus, there is only one feasible type $\overline{k}$ of partition:
choosing $\beta n=d^{a-1}$ elements to each set in the partition
(that is, $k_{i}=\beta n$ for any $i\in\mleft[d\mright]$, assuming
that the tail is treated as a single element). The total
probability of each set in the partition is thus $d^{a-1}\cdot d^{-a}=1/d$. This discussion
leads us to the following bound: 
\begin{align*}
\rho\mleft(\Div\mleft(\mu\mright)\mright) & =\rho_{\overline{k}}\mleft(\Div\mleft(\mu\mright)\mright)\\
 & =\frac{\frac{\mleft(d\beta n\mright)!}{\mleft(\beta n\mright)!^{d}}}{\frac{n!}{\mleft(\beta n\mright)!^{d-1}\mleft(\mleft(1-\mleft(d-1\mright)\beta\mright)n\mright)!}}\\
 & \underset{\text{\eqref{eq:multi_bounds}}}{\leq}\frac{2^{d\cdot\beta n\cdot\log_{2}d}}{2^{n\mleft(\mleft(1-\mleft(d-1\mright)\beta\mright)\log_{2}\frac{1}{1-\mleft(d-1\mright)\beta}+\sum_{i=1}^{d-1}\beta\log_{2}\frac{1}{\beta}\mright)}/O\mleft(n^{d}\mright)}\\
 & =O\mleft(n^{d}\mright)\cdot2^{\mleft[d\cdot\beta\cdot\log_{2}d-\mleft(\mleft(1-\mleft(d-1\mright)\beta\mright)\log_{2}\frac{1}{1-\mleft(d-1\mright)\beta}+\mleft(d-1\mright)\beta\log_{2}\frac{1}{\beta}\mright)\mright]n}=2^{f_{d}\mleft(\beta\mright)n+o\mleft(n\mright)}.
\end{align*}

We now assume that  $\frac{d^{a}}{d\cdot\beta}$ is not necessarily natural, and not necessarily equal to $1$ modulo $d-1$. We choose $\beta' \geq \beta$ such  that $n' =  \frac{d^{a}}{d\cdot\beta'} \in \mathbb{N}$, where $n' \equiv 1 \pmod {d-1}$, and such that  $n - n'$ is minimal across all $\beta'$ satisfying these requirements. Notice that there exists such $\beta'$ that satisfies $n - n' \leq d-1$. Therefore, the  definition of $n$ as $n = \mleft\lfloor \frac{d^{a}}{d\cdot\beta}\mright\rfloor$, where $\beta \geq \frac{1}{d^2 + 1}$, implies that $\frac{1}{\beta'} \in [1/\beta -  1/d^{a-3}, 1/\beta]$. Calculating the series expansion of $\frac{1}{1/\beta -  1/d^{a-3}}$ implies that $\beta' = \beta + O(1/d^{a-3})$. Since $\beta \geq \frac{1}{d^2 + 1}$, it follows that $\beta' = \beta + O(1/n)$. We apply the arguments above for $n'$ instead of $n$ (the spare elements are simply given probability $0$), deducing
\[
\rho\mleft(\Div\mleft(\mu\mright)\mright)
\leq
2^{f_{d}\mleft(\beta'\mright)n'+o\mleft(n'\mright)}
\leq
2^{f_{d}\mleft(\beta +  O(1/n)\mright)n+o\mleft(n\mright)}
=
2^{f_{d}\mleft(\beta \mright)n+o\mleft(n\mright)},
\]
as desired. The second inequality follows from the fact that $f_d$ is continuous.
%\textcolor{red}{Previous solution is below: Delete it afterwards.}
%Now, assume that $\mleft\lfloor \frac{d^{a}}{d\cdot\beta}\mright\rfloor <\frac{d^{a}}{d\cdot\beta}$.
%In that case, let $\beta'$ such that 
%\[
%\mleft\lfloor \frac{d^{a}}{d\cdot\beta}\mright\rfloor =\frac{d^{a}}{d\cdot\beta'}.
%\]
%Now, construct the aforementioned $d$-adic distribution $\mu$ for
%$\beta'$ instead of $\beta$. From previous arguments, we have: 
%\[
%\rho\mleft(\Div\mleft(\mu\mright)\mright)\leq 2^{f_{d}\mleft(\beta'\mright)n+o\mleft(n\mright).}
%\]
%Fortunately, this is enough: by the definition of $\beta'$ and the
%constraint $\beta\leq1/d$, it holds that 
%\[
%\beta\leq\beta'\leq\beta+\frac{1}{d\mleft(d^{a}-1\mright)}.
%\]
%Recall that $\beta \geq \frac{1}{d^2+1}$. This implies $d\mleft(d^{a}-1\mright)=\Theta\mleft(n\mright)$, that is 
%\[
%\beta\leq\beta'\leq\beta+\Theta\mleft(1/n\mright).
%\]
%Therefore, it holds that $f_{d}\mleft(\beta'\mright)n\leq f_{d}\mleft(\beta\mright)n+O\mleft(1\mright)$,
%and hence the lemma holds also for the case $\mleft\lfloor \frac{d^{a}}{d\cdot\beta}\mright\rfloor %<\frac{d^{a}}{d\cdot\beta}$. 
\end{proof}
For $n$ of the form $n=\mleft\lfloor \frac{d^{a}}{d\cdot\beta}\mright\rfloor $,
obviously $\rho^{\mleft(d\mright)}_{\min}\mleft(n\mright)\leq\rho\mleft(\Div\mleft(\mu\mright)\mright)$,
where $\mu$ is the distribution defined in the proof of Lemma \ref{lem:dadic_upper_bound}.
Hence for such $n$ values we have
\[
\rho^{\mleft(d\mright)}_{\min}\mleft(n\mright)\leq2^{f_{d}\mleft(\beta\mright)n+o\mleft(n\mright)}.
\]

\subsection{Lower bounding $\rho^{\mleft(d\mright)}_{\min}\mleft(n\mright)$ }

We will use the following partition of $X_n$ in order
to lower bound $\rho\mleft(\Div\mleft(\mu\mright)\mright)$ for some non-constant
$d$-adic distribution $\mu$: 
\begin{lemma} \label{lem:partition} 
Let $\mu$ be a non-constant $d$-adic distribution
over $X_n$. There exists a partition of $X_n$
of the form 
\[
X_n=\bigcup_{i=1}^{\gamma}\mleft(D_{i}\cup E_{i}\mright)
\]
 such that: 
\begin{enumerate}
\item $D_{i}$ consists of elements with equal probabilities $p_{i}$.
\item $\mleft|D_{i}\mright|=d c_i-r_i$ for some natural
$c_{i}$ and $0\leq r_i<d$, and $\mu\mleft(E_{i}\mright)=r_i p_{i}$.
\item $\gamma=O\mleft(\log n\mright)$.
\end{enumerate}
\end{lemma}
\begin{proof}
We assume w.l.o.g that the elements are sorted $\mu_1\geq\mu_2\geq\dots \geq\mu_n$. We construct the sets $D_{i},E_{i}$ in iterations on $i$ from $1$
to $\gamma$. In each iteration $i$, assume we have ordered probabilities
$\mu_{\alpha_{i}}\geq\mu_{2}\geq\dots\geq\mu_{N_{i}}$ of the available
elements which were not chosen in previous iterations to $D_{i},E_{i}$
(initially, $\alpha_{1}=1,N_{1}=n$). The elements chosen for $D_{i}$
are always an interval which begins in the leftmost index $\alpha_{i}$
and up to some index $\beta_{i}$. The elements in $E_{i}$ (if it
is not empty) are always an interval which begins in some index $M_{i}>\beta_{i}$
and ends at $N_{i}$. The rest of the elements are available for the
next iteration, until no elements are available and the partition
is complete. The partition must be completed since $D_{i}\neq\emptyset$
for any $i$. Now let us describe an iteration $i$ in detail. Let
$\beta_{i}$ be the last index with $\mu_{\beta_{i}}=\mu_{\alpha_{i}}$
(that is, $\mu_{\beta_{i}+1}<\mu_{\alpha_{i}}$ or $\beta_{i}=N_{i}$).
Let $D_{i}=\mleft\{ x_{\alpha_{i}},\dots,x_{\beta_{i}}\mright\} $. Denote $\mleft|D_{i}\mright|=d c_i-r_i$ where $c_i,r_i\in \mathbb{N}$ and $0\leq r_i<d$. Let $M_{i}>\beta_i$ be an index such that $\sum_{j=M_{i}}^{N_{i}}\mu_{j}=r_i p_i$ if $r_i>0$, and $M_i=\infty$ otherwise. Define $E_{i}=\mleft\{ x_{M_{i}},\dots,x_{N_{i}}\mright\} $ (if $M_i=\infty$, then $E_i=\emptyset$).
We show inductively that for any $i$, $\sum_{j=\alpha_{i}}^{N_{i}}\mu_{j}$ is a multiple
of $d\cdot\mu_{\alpha_{i}}$, and $M_i$ exists. For the base case $\alpha_{1}=1$
and $N_{1}=n$, note that $\sum_{j=1}^{n}\mu_{i}=1$ which is a multiple
of $d\cdot\mu_{1}$ since $\mu$ is non-constant. The existence of the index $M_1$ now follows from Lemma~\ref{lem:useful}: Suppose that $\sum_{j=\alpha_{1}}^{N_{1}}\mu_{j}=k\cdot d p_1$, so by Lemma~\ref{lem:useful} we can partition $\mleft\{x_{\alpha_1},\dots,x_{N_1}\mright\}$ to $k$ intervals, each of probability $p_1$. So, $M_1$ is simply the first index of the interval composed from the concatenation of the last $r_1$ intervals, if $r_1>0$. For the induction
step, assume that for iteration $i-1$, $\sum_{j=\alpha_{i-1}}^{N_{i-1}}\mu_{j}$
is a multiple of $d\cdot\mu_{\alpha_{i-1}}$ and that $M_{i-1}$ exists. By assumption, $\sum_{j=\alpha_{i-1}}^{\beta_{i-1}}\mu_{j}+\sum_{j=M_{i-1}}^{N_{i-1}}\mu_{j}=d\cdot\mu_{\alpha_{i-1}}\cdot c_{i-1}$
for some integer $c_{i-1}$. When continuing to iteration $i$, we
are removing $D_{i-1}\cup E_{i-1}$ from the available elements, and
recall that $\sum_{j=\alpha_{i-1}}^{N_{i-1}}\mu_{j}$ is also a multiple
of $d\cdot\mu_{\alpha_{i-1}}$ by assumption, and thus we still have
a multiple of $d\cdot\mu_{\alpha_{i-1}}$ in the available elements
of iteration $i$ (that is, after removing $D_{i-1}\cup E_{i-1}$).
Since $\mu_{\alpha_{i-1}}$ is a multiple of $\mu_{\alpha_{i}}$,
we also have a multiple of $d\cdot\mu_{\alpha_{i}}$. The existence of the index $M_i$ now follows from Lemma~\ref{lem:useful} similarly as in the base case. \par
It remains to show that $\gamma=O\mleft(\log n\mright)$. Let us consider
the first iteration. If the case is that $\mleft|D_{1}\mright|$ is
a multiple of $d\cdot\mu_{1}$, we change the partition a bit, and
leave the last element of $D_{1}$ out, and therefore use a non-empty
$E_{1}$. Now, it must hold that $\mu\mleft(E_{1}\mright)\geq\mu_{1}$.
Since the probabilities are ordered $\mu_{1}\geq\dots\geq\mu_{n}$
we have 
\[
n\cdot\mu_{M_{1}-1}\geq n\cdot\mu_{M_{1}}\geq\mu\mleft(E_{1}\mright)\geq\mu_{1},
\]
 that is, $\mu_{M_{1}-1}\geq\mu_{1}/n$. Since the probabilities are
$d$-adic, there are at most $\log n+1$ different probabilities in
$\mu_{2},\dots,\mu_{M_{1}-1}$: 
\[
\mu_{1}/d^{0},\mu_{1}/d^{1},\mu_{1}/d^{2},\dots,\mu_{1}/d^{\log n}
\]
and therefore $\gamma=O\mleft(\log n\mright)$. 
\end{proof}

Now we can prove the main lemma:
\begin{lemma}
\label{lem:dadic_low_bound}If $d=o\mleft(n/\log^{2}n\mright)$, then for
every non-constant $d$-adic distribution $\mu$ there is $0<\beta<1$
such that 
\[
\rho\mleft(\Div\mleft(\mu\mright)\mright)\geq2^{f_{d}\mleft(\beta\mright)n-o\mleft(n\mright)}.
\]

\end{lemma}
\begin{proof}
We will use a partition of $X_n$ of the form 
\[
X_n=\bigcup_{i=1}^{\gamma}\mleft(D_{i}\cup E_{i}\mright)
\]
 as constructed in Lemma~\ref{lem:partition}. It is implied from Lemma~\ref{lem:partition} that $\mu\mleft(D_{i}\cup E_{i}\mright)=d\cdot c_{i}\cdot p_{i}$
for some $c_{i}\in\mathbb{N}$. If $E_i\neq \emptyset$, we consider a partition of $E_{i}$ into
$r_{i}$ subsets, each with total probability $p_{i}$. Indeed, such a
partition exists by Lemma~\ref{lem:useful}. Denote by $E'_{i}$ the
set of those subsets, where each subset is contracted into a single
element. So, $E'_{i}$ is a set of $r_{i}$ elements with probability
$p_{i}$ each, and thus in $D_{i}\cup E'_{i}$ we have $d\cdot c_{i}$
elements, each with probability $p_{i}$. Denote $X_n'=\bigcup_{i=1}^{\gamma}\mleft(D_{i}\cup E'_{i}\mright)$. 

Let us define a form of partition of $X_n'$ into $d$
subsets with equal total probabilities: for any $i\in\mleft[\gamma\mright]$,
let the sets $S_{i}\mleft(1\mright),S_{i}\mleft(2\mright),\dots,S_{i}\mleft(d\mright)$
be $d$ distinct subsets of $D_{i}\cup E'_{i}$ of size $c_{i}$ each.
For any $j\in\mleft[d\mright]$, define $S_{j}$ by 
\[
S_{j}:=\bigcup_{i=1}^{\gamma}S_{i}\mleft(j\mright).
\]
Indeed $S=\mleft(S_{j}\mright)_{j\in\mleft[d\mright]}$ exists in $\Div\mleft(\mu\mright)$
after ``unpacking'' all elements in $\bigcup_{i\in\mleft[\gamma\mright]}E'_{i}$
back to their original state: for any $j\in\mleft[d\mright]$ we have
\[
\mu\mleft(S_{j}\mright)=\sum_{i=1}^{\gamma}\mu\mleft(S_{i}\mleft(j\mright)\mright)=\sum_{i=1}^{\gamma}c_{i}p_{i}=\frac{1}{d}\sum_{i=1}^{\gamma}\mu\mleft(D_{i}\cup E_{i}\mright)=1/d.
\]
 So, any partition $S=\mleft(S_{j}\mright)_{j\in\mleft[d\mright]}$ defined
in this fashion exists in $\Div\mleft(\mu\mright)$. Having that, consider
the following type $\overline{k}$ of partitions which includes at
least some of those partitions: let $k_{j}=\sum_{i=1}^{\gamma}c_{i}$
for any $1\leq j\leq d-1$ and $k_{d}=n-\mleft(d-1\mright)k_{1}$ ($k_{d}$
can be thought as the size of the set that ``contains the tail'',
as discussed in the upper bound). $\Div\mleft(\mu\mright)_{\overline{k}}$
contains at least the partitions in which $S_{1},\dots,S_{d-1}$ contain
only elements from $\bigcup_{i=1}^{\gamma}D_{i}$. So, for any $i\in \mleft[\gamma\mright]$ we choose $c_i$ elements from $D_i$ to $S_j$, for $1\leq j\leq d-1$. Moreover, we put all the elements of $\bigcup_{i=1}^{\gamma}E_{i}$ in $S_d$. Thus:

\[
\mleft|\Div\mleft(\mu\mright)_{\overline{k}}\mright| \geq \prod_{i=1}^{\gamma}\frac{\mleft(d\cdot c_{i} -r_i\mright)!}{\mleft(c_{i}!\mright)^{d-1}\mleft(c_i-r_i\mright)!}.
\]
Hence:

\begin{align*}
\mleft|\Div\mleft(\mu\mright)_{\overline{k}}\mright| &\geq  \prod_{i=1}^{\gamma}\frac{\mleft(d\cdot c_{i} -d\mright)!}{\mleft(c_{i}!\mright)^{d}}\\
&\geq \prod_{i=1}^{\gamma} \frac{1}{\mleft(d c_i\mright)^d} \frac{\mleft(d\cdot c_{i}\mright)!}{\mleft(c_{i}!\mright)^{d}}\\
&\geq \frac{1}{n^{d\gamma}}\prod_{i=1}^{\gamma} \frac{\mleft(d\cdot c_{i}\mright)!}{\mleft(c_{i}!\mright)^{d}}\\ 
&\underset{\text{\eqref{eq:multi_bounds}}}{\geq} \frac{1}{O\mleft(n^{2d\gamma}\mright)}\prod_{i=1}^{\gamma}2^{c_{i}\cdot d\log d}\\
&\geq \frac{1}{2^{2\log_{2}n\cdot o\mleft(\frac{n}{\log^{2}n}\mright)\cdot O\mleft(\log n\mright)}}\prod_{i=1}^{\gamma}2^{c_{i}\cdot d\log d}\geq2^{d\log_{2}d\sum_{i=1}^{\gamma}c_{i}-o\mleft(n\mright)}.
\end{align*}

Now, denote $\beta=\frac{1}{n}\sum_{i=1}^{\gamma}c_{i}$ and note
that 
\[
\binom{n}{k_1, k_2,\dots, k_d}\leq2^{\mleft[\mleft(d-1\mright)\beta\log_{2}\frac{1}{\beta}+\mleft(1-\mleft(d-1\mright)\beta\mright)\log_{2}\frac{1}{\mleft(1-\mleft(d-1\mright)\beta\mright)}\mright]n}
\]
(similarly to the discussion in the upper bound section). Therefore
overall 
\[
\rho_{\overline{k}}\mleft(\Div\mleft(\mu\mright)\mright)\geq2^{\mleft[\beta d\log_{2}d-\mleft(\mleft(d-1\mright)\beta\log_{2}\frac{1}{\beta}+\mleft(1-\mleft(d-1\mright)\beta\mright)\log_{2}\frac{1}{\mleft(1-\mleft(d-1\mright)\beta\mright)}\mright)\mright]n-o\mleft(n\mright)}=2^{f_{d}\mleft(\beta\mright)n-o\mleft(n\mright)},
\]
and since obviously $\rho\mleft(\Div\mleft(\mu\mright)\mright)\geq\rho_{\overline{k}}\mleft(\Div\mleft(\mu\mright)\mright)$,
we get the desired result.
\end{proof}

\subsection{Estimating $q^{\mleft(d\mright)}\mleft(n\mright)$}

Now we can deduce some explicit bounds on $q^{\mleft(d\mright)}\mleft(n\mright)$. Those bounds allow us to calculate $q^{\mleft(d\mright)}\mleft(n\mright)$
up to sub-exponential factors, for infinitely many $n$ values. The
upper bound on $q^{\mleft(d\mright)}\mleft(n\mright)$ will imply that even
though the trivial upper bound on the cardinality of $\mathcal{Q}$
which allows constructing optimal strategies for all distributions
is $d^{n}/d!$, the true minimal cardinality is much smaller, and
in particular it is less than $2^{n+o\mleft(n\mright)}$. 
\begin{theorem}
For any $n$ and any $d=o\mleft(n/\log^{2}n\mright)$:
\[
q^{\mleft(d\mright)}\mleft(n\mright)\leq\mleft(1+\frac{d-1}{d^{\frac{d}{d-1}}}\mright)^{n+o\mleft(n\mright)}=\mleft(2-\Theta\mleft(\frac{\log d}{d}\mright)\mright)^{n+o\mleft(n\mright)}.
\]
 Moreover, for any fixed $d$, the following holds for infinitely
many $n$ values: 
\[
q^{\mleft(d\mright)}\mleft(n\mright)=\mleft(1+\frac{d-1}{d^{\frac{d}{d-1}}}\mright)^{n\pm o\mleft(n\mright)}=\mleft(2-\Theta\mleft(\frac{\log d}{d}\mright)\mright)^{n\pm o\mleft(n\mright)}.
\]
\end{theorem}
\begin{proof}
Since Lemma \ref{lem:dadic_low_bound} holds for any $d$-adic distribution
$\mu$ where $d=o\mleft(n/\log^{2}n\mright)$, we can deduce the lower
bound
\begin{gather*}
\rho^{\mleft(d\mright)}_{\min}\mleft(n\mright)\geq\exp_{2}\mleft(\mleft(\min_{0<\beta<1}f_{d}\mleft(\beta\mright)\mright)n-o\mleft(n\mright)\mright).
\end{gather*}
 Calculation shows that 
\[
f_{d}'\mleft(\beta\mright)=\mleft(d-1\mright)\log_{2}\frac{\beta}{\mleft(1-\mleft(d-1\mright)\beta\mright)}+d\log_{2}d
\]
 and the minimum is attained at
\[
\beta=\frac{1}{d^{\frac{d}{d-1}}-1+d}.
\]
We are interested in what happens when $\beta$ minimizes $f_{d}$.
So, we want to estimate the following function of $d$: 
\begin{align*}
f\mleft(d\mright) & =f_{d}\mleft(\frac{1}{d^{\frac{d}{d-1}}-1+d}\mright).
\end{align*}
 After some algebraic simplifications, we get: 
\[
f\mleft(d\mright)=\log_{2}\mleft(\frac{d^{\frac{d}{d-1}}}{d^{\frac{d}{d-1}}+d-1}\mright).
\]
Since $d=o\mleft(n/\log n \log \log n\mright)$, the reduction in Theorem \ref{theo:reduction}
allows us to calculate $\exp_{2}\mleft(-f\mleft(d\mright)\mright)$ in
order to get an estimate of $q^{\mleft(d\mright)}\mleft(n\mright)$: we
have 
\[
\exp_{2}\mleft(-f\mleft(d\mright)\mright)=\frac{d^{\frac{d}{d-1}}+d-1}{d^{\frac{d}{d-1}}}=1+\frac{d-1}{d^{\frac{d}{d-1}}},
\]
 which implies 
\[
q^{\mleft(d\mright)}\mleft(n\mright)\leq\mleft(1+\frac{d-1}{d^{\frac{d}{d-1}}}\mright)^{n+o\mleft(n\mright)}.
\]
Moreover, it holds that: 
\[
\frac{d-1}{d^{\frac{d}{d-1}}}=\Theta\mleft(d^{1-\frac{d}{d-1}}\mright)=\Theta\mleft(d^{-\frac{1}{d-1}}\mright).
\]
Calculating the Puiseux expansion of $d^{-\frac{1}{d-1}}$ shows that
$d^{-\frac{1}{d-1}}=1-\Theta\mleft(\frac{\log d}{d}\mright)$ and hence:
\[
q^{\mleft(d\mright)}\mleft(n\mright)\leq\exp_{2}\mleft(-f\mleft(d\mright)n +o\mleft(n\mright)\mright)=\mleft(2-\Theta\mleft(\frac{\log d}{d}\mright)\mright)^{n+o\mleft(n\mright)}.
\]
 For the second part of the theorem, assume that $d$ is fixed, let
$\beta=\frac{1}{d^{\frac{d}{d-1}}-1+d}$ and suppose that $n=\mleft\lfloor \frac{d^{a}}{d\cdot\beta}\mright\rfloor $
where $a\in\mathbb{N}$. Note that $\frac{1}{d^2+1}\leq \frac{1}{d^{\frac{d}{d-1}}-1+d}\leq 1/d$, so we can use Lemma~\ref{lem:dadic_upper_bound} and deduce
\[
\rho^{\mleft(d\mright)}_{\min}\mleft(n\mright)\leq\exp_{2}\mleft(f_{d}\mleft(\beta\mright)n+o\mleft(n\mright)\mright)=\exp_{2}\mleft(f\mleft(d\mright)n+o\mleft(n\mright)\mright)=\mleft(\frac{d^{\frac{d}{d-1}}}{d^{\frac{d}{d-1}}+d-1}\mright)^{n+o\mleft(n\mright)},
\]
 and hence 
\[
q^{\mleft(d\mright)}\mleft(n\mright)\geq\mleft(1+\frac{d-1}{d^{\frac{d}{d-1}}}\mright)^{n-o\mleft(n\mright)}=\mleft(2-\Theta\mleft(\frac{\log d}{d}\mright)\mright)^{n-o\mleft(n\mright)}.\qedhere
\]
\end{proof}

\section{Open questions} \label{sec:Open-questions}

Our work suggests a few open questions which we think are interesting
enough for future research.
\begin{question}
Is $G$ continuous?
\end{question}

\begin{description}
\item [{Notes}] It seems that techniques similar to those used in Section
\ref{sec:Approximating} can show continuity-related results, but
additional work seems necessary in order to determine whether $G$
is continuous. First, it seems not hard to show that $G$ is upper
semi-continuous. Moreover, denote by $G_{b}$ the function $G$ restricted
to some fixed $b$, such that $G\mleft(\beta\mright)=\inf_{b\in\mathbb{N}}G_{b}\mleft(\beta\mright)$.
It also seems not hard to show that $G_{b}$ is continuous. We should
use a fixed $b$ since otherwise Lemma \ref{lem:s_index} is not helpful.
It is not clear, however, whether $G$ is continuous as well. If we
could show that $G$ is lower semi-continuous, or that $b$ can be
chosen over some compact subset of $\mathbb{N}$ instead of the entirety
of $\mathbb{N}$, then continuity of $G$ would follow. 
\end{description}
\begin{question}
Is the outer infimum in $G$ attained?
\end{question}

\begin{description}
\item [{Notes}] We have shown that the inner supremum in the definition
of $G$ is attained and thus can be written as maximum. It is not
clear, however, whether the outer infimum is attained as well. Unfortunately,
even if we assume that $b$ is fixed, we still can not apply a similar
argument to the one used in the supremum case: Say we have a fixed
$b\in\mathbb{N}$ and a sequence of sequences $\mleft(\ser c^{j}\mright)_{j\in\mathbb{N}}\in\mathcal{C}$
converging to the infimum, and converging pointwise to a sequence
$\ser c$. It does not guarantee (not immediately, at least) that
$\mleft(\ser c^{j}\mright)_{j\in\mathbb{N}}$ converges to $\ser c$
in $\ell_{1}$-norm, and that property is crucial for $\ser c$ being
a minimizing sequence for $\max_{\ser{\alpha}\in\mathcal{A}}P\mleft(\ser c,\ser{\alpha}\mright)$
across all sequences in $\mathcal{C}$. 
\end{description}
\begin{question}
Can we calculate $G\mleft(\beta\mright)$?
\end{question}

\begin{description}
\item [{Notes}] While our formula for $G$ implies $G\mleft(\beta\mright)\leq-0.305758$
for any $1\leq\beta<2$, it would be interesting to calculate $G\mleft(\beta\mright)$
in terms of $\beta$, similarly to the calculation suggested in \cite{DFGM2019}
for $\beta=1.25$, that is $G\mleft(1.25\mright)=-\log1.25$. This will
allow us to calculate $q\mleft(n\mright)$ for $n$ of
the form $n=\beta2^{k}$, up to sub-exponential factors. 
\end{description}

\begin{question}
Can we generalize the function $G\colon \mleft[1,2\mright)\rightarrow \mathbb{R}$ to a function $G^{\mleft(d\mright)}\colon \mleft[1,d\mright)\rightarrow \mathbb{R}$ such that $\rho^{\mleft(d\mright)}_{\min}\mleft(n\mright)=2^{G^{\mleft(d\mright)}\mleft(\beta\mright)n\pm o\mleft(n\mright)}$?
\end{question}

\bibliographystyle{alpha}
\bibliography{bib.bib}

\end{document}